\documentclass[12pt]{article}
\usepackage{amsmath}
\usepackage{graphicx,psfrag,epsf}
\usepackage{enumerate}
\usepackage{url} 

 \usepackage{xr}

\usepackage{amsfonts,amssymb,amsmath,amsthm,epsfig,euscript,bbm,amsthm}
\usepackage{amsthm}

\usepackage {tikz}
\usetikzlibrary {positioning}
\usetikzlibrary{arrows}
\usepackage {graphicx}
\usepackage{enumerate}
\pdfoutput=1

\usepackage{subcaption}
\usepackage{caption}
\usepackage[authoryear,sort]{natbib}
\usepackage[titletoc,title]{appendix}
\usepackage[pdftex,colorlinks]{hyperref}
\usepackage{mathtools}
 \usepackage{xcolor}

\usepackage{graphicx}
\usepackage{subcaption}
 
 \usepackage{xr}
\usepackage{algorithm}
\usepackage{algorithmic}
\usepackage{setspace}
 \definecolor{orange}{RGB}{230,170,120}
  \definecolor{green}{RGB}{120,200,120}

 \hypersetup{
 	colorlinks=true,
 	linkcolor=blue,
 	filecolor=blue,      
 	urlcolor=cyan,
 }
 \hypersetup{colorlinks,linkcolor={blue},citecolor={blue},urlcolor={red}} 
 \usepackage{geometry}                
 \usepackage{multirow}
\usepackage[section]{placeins}
 \usepackage{bbm}
 \usepackage{graphicx}
 \usepackage{amssymb}
\newcommand{\ra}[1]{\renewcommand{\arraystretch}{#1}} 
\def\addlegendimage{\csname pgfplots@addlegendimage\endcsname}

\pgfkeys{/pgfplots/number in legend/.style={%
        /pgfplots/legend image code/.code={%
            \node at (0.295,-0.0225){#1};
        },%
    },
}
\usepackage{amsfonts,amssymb,amsmath,amsthm,epsfig,euscript,bbm,amsthm}
\usepackage{algorithm}
\usepackage{amsthm}
\usepackage{threeparttable,subcaption}
 \usepackage{booktabs}
 
\usepackage{rotating,multirow,makecell,array}
\usepackage {tikz}
\usetikzlibrary {positioning}
\definecolor {processblue}{cmyk}{0.96,0,0,0}
\usepackage{colortbl}

\usepackage{tikz}
\usetikzlibrary{arrows}
\usepackage {graphicx}

 \usepackage{epstopdf}
 \usepackage{tikz} 
 \usepackage{amsmath} 
\usepackage[utf8]{inputenc}
\usepackage{pgfplots}
 \theoremstyle{plain}
 \newtheorem{thm}{Theorem}[section]
 \newtheorem{lem}[thm]{Lemma}
 \newtheorem{prop}[thm]{Proposition}
 \newtheorem{cor}{Corollary}
 \usepackage{caption}
\usepackage{subcaption}
 \theoremstyle{definition}
 \newtheorem{defn}{Definition}[section]
 
 \newtheorem{exmp}{Example}[section]
  \newtheorem{ass}{Assumption}[section]

 \theoremstyle{definition}
 \newtheorem{rem}{Remark}
 
 \makeatletter
 \def\BState{\State\hskip-\ALG@thistlm}
 \makeatother
 
 \usepackage{authblk}

\def\spacingset#1{\renewcommand{\baselinestretch}%
{#1}\small\normalsize} \spacingset{1}
 
\definecolor{coquelicot}{rgb}{1.0, 0.22, 0.0}

\usepackage{setspace}
\title{ Fair Policy Targeting \footnote{We thank Graham Elliott, James Fowler, Ashesh Rambachan, Yixiao Sun and Kaspar W\"{u}thrich for helpful comments. All mistakes are our own.}  }
\author{ Davide Viviano\footnote{Corresponding author. Department of Economics, University of California at San Diego, La Jolla, CA, 92093. Email: dviviano@ucsd.edu.} $\quad$ Jelena Bradic\footnote{Department of  Mathematics and Halicio\u{g}lu Data Science Institute, 
        University of California at San Diego, La Jolla, CA, 92093. Email: jbradic@ucsd.edu.} }
\date{First Version: May, 2020 \\ This Version: June, 2022 }

\begin{document}
\maketitle
\spacingset{1.5}
\begin{abstract}
One of the major concerns of targeting interventions on individuals in social welfare programs is discrimination: individualized treatments may induce disparities across sensitive attributes such as age, gender, or race.
This paper addresses the question of the design of fair and efficient treatment allocation rules. We adopt the non-maleficence perspective of ``first do no harm'':  we select the fairest allocation \textit{within} the Pareto frontier. We cast the optimization into a mixed-integer linear program formulation, which can be solved using off-the-shelf algorithms.  We derive regret bounds on the unfairness of the estimated policy function and small sample guarantees on the Pareto frontier under general notions of fairness. Finally, we illustrate our method using an application from education economics.  
\end{abstract}


\noindent%
{\it Keywords:} Fairness, Causal Inference, Welfare Programs, Pareto Optimal, Treatment Rules. 
\vfill

\newpage

\section{Introduction} 

Heterogeneity in treatment effects, widely documented in social sciences, motivates treatment allocation rules that assign treatments to individuals differently, based on observable characteristics \citep{murphy2003optimal, manski2004}. 
 However, targeting individuals may induce disparities across sensitive attributes, such as age, gender, or race.
 Motivated by evidence for policymakers' preferences towards non-discriminatory actions \citep{cowgill2019economics},
this paper designs fair and efficient targeting rules for applications in welfare and health programs. We construct treatment allocation rules using data from experiments or quasi-experiments, and we develop policies that trade-off efficiency and fairness.

Fair targeting is a controversial task due to the lack of consensus on the formulation of the decision problem.
 Conventional approaches mostly developed in computer science consist in designing algorithmic decisions that maximize the \textit{expected} utility across all individuals by imposing fairness constraints on the decision space of the policymaker \citep{nabi2019learning}.\footnote{For a review, the reader may refer to \cite{corbett2018measure}. Further discussion on the related literature is contained in Section \ref{sec:literature}.} In contrast, the economic literature has outlined the importance of taking into account the welfare-effects of such policies \citep{kleinberg2018algorithmic}. Fairness constraints on the policymaker's decision space may ultimately lead to sub-optimal welfare for \text{both} sensitive groups. 
This is a significant limitation when the policymakers are concerned with the effects of their decisions on each individual's utilities: absent of legal constraints, we may not want to impose unnecessary constraints on the policy if such constraints are \textit{harmful} for some or all individuals.

This paper advocates for fair and Pareto optimal treatment rules. We discuss targeting in a setting where decision-makers prefer allocations for which we cannot find any other policy that strictly improves welfare for one of the two sensitive groups without decreasing welfare on the opposite group. Within such a set, she then chooses the fairest allocation. The decision problem is conceived for applications in social welfare and health programs and motivated by the Hippocratic notion of ``first do no harm'' ( ``\textit{primum non nocere}") \citep{rotblat1999hippocratic}: instead of imposing possibly harmful fairness constraints on the decision space, we restrict the set of admissible solutions to the Pareto optimal set, and among such, we choose the fairest one.
For example, during a health-program campaign, the policy-makers may not be willing to decrease all individuals' health status to gain fairness. Instead, they may be willing to \textit{trade-off} health-status of different groups (e.g., young and old individuals) when considering fairness. 
Our framework has three desirable properties: (i) it applies to general notions of fairness which may reflect different decision makers' preferences;
(ii) it guarantees Pareto efficiency of the policy-function, with the relative importance of each group solely chosen based on the notion of fairness adopted by the decision-maker; (iii) it also allows for arbitrary legal or ethical constraints, incorporating as a special case the presence of fairness constraints whenever such constraints are \textit{binding} due to ethical or legal considerations.\footnote{For binding fairness constraints, our proposed policy achieves a lower unfairness compared to the policy that maximizes welfare under fairness constraints while being Pareto optimal. See Section \ref{sec:decision} for details. }    
We name our method Fair Policy Targeting.


This paper contributes to the statistical treatment choice literature by introducing the notion, estimation procedures, and studying properties of Pareto optimal and fair treatment allocation rules. We allow for general notions of fairness, and as a contribution of independent interest, we define envy-freeness fairness \citep{varian1976two} for policy targeting.

The decision problem consists of lexicographic preferences of the policymaker of the following form: (i) Pareto dominant allocations are preferred over dominated ones; (ii) Pareto optimal allocations are ranked based on fairness considerations. We identify the Pareto frontier as the set of maximizers over any weighted average of each group's welfares.
Therefore, such an approach embeds as a special case maximizing a weighted combination of welfares of each sensitive group such as in \cite{athey2017efficient}, \cite{KitagawaTetenov_EMCA2018}\footnote{Under the utilitarian perspective considered in \cite{KitagawaTetenov_EMCA2018}, \cite{athey2017efficient}, the welfare maximization problem is equivalent to maximizing a weighted combination of the welfare of different groups with weights equal to corresponding probabilities. See Section \ref{sec:2} for more discussion.}, and in \cite{rambachan2020economic}. The above references take a specific weighted combination of welfares with weights as given, while in our case, weights are part of the decision problem and directly selected to maximize fairness. This has important practical implications: our procedure is solely based on the notion of fairness adopted by the social planner, and it does not require specific importance weights assigned to each sensitive group, which would be hard to justify to the general public.

Estimating the set of Pareto optimal allocations represents a fundamental challenge since (i) the set consists of maximizers over a continuum of weights between zero and one; (ii) each maximizer of the welfare (or a weighted combination of welfares) is often not unique \citep{elliott2013predicting}. To overcome these issues, we show that the Pareto frontier can be approximated using simple linear constraints.  We use a discretization argument, and we evaluate weighted combinations of the objective functions separately to construct a polyhedron that contains Pareto allocations. Our approach drastically simplifies the optimization algorithm: instead of estimating the entire set of Pareto allocations, we maximize fairness under easy-to-implement linear constraints. Our theorems show that the distance between the Pareto frontier obtained via linear constraints and its population counterpart converges uniformly to zero at rate $1/\sqrt{n}$.

We study regret guarantees, i.e., the difference between the estimated policy function's expected unfairness against the minimal possible unfairness achieved by Pareto optimal allocations. We characterize the rate under high-level conditions for general notions of unfairness and derive upper bounds that scale at rate $1/\sqrt{n}$, in several examples, and a lower bound that matches the same rate. 
An application and a calibrated numerical study on targeting student awards illustrate the advantages of the proposed method compared to alternatives that ignore Pareto optimality.

The remainder of this paper is organized as follows: we provide a brief overview of the literature in the following section; we introduce the decision problem in Section \ref{sec:2};  Section \ref{sec:3} discusses estimation; Section \ref{sec:4} contains the theoretical analysis; Section \ref{sec:counterfactual} discusses counterfactual notions of fairness; Section \ref{sec:app} discusses an empirical application and numerical studies, and Section \ref{sec:conclusion} concludes. Derivations and extensions are in the Appendix.

\subsection{Related Literature} \label{sec:literature}

This paper relates to a growing literature on statistical treatment rules  \citep{sun2020empirical, manski2004, athey2017efficient, armstrong2015inference, bhattacharya2012inferring, hirano2009asymptotics, KitagawaTetenov_EMCA2018, kitagawa2017equality, mbakop2016model, stoye2012minimax, tetenov2012statistical, viviano2019policy, zhou2018offline}. Further connections are also related to the literature on classification \citep{elliott2013predicting}. However, none of these discuss the design of fair and Pareto optimal decisions.

Fairness is a rising concern in economics, see \cite{cowgill2019economics},
 \cite{kleinberg2018algorithmic}, \cite{rambachan2020economic}. The authors provide economic insights on the characteristics of optimal decision rules when discrimination bias occurs. Here, we answer the different questions of the design and estimation of the optimal targeting rule within a statistical framework and derive the method's properties. A further difference is the decision problem with a multi-objective, instead of a single-objective utility function, as in previous references.
 Additional references include \cite{kasy2020} that provide comparative statics on the impact of fairness on the individuals' welfare, focusing on the analysis of algorithms, and \cite{narita2021incorporating} who motivates fairness based on incentive compatibility in the different context of the design of experiments. 

In computer science, Pareto optimality has been considered in the context of binary predictions by \cite{balashankar2019fair} and \cite{martinez2019fairness}. The authors propose semi-heuristic and computationally intensive procedures for estimating Pareto efficient classifiers. \cite{xiao2017fairness} discuss the different problem of estimation of a Pareto allocation that trade-offs fairness and individual utilities for recommender systems, where the relative importance weights of the different objectives are selected a-priori. These references do not address the treatment choice problem discussed in the current paper.  

References in computer science include  \cite{chouldechova2017fair}, \cite{dwork2012fairness}, \cite{hardt2016equality} among others. \cite{corbett2018measure} contain a review. Additional work also includes \cite{liu2017calibrated} who discuss fair bandits, and \cite{ustun2019fairness} who propose decoupled estimation of tree classifiers without allowing for exogenous (legal or economic) constraints on the policy space. While the above references address the decision problem as a prediction problem, several papers discuss algorithmic fairness within a causal framework \citep{coston2020counterfactual, kilbertus2017avoiding,  nabi2019learning, kusner2019making}. All such papers estimate decision rules under fairness constraints without discussing Pareto optimality. The different decision problem considered here is motivated by applications in social welfare and health programs. When not binding on policy-makers decisions, fairness constraints may lead to Pareto-dominated allocations and possibly harmful policies for advantaged and disadvantaged individuals. When fairness constraints are binding, instead, the decision problem proposed in this paper leads to fairer allocations compared to a constrained welfare maximization problem while not being Pareto dominated.

\section{Decision Making and Fairness} \label{sec:2}

 We start by introducing some notation. 
For each unit, we denote with $S \in \mathcal{S}$ a sensitive or protected attribute.   For expositional convenience, we let $\mathcal{S} = \{0,1\}$, with $S = 1$ denoting the disadvantaged group, and $X \in \mathcal{X} \subseteq \mathbb{R}^p$ individual characteristics. We define the post-treatment outcome with  $Y \in \mathcal{Y} \subseteq \mathbb{R}$ realized only once the sensitive attribute, covariates, and the treatment assignment are realized. We define $Y(d)$, $d \in \{ 0,1\}$ the potential outcomes under treatment $d$. The observed $Y$ satisfies the  
  Single Unit Treatment Value Assumption (SUTVA)  \citep{rubin1990formal}. 
Let
\begin{equation} \label{eqn:propensity} 
e(x, s) = P(D = 1 | X = x, S = s), \quad p_1 = P(S = 1)
\end{equation} 
be the propensity score and the probability of being assigned to the disadvantaged group. Here, treatments are independent of potential outcomes. 
\begin{ass}[Treatment Unconfoundedness] \label{ass:unconf} For $d \in \{0,1\}$,
$Y(d) \perp D | X, S$.  
\end{ass}



\subsection{Social Welfare}

Given observables, $(Y_i,X_i,D_i,S_i)$ we seek to design a treatment assignment rule (i.e. policy function)
$
\pi:  \mathcal{X} \times \mathcal{S} \mapsto \mathcal{T} \subseteq [0,1], \pi \in \Pi
$
that depends on the individual characteristics and protected attributes, and which can be either probabilistic or deterministic.\footnote{It is deterministic if $\mathcal{T} = \{0,1\}$ and probabilistic if $\mathcal{T} = [0,1]$.} Here, $\Pi$ incorporates given and binding legal or economic constraints that restrict the decision space. The welfare generated by a policy $\pi$ on those individuals with sensitive attribute $S = s$ is defined as\footnote{Welfare is interpreted from an \textit{intention-to-treat} perspective similarly to \cite{KitagawaTetenov_EMCA2018}, \cite{athey2017efficient}.} 
\vspace{-4mm}
\begin{equation} \label{eqn:welfare} 
W_s(\pi) = \mathbb{E}\Big[ (Y(1) - Y(0))\pi(X, S)\Big| S = s\Big].  
\end{equation} 

 Under the utilitarian perspective \citep{manski2004}, the welfare maximization problem, i.e., the population counterpart of the empirical welfare maximization (EWM) \citep{KitagawaTetenov_EMCA2018}, solves 
 \vspace{-6mm}
 \[
\max _{\pi \in \Pi} \Big\{ p_1  W_1(\pi) + (1-p_1)  W_0(\pi)\Big \}
 \]
 where $p_1$ is defined as in Equation \eqref{eqn:propensity}. However, whenever the sensitive group is a \textit{minority} group, welfare maximization assigns a small weight to the welfare of the minority, disproportionally favoring the majority group. An alternative approach is to maximize the welfare separately for each possible sensitive group designing different policies for different groups \citep{ustun2019fairness}. 
This approach may violate discriminatory laws, i.e., the resulting policy function violates the constraint in $\Pi$. A simple example is when, due to legal reasons, the policy $\pi(x,s)$ must be constant in the sensitive attribute $s$. Instead, we consider a framework where the policymaker simultaneously maximizes each group's welfare, imposing Pareto efficiency on the estimated policy, under arbitrary legal or economic constraints encoded in $\Pi$. Given the set of efficient policies, the planner then selects the least unfair one.  Our approach is designed for social and welfare programs where legal constraints naturally occur and where, given such constraints, the policymaker’s preferences align with classical notions of ``first do no harm".


\subsection{Pareto Principle for Treatment Rules }

The set of Pareto optimal choices is defined as $\Pi_{\textrm{\tiny  o}}$, and it contains all such allocations $\pi \in \Pi$ for which the welfare for one of the two groups cannot be improved without reducing the welfare for the opposite group. We characterize $\Pi_{\textrm{\tiny  o}}$ in the following lemma.

\begin{lem}[Pareto Frontier] \label{lem:hyperplane}The set {\normalfont$\Pi_{\mbox{\tiny o}} \subseteq \Pi$} is such that 
\normalfont{
\begin{align}  \label{eqn:hg}
\Pi_{\textrm{\tiny  o}} = \Big\{\pi_\alpha: \pi_\alpha \in \mathrm{arg} \sup_{\pi \in \Pi} \alpha W_1(\pi) + (1 - \alpha)W_0(\pi), \quad \alpha \in (0,1)\Big\}. 
\end{align} 
}
\end{lem}

The lemma follows from \cite{negishi1960welfare}, whose proof is in Appendix \ref{sec:a1}. It will be convenient to define 
 \begin{equation} \label{eqn:W_a}
  \bar{W}_\alpha = \sup_{\pi \in \Pi} \alpha W_1(\pi) + (1 - \alpha)W_0(\pi),
  \end{equation} 
 the largest value of the objective in Equation \eqref{eqn:hg} for a fixed $\alpha$. In the following examples, we show that Pareto allocations \textit{generalize} notions of treatment rules from previous literature.

\begin{exmp}[Welfare Maximization]  The population equivalent of the EWM problem belongs to the Pareto frontier. Namely,  
$
\mathrm{arg} \max_{\pi \in \Pi} \Big \{p_1 W_1(\pi) + (1 - p_1)W_0(\pi) \Big\} \subseteq \Pi_{\mbox{\tiny o}}.
$
An alternative approach consists in maximizing weighted combinations of the welfare with the weights for each group as given. For instance the allocation \citep{rambachan2020economic}
\begin{equation} \label{eqn:omega} 
\check{\pi}_{\omega} \in \mathrm{arg} \max_{\pi \in \Pi} \Big \{\omega W_1(\pi) + (1 - \omega)W_0(\pi) \Big\} \subseteq \Pi_{\mbox{\tiny o}}, 
\end{equation} 
for some \textit{specific} weight $\omega$ belongs to the Pareto frontier. \qed 
\end{exmp}


Pareto optimal allocations are often non-unique, allowing for flexibility in the choice of
efficient policies. The policy-maker must appeal to some preferential ranking principle based on her preferences. We discuss those in the following lines.

\subsection{Decision Problem} \label{sec:decision}

We start by defining $\mathcal{C}(\Pi)$ the \textit{choice set} of the policy maker \citep{mas1995microeconomic}, where $\mathcal{C}$ is a choice function with $\mathcal{C}(\{\pi_1, \pi_2\}) = \pi_1$ if $\pi_1$ is strictly preferred to $\pi_2$.  We let
\vspace{-3mm}
\begin{equation} \label{eqn:unfairness} 
\mathrm{UnFairness}: \Pi \mapsto \mathbb{R}
\end{equation}
an operator which quantifies the unfairness of a policy. We leave unspecified UnFairness and provide examples in Section \ref{sec:examples} and Section \ref{sec:counterfactual}. We now state the planner's preferences. 

\begin{ass}[Policy-maker's Preferences] \label{ass:optimality} Preferences are rational\footnote{Rational preferences imply transitivity and completeness \citep{mas1995microeconomic}.}, and for each $\pi_1, \pi_2 \in \Pi$,
(i) $\mathcal{C}(\{\pi_1, \pi_2\}) = \pi_1$ if $W_1(\pi_1) \ge W_1(\pi_2)$ and $W_0(\pi_1) \ge W_0(\pi_2)$ and either (or both) of the two inequalities hold strictly; (ii) if neither $\pi_1$ Pareto dominates $\pi_2$ nor $\pi_2$ Pareto dominates $\pi_1$, $\mathcal{C}(\{\pi_1, \pi_2\}) = \pi_1$ if UnFairness($\pi_1$) $<$ UnFairness($\pi_2$); (iii) if neither Pareto dominates the other and with equal UnFairness,
$\mathcal{C}(\{\pi_1, \pi_2 \}) = \{\pi_1 , \pi_2 \}$.
\end{ass} 

Assumption \ref{ass:optimality} postulates lexicographic preferences of the following form: 
 (i) an allocation is strictly preferred to another if it weakly improves welfare for both groups and strictly improves welfare for at least one group; (ii) given two allocations where none of the two Pareto dominates the other, allocations are ranked based on fairness.

 
While different applications reflect different planner's preferences, Assumption \ref{ass:optimality} is motivated by the applications for social welfare and health programs, where welfare depends on outcomes such as health-status \citep{finkelstein2012oregon}, future earnings,  or school achievements. Sacrificing welfare (e.g., \textit{health-status}) of each group for fairness is undesirable in such applications. Conditional on achieving a Pareto efficient allocation, the planner minimizes UnFairness. Whenever, however, fairness constraints are binding (e.g., because of legal considerations), these can be directly incorporated in the function class $\Pi$. We can now characterize the decision problem.


\begin{prop}[Decision Problem] \label{lem:soc_p} Under Assumption \ref{ass:optimality}, $\pi^\star \in \mathcal{C}(\Pi)$ if and only if  
\begin{equation} \label{eqn:main_eq} 
\begin{aligned} 
\pi^\star \in \mathrm{arg} &\inf_{\pi \in \Pi} \mathrm{UnFairness}(\pi)  \\ \text{subject to } \alpha W_1(\pi) + &(1 - \alpha) W_0(\pi) \ge \bar{W}_{\alpha}, \text{ for some } \alpha \in (0,1). 
\end{aligned} 
\end{equation} 
\end{prop} 

The proof is contained in Appendix \ref{sec:a1}. Proposition \ref{lem:soc_p} formally characterizes the policy-makers decision problem, which consists of minimizing the policy's unfairness criterion, under the condition that the policy is Pareto optimal. The policy-maker does not maximize a weighted combination of welfares, with some \textit{pre-specified} and hard-to-justify weights. Instead, each group's importance (i.e., $\alpha$) is implicitly chosen within the optimization problem to maximize fairness. This approach allows for a transparent choice of the policy based on the policy-makers definition of fairness.

 \begin{exmp}[Why Pareto Efficiency? A simple example] \label{exmp:omega} Let $X = 1$ for simplicity, take $\tau_s, \phi \in (0,1), s \in \{0,1\}$ and let
$
Y(d) = \tau_S d + \varepsilon(d), 
$
with $\mathbb{E}[\varepsilon(d) | S] = 0$.
Consider a class of probabilistic decision rules 
$$
\Pi = \Big\{\pi(x, s) = \beta_s, \quad \beta_1, \beta_0 \in (0,1), \quad \beta_0 p_0 + \beta_1 p_1 \le \phi \Big\},  
$$ 
with the share of treated units being at most $\phi$. Let UnFairness be the difference in the groups' welfares, namely 
$
|\tau_1 \beta_1 - \tau_0 \beta_0|.
$ The smallest possible unfairness is zero, since we can choose $\beta_1 = \beta_0 = 0$ with one of the fairest allocation selecting none of the individuals to treatment. Consider now the Pareto frontier, defined as:\begin{equation} \label{eqn:set} 
\Pi_{\mbox{\tiny o}} = \Big\{\pi(x,s) = \beta_s^*, \quad \beta_0^* = \frac{\phi - p_1 \beta_1^*}{p_0}, \beta_1^* \in [0,1] \Big\} \subset \Pi. 
\end{equation}  
The set of Pareto allocation rules out all those allocation for which the capacity constraint is attained with strict inequality, also excluding $\beta_1 = \beta_0 = 0$. The proposed policy assigns all benefits to individuals, and it trade-offs \textit{who} to treat to minimize $|\tau_1 \beta_1 - \tau_0 \beta_0|$.\footnote{Observe that the level of unfairness with the frontier may or may not be potentially strictly larger than the unfairness obtained in an unconstrained scenario. Namely, to achieve zero unfairness for every $\pi \in \Pi_{\mbox{\tiny o}}$ , we need that $\tau_1 \beta_1^* = \tau_0 \beta_0^*$. Substituting $\beta_0^* = \phi/p_0 - p_1 \beta_1^*/p_0$ this would require $\beta_1^* = \frac{\phi}{p_0} (\tau_1/\tau_0 + p_1/p_0)^{-1}$ which is not necessarily feasible (i.e., the expression is larger than one).} \qed 
\end{exmp}


We conclude by comparing the properties of the policy in Proposition \ref{lem:soc_p} with existing alternatives, stated as a corollary of Proposition \ref{lem:soc_p}. In particular, we compare our method with the policy that maximizes welfare with importance weights for different groups \citep{rambachan2020economic} in Equation \eqref{eqn:omega} and the one with fairness constraints. 
For the latter, define 
 $
 \Pi(\kappa) = \Big\{\pi \in \Pi: \mathrm{UnFairness}(\pi) \le \kappa\Big\} \subseteq \Pi, 
 $
 the set of policies with constraint, and 
\begin{equation} \label{eqn:pi_c} 
 \widetilde{\pi} \in \arg \max_{\pi \in \Pi(\kappa)} p_1 W_1(\pi) + (1 - p_1) W_0(\pi)
 \end{equation} 
the policy that maximizes the welfare imposing fairness constraints \citep{nabi2019learning}.

\begin{cor}[Properties] \label{cor:prop}  Let $\pi^\star$ be defined as in Equation \eqref{eqn:main_eq} and $\pi_{\omega}, \tilde{\pi}$ and in Equation \eqref{eqn:omega}, \eqref{eqn:pi_c}, respectively. Then  
$
\mathrm{UnFairness}(\pi^\star) \le \mathrm{UnFairness}(\check{\pi}_\omega), \forall \omega \in (0,1)$.  

Suppose that either $\widetilde{\pi} \in \Pi_{\mbox{\tiny o}}$ (i.e., it belongs to the Pareto frontier), or fairness constraints are binding to the policy-maker, i.e. $\Pi(\kappa) = \Pi$. Then  
 $
\mathrm{UnFairness}(\pi^\star) \le \mathrm{UnFairness}(\widetilde{\pi}). 
$
Suppose instead that $\widetilde{\pi} \not \in \Pi_{\mbox{\tiny o}}$. Then 
 $
\mathrm{UnFairness}(\pi^\star) \le \mathrm{UnFairness}(\pi_{\mbox{\tiny o}}) 
$ 
for all $\pi_{\mbox{\tiny o}} \in \Pi_{\mbox{\tiny o}}$ that Pareto dominate $\widetilde{\pi}$. In addition, $\pi_{\omega}$ and $\tilde{\pi}$ do not Pareto dominate $\pi^\star$.  
\end{cor}

Corollary \ref{cor:prop} shows that UnFairness of $\pi^\star$ is Pareto optimal and \textit{uniformly} smaller than UnFairness of the policy $\pi_\omega$ that maximizes a weighted combination of the welfares. It also shows that if $\widetilde{\pi}$ \textit{is} Pareto optimal, then its UnFairness is larger than UnFairness of $\pi^\star$. When instead $\widetilde{\pi}$ is \textit{not} Pareto optimal, its Pareto dominant allocations have larger UnFairness than $\pi^\star$. 
Further intuition can be gained under strong duality, which we discuss in Appendix \ref{sec:a22}. Intuitively, the constraint in Proposition \ref{lem:soc_p} holding for \textit{some} weighted combinations of welfares (instead of a particular choice of the weights) is key to achieve lower unfairness of $\pi^\star$ relative to $\widetilde{\pi}$, when $\widetilde{\pi}$ is Pareto efficient.\footnote{Under strong duality, the dual of $\widetilde{\pi}$ corresponds to minimize UnFairness for \textit{one particular} weighted combination of welfare exceeding a certain threshold. In contrast, our decision problem imposes the constraint that \textit{some} weighted combination of welfares exceeding a certain threshold. This difference reflects the difference between the lexicographic preferences that we propose as opposed to an additive social planner's utility.} Finally, when fairness constraints are binding, the proposed procedure always leads to smaller UnFairness.

\section{Fair Targeting: Estimation} \label{sec:3}


We now construct an estimator of $\pi^\star$. We introduce some notation, and we define 
\vspace{-2mm}
\begin{equation} \label{eqn:gamma}
m_{d,s}(x ) = \mathbb{E}\Big[Y_i(d) \Big| X_i = x, S_i=s \Big] , \quad \Gamma_{d,s,i} = \frac{1\{S_i = s\}}{p_s} \Big[ \frac{1\{D_i = d\}}{e(X_i, S_i)} \Big(Y_i - m_{d,s}(X_i)\Big) + m_{d,s}(X_i)\Big]
\end{equation} 
the conditional mean of the group $s$ under treatment $d$, and 
 the doubly robust score \citep{robins1995semiparametric}, respectively. We let $\hat{\Gamma}_{d,s,i}$ the estimated counterpart of $\Gamma_{d,s,i}$. Define

\begin{equation} \label{eqn:welf_hat}
\begin{aligned} 
&\hat{W}_s(\pi) = \frac{1}{n} \sum_{i=1}^n \Big(\hat{\Gamma}_{1, s,i} - \hat{\Gamma}_{0,s,i}\Big) \pi(X_i,s). 
\end{aligned} 
\end{equation} 
the estimated welfare built upon semi-parametric  literature  \citep{newey1990semiparametric, robins1995semiparametric}, with 
$
\hat{m}_{d,s}(.),\hat{e}(.), \hat{p}_s, 
$
constructed via  cross-fitting \citep{chernozhukov2018double}. Details of the cross-fitting procedure are contained in Appendix \ref{sec:a23}. We consider first general notions of fairness, and introduce the corresponding estimator below.

\begin{defn}[Empirical UnFairness] We define $\mathcal{V}_n(\pi, p_s, e, m)$ an unbiased estimate of $\mathrm{UnFairness}(\pi)$ which depends on observables and the population propensity score and conditional mean. We write 
$
\widehat{\mathcal{V}}_n(\pi) = \mathcal{V}_n(\pi, \hat{p}_s, \hat{e}, \hat{m}), 
$
the empirical counterpart. 
\end{defn} 

 We defer to Section \ref{sec:examples} and Section \ref{sec:counterfactual} explicit examples of $\widehat{\mathcal{V}}_n(\pi)$.

\subsection{(Approximate) Pareto Optimality}

Next, we characterize the Pareto frontier using linear inequalities. To construct the Pareto frontier we use the constraint in Equation \eqref{eqn:main_eq} after discretizing the set of weights $\alpha$. Namely, in the first step, we discretize the Pareto frontier, and construct a grid of equally spaced values $\alpha_j \in (0,1)$, $j \in \{1, ..., N\}$, 
with $N = \sqrt{n}$. 
We approximate the Pareto frontier using the set ($\hat{W}_0, \hat{W}_1$ are defined in Equation \eqref{eqn:welf_hat})
\begin{equation} \label{eq:pf}
\widehat{\Pi}_{\mbox{\tiny o}}  =  \Big\{\pi_\alpha \in \Pi: \pi_\alpha \in \mathrm{arg} \sup_{\pi \in \Pi} \Big\{\alpha \hat{W}_0(\pi) + (1 - \alpha) \hat{W}_1(\pi) \Big\}, \text{ s.t. } \alpha \in \{\alpha_1, ..., \alpha_N\}\Big\}. 
\end{equation}
The grid's choice is arbitrary, as long as values are \textit{equally spaced}. 

The set $\widehat{\Pi}_{\mbox{\tiny o}}$ may be hard, if not impossible, to directly estimate, since we may have uncountably many solutions \citep{manski1989estimation, elliott2013predicting}. In particular, the solution to each optimization problem in Equation \eqref{eq:pf} may not be unique.  
Instead of directly estimating $\widehat{\Pi}_{\mbox{\tiny o}}$, we characterize it through linear constraints. First, we find the largest empirical welfare achieved on the discretized Pareto Frontier defined as
\begin{equation} \label{eqn:opt1} 
 \bar{W}_{j,n} = \sup_{\pi \in \Pi} \Big\{ \alpha_j \hat{W}_0(\pi) + (1 - \alpha_j) \hat{W}_1(\pi) \Big\}, \text{ for each } j \in \{1, ..., N\}, 
\end{equation} 
which can be obtained through standard optimization routines \citep{KitagawaTetenov_EMCA2018, zhou2018offline}.
Second, we observe that any $\pi \in \widehat{\Pi}_{\mbox{\tiny o}}$, must satisfy $\alpha_j \hat{W}_0(\pi) + (1 - \alpha_j) \hat{W}_1(\pi)  \ge \bar{W}_{j,n}, \text{ for some } j \in \{1, ..., N\}$, since $\bar{W}_{j,n}$ defines the largest objective for a given $\alpha_j$. We impose such constraint up to a small slackness parameters $\lambda/\sqrt{n}$ and
construct an approximate Pareto frontier as follows: 
\begin{equation} \label{eqn:set_const2}
\widehat{\Pi}_{\mbox{\tiny o}}(\lambda) = \Big\{\pi \in \Pi: \exists j \in \{1, ..., N\} \mbox{ such that } \alpha_j \hat{W}_{0,n}(\pi) + (1 - \alpha_j) \hat{W}_{1,n}(\pi) \geq  \bar{W}_{j, n} - \frac{\lambda}{\sqrt{n}}\Big\}, 
\end{equation} 
where $\widehat{\Pi}_{\mbox{\tiny o}}(0) = \widehat{\Pi}_{\mbox{\tiny o}}$, and $\widehat{\Pi}_{\mbox{\tiny o}} \subseteq \widehat{\Pi}_{\mbox{\tiny o}}(\lambda)$ for any $\lambda \ge 0$.  

Here, we introduced $-\frac{\lambda}{\sqrt{n}}$ which imposes that the resulting policy is ``approximately'' Pareto optimal. As shown in Section \ref{sec:4}, $\lambda/\sqrt{n}$ guarantees that $\widehat{\Pi}_{\mbox{\tiny o}}(\lambda)$ contains all Pareto optimal policies with high-probability, for $\lambda = \mathcal{O}(1)$.
The estimated policy is defined as 
\begin{equation} \label{eqn:pp} 
\hat{\pi}_\lambda \in \mathrm{arg} \min_{\pi \in \widehat{\Pi}_{\mbox{\tiny o}}(\lambda)} \widehat{\mathcal{V}}_n(\pi). 
\end{equation}

\begin{rem}[The choice of the grid and $\lambda$]  The choice of $\lambda$ depends on the function class $\Pi$. In Theorems \ref{thm:1b}, \ref{thm:2}, \ref{thm:between_groups} we discuss guarantees by imposing that $\lambda/\sqrt{n} \ge M \sqrt{v}/N$, for some finite constant $M$ with $\lambda$  increasing in the geometric complexity $v$ of $\Pi$, and where we choose $N = \sqrt{n}$.\footnote{This guarantees that the estimated function class does not exclude Pareto optimal policies with high probability, while $\lambda = \mathcal{O}(1)$ guarantees uniform converence of the Pareto frontier at $1/\sqrt{n}$ rate.}
In contrast, the function class complexity does not affect the choice of the grid (i.e., $N$). This is because the welfare loss due to the grid's approximation error is uniformly bounded by a constant independent of $\Pi$.\footnote{Namely, take a grid of $N + 1$ equally spaced $\alpha_j$. Then the approximation error reads as   
$
\sup_{\pi \in \Pi} |\alpha W_1(\pi) + (1 - \alpha) W_0(\pi) - \max_{\alpha_j \in \{\alpha_1, \cdots, \alpha_N\}} \alpha_j W_1(\pi) - (1 - \alpha) W_0(\pi)| \le 2M/N,
$
which is uniformly bounded by $M$ where $M$ bounds the first moment of the potential outcomes independent of $\Pi$. 
} 
 \qed 
\end{rem}

\subsection{Optimization: Mixed Integer Quadratic Program}

We provide a mixed-integer quadratic program (MIQP) for optimization. We define
$
\mathbf{z}_s =(z_{s,1},\cdots, z_{s,n}),  z_{s,i} = \pi(X_i, s), \pi \in \Pi.
$
Here, $z_{s,n}$ defines the treatment assignment under policy $\pi$ and sensititive attribute $s$ (see the example below); $\mathbf{z}_s$ have simple representation for general classes of policy functions, such as either probabilistic rules which we derive in Appendix \ref{sec:a23} or deterministic linear decision rules \citep{florios2008exact}.

 \begin{exmp}[Maximum score]  For the maximum score $\pi(X_i, s) = 1\{X_i \beta_x + S \mu \ge 0\}, \beta = (\beta_x, \mu) \in \mathcal{B}$, the indicators $z_{s,n}$ are defined via mixed-integer constraints of the form \citep{florios2008exact} 
$
\frac{X_i^\top \beta + s \mu}{|C_i|} < z_{s,i} \le \frac{X_i^\top \beta + s \mu }{|C_i|} + 1,  C_i \ge \sup_{\beta \in \mathcal{B}} |(X_i, S_i)^\top \beta|, z_{s,i} \in \{0,1\}. 
$ 
Such constraint guarantees that $z_{s,i} = 1\{X_i^\top \beta_x + s \mu \ge 0\}$. \qed 
\end{exmp} 

We now need to impose the constraint of Pareto optimality. To do so, we introduce an additional set of decision variables that guarantee the constraints in Equation \eqref{eqn:set_const2} hold. 
The vector $\mathbf{u}=(u_1, ..., u_N) \in \{0,1\}^N$
encodes the locations on the grid of $\alpha$ for which the supremum in \eqref{eqn:set_const2} is reached at; here,  $u_j=1$ whenever the constraint in Equation \eqref{eqn:set_const2} holds for $\alpha_j$. 
The chosen policy must be Pareto optimal, i.e., $u_j$ must be equal to one for at least one $j$.  To ensure this, we impose the constraint  $ \sum_{j=1}^N u_j \ge 1$. 

Combining such constraints, it directly follows that $\hat{\pi}_\lambda$ satisfies Equation \eqref{eqn:pp} if and only if

\vspace{-9mm}

\begin{equation} \label{eqn:opt1}
\hat{\pi}_\lambda \in \mathrm{arg} \min_{\pi} \min_{ \mathbf{z_{0}}, \mathbf{z_{1}} ,\mathbf{u}} \qquad   \widehat{\mathcal{V}}_n(\pi)  
\end{equation} 
\vspace{-15mm}
\begin{align}
   \mbox{subject to}  & \quad \ z_{s,i} = \pi(X_i, s),  \quad 1 \leq i \leq n, & \mbox{(A)}  & \ \nonumber \\
  &\quad  \  u_j \alpha_j   \langle \hat{\boldsymbol{\Gamma}}_{1, 0} - \hat{\boldsymbol{\Gamma}}_{0, 0},  \boldsymbol{z}_0 \rangle + u_j (1 - \alpha_j)   \langle \hat{\boldsymbol{\Gamma}}_{1, 1} - \hat{\boldsymbol{\Gamma}}_{0, 1},  \boldsymbol{z}_1 \rangle  \ge u_j  n\bar{W}_{j, n} -\sqrt{n} \lambda  \  & \mbox{(B)}&\nonumber \\
   &\quad  \  \langle \mathbf{1}, \mathbf u \rangle  \ge 1 \nonumber  & \mbox{(C)}& \\
 &\quad \  \pi \in \Pi  & \mbox{(D)}& \nonumber \\
   & \quad \  u_j  \in \{0,1\},     \  \ 1 \leq j \leq N.  & \mbox{(E)}&\nonumber
\end{align}   



  Here, $\mathbf{\Gamma}_{d,s}$ is the vector of $\Gamma_{d,s, i}$ defined in Equation \eqref{eqn:gamma}. 
 Constraints (B) and (C) state that the resulting policy is (approximately) Pareto optimal, or, equivalently, it maximizes a weighted combination of groups' welfare for \textit{some} $\alpha_j$. Constraints (A), (C), (E) are (mixed-integer) linear constraints, while Constraint (B) is quadratic. Notice that we can further simplify (B) as a linear constraint at the expense of introducing additional $N n$ binary variables and $2 N n$ additional constraints (e.g., see \citealt{wolsey1999integer, viviano2019policy}). Finally, (D) is either linear or quadratic for deterministic assignments and linear probability models.  Hence the objective admits a MIQP representation whenever $\widehat{\mathcal{V}}_n(\pi)$ admits linear representation in $\pi$, as discussed in the following section. 
Note that the solution to the optimization problem might not be unique, depending on the function class. However, non-uniqueness does not affect theoretical properties in Section \ref{sec:4}. 

\begin{rem}[Computational complexity] 
The complexity of the optimization problem depends on the policy function class. For discrete covariates, in Section \ref{sec:complexity} we show that the problem can be solved as a \textit{sequence} of linear programs, for which algorithms that returns exact solutions in polynomial time exist; e.g., \citealt{karmarkar1984new}. For the maximum score and the optimal tree, researchers may rely on existing algorithms such as the branch and bound \citep{wolsey1999integer}, efficiently computed by existing software; e.g., GUROBI and CPLEX. However, their worst-case scalability may grow exponentially with the sample size, similarly to what discussed in the policy learning literature \citep{zhou2018offline, KitagawaTetenov_EMCA2018}. One solution for the optimal tree is to use an exhaustive search method \citep{zhou2018offline}, which, we show in Section \ref{sec:numerics} is feasible for a moderately large sample size also for our program. For the maximum score, researchers may instead use the early termination strategy, which we study in Section \ref{sec:complexity}.   \qed 
\end{rem}

\section{Theoretical Analysis} \label{sec:4}


Below we impose that Condition (A), which restricts the function class of interest of the policy function, which holds for linear scores \citep{manski1975maximum}, and decision trees \citep{zhou2018offline}.  Condition (B) ensures the measurability \citep{rai2018statistical, kosorok2008introduction}.

\begin{ass} \label{ass:moment} Suppose that the following conditions hold: (A) $\Pi$ has finite VC-dimension, denoted as $v$; (B) $\Pi$ is pointwise measurable.
\end{ass}

\begin{ass}\label{ass:overlap} Let: (i) $e(X_i, s), p_s \in (\delta, 1 - \delta)$, almost surely, for $\delta \in (0,1)$, for all $s \in \{0,1\}$; (ii) $Y_i(d) \in [-M, M]$, for some $M < \infty$, for all $d \in \{0,1\}$ almost surely.    
\end{ass} 

Condition (i) imposes the standard overlap assumption; Condition (ii) assumes uniformly bounded outcomes \citep[e.g.,][for related conditions]{mbakop2016model}. The following assumptions are imposed on the estimators.

\begin{ass}[Nuisances' regularities] \label{ass:dr} For some $\xi_1 \ge 1/4, \xi_2 \ge 1/4$:  
\begin{equation}
\small 
\begin{aligned}
&\mathbb{E}\Big[\Big(\hat{m}_{d,s}(X_i) - m_{d,s}(X_i)\Big)^2\Big] = \mathcal{O}(n^{-2\xi_1}), \quad \mathbb{E}\Big[\Big(1\Big/\hat{p}_s \hat{e}(X_i, s) - 1\Big/p_s e(X_i, s)\Big)^2 \Big] = \mathcal{O}(n^{-2\xi_2}).
\end{aligned}  
\end{equation} 
for all $s,d \in \{0,1\}$, where $X_i$ is out-of-sample. In addition, for a finite constant $M$ and $\delta \in (0,1)$, $\sup_{d \in \{0,1\},s \in \{0,1\}, x \in \mathcal{X}} |\hat{m}_{d,s}(x)| < M$, and $\hat{e}(X, S), \hat{p} \in (\delta, 1 - \delta)$ almost surely.  
\end{ass} 

Assumption \ref{ass:dr} states that the \textit{product} of the mean-squared error of the estimated propensity score and conditional mean converges to zero at the parametric rate. This condition is standard in the doubly-robust literature \citep{chernozhukov2018double, farrell2015robust}.  
Assumption \ref{ass:dr} also states that the conditional mean and the propensity score functions are uniformly bounded. The conditions can be stated asymptotically, in which case the uniform bound on estimated nuisance functions is not required and results should be interpreted in the asymptotic sense only \citep{athey2017efficient}.

\subsection{Guarantees on the Pareto Frontier} \label{sec:41}

It is interesting to study the behavior of the estimated frontier relative to its population counterpart. We do so in the following theorems.  

\begin{thm} \label{lem:cons_dr} Under Assumptions \ref{ass:unconf}, \ref{ass:moment}-\ref{ass:dr}, for any $\gamma \in (0,1), \lambda \ge 0$, a universal constant $c_0 < \infty$, with probability larger than $1  - \gamma$, 
\begin{equation} \label{eqn:dist} 
\small 
\begin{aligned} 
&\sup_{\alpha \in (0,1), \pi \in \Pi} \Big| \alpha W_{0}(\pi) + (1 - \alpha) W_{1}(\pi) - \max_{\alpha_j \in \{\alpha_1, \cdots , \alpha_N\}} \Big\{\alpha_j \widehat{W}_{0}(\pi) + (1 - \alpha_j) \widehat{W}_{1}(\pi) - \frac{\lambda}{\sqrt{n}} \Big\} \Big| \le  \\ &c_0 \sqrt{\frac{v}{n}} +  c_0 \sqrt{\frac{\log(2/\gamma)}{n}} + \frac{\lambda}{\sqrt{n}}. 
\end{aligned} 
\end{equation} 
\end{thm}

Theorem \ref{lem:cons_dr} shows that the distance between the estimated Pareto frontier and its population counterpart converges to zero at rate $1/\sqrt{n}$ for a choice of $\lambda = \mathcal{O}(1)$ where $\lambda$ is defined in Equation \eqref{eqn:set_const2}. The derivation uses properties of the double-robust estimator \citep{farrell2015robust}, and connects to the literature on empirical welfare maximization \citep{KitagawaTetenov_EMCA2018, zhou2018offline, athey2017efficient}, while differently here we control the maximum deviation uniformly over a set of weights $\alpha$. 

A natural question is whether the estimated Pareto frontier also contains all Pareto optimal allocations for a finite $\lambda$. We complement Theorem \ref{lem:cons_dr} by showing that with high probability the set of estimated allocations $\widehat{\Pi}_{\textrm{\mbox{\tiny o}}}(\lambda)$ contains the Pareto frontier for finite $\lambda$.

\begin{thm}  \label{thm:1b} Let Assumptions \ref{ass:unconf}, \ref{ass:moment}-\ref{ass:dr} hold. For any $\gamma \in (0,1), \lambda \ge \underline{b} (\sqrt{v} + \sqrt{\log(2/\gamma)})$, for a constant $\underline{b} > 0$, independent of $n$, $N = \sqrt{n}$,  then 
$
\mathbb{P}\Big( \Pi_{\textrm{\tiny  o}} \subseteq \widehat{\Pi}_{\textrm{\mbox{\tiny o}}}(\lambda)\Big)  \ge  1- \gamma                                                       . 
$
\end{thm}

 Theorem \ref{thm:1b} complements Theorem \ref{lem:cons_dr} showing that it suffices $\lambda = \mathcal{O}(1)$ (and hence a slackness of order $\mathcal{O}(1/\sqrt{n})$) for the set of estimated allocations to contain the Pareto frontier. The proofs of Theorems \ref{lem:cons_dr}, \ref{thm:1b} are contained in Appendix \ref{sec:a1}.
Theorem \ref{thm:1b} uses \textit{finite sample} properties of the estimated (discretized) frontier showing uniform concentration. The choice of $\lambda/\sqrt{n}$ matches the upper-bound on the maximal deviations, and the choice $N = \sqrt{n}$ guarantees that the grid is coarse enough to control the estimation error.

\begin{rem}[Non-binary policies]  While our framework considers a binary action spaces, our guarantees also generalize to multi-action policies. In such a case, the bound depends on the entropy integral of $\Pi$ and the derivation leverages concentration of $\Big|\hat{W}_d(\pi) - W_d(\pi)\Big|$ for multi-action spaces \citep{zhou2018offline}. See Appendix \ref{sec:multi_actions} for a discussion. \qed 
\end{rem}

\subsection{General Fairness Bounds}

Given the guarantees on the frontier, we next analyze guarantees on fairness. We start our discussion by introducing regret bounds for generic notions of unfairness under high-level assumptions and then provide examples of upper and lower bounds.

\begin{ass}[High-level conditions on UnFairness] \label{ass:unique2} For some $\eta > 0, \gamma > 0$, 
$$\mathbb{P}\Big(\sup_{\pi \in \Pi} \Big|\widehat{\mathcal{V}}_n(\pi) - \mathrm{UnFairness}(\pi)\Big| \le \mathcal{K}(\Pi, \gamma) n^{-\eta}\Big) \ge 1 - \gamma
$$
for some $\mathcal{K}(\Pi, \gamma) < \infty$. Also assume that $\mathrm{UnFairness}(\pi)$ is uniformly bounded.
\end{ass} 
 Assumption \ref{ass:unique2} states that the estimated unfairness converges with probability $1 - \gamma$ to population unfairness uniformly over $\Pi$ at rate $n^{-\eta}$ for some arbitrary $\eta$. The constant $\mathcal{K}(\Pi, \gamma)$ depends on the function class' complexity and the probability $\gamma$. We characterize the constant and the rate $\eta$ in examples in Section \ref{sec:examples} and Appendix \ref{sec:a25}.  
\begin{thm} \label{thm:2} Let Assumptions \ref{ass:unconf}, \ref{ass:moment}-\ref{ass:unique2} hold. Then for some constants $0 < c_0, \underline{b} < \infty$, independent of $n$, $\lambda \ge \underline{b} (\sqrt{v} + \sqrt{\log(2/\gamma)}), N = \sqrt{n},$
with probability at least $1 - 2 \gamma$, 
\normalfont{
\begin{equation}  \label{eqn:distributions}
\begin{aligned} 
&  \mathrm{UnFairness}(\hat{\pi}_\lambda) - \inf_{\pi \in \Pi_{\textrm{\mbox{\tiny o}}}} \mathrm{UnFairness}(\pi) \le \frac{c_0}{\sqrt{n}} + \frac{c_0 \mathcal{K}(\Pi, \gamma)}{n^{\eta}}.  
\end{aligned} 
\end{equation} 
}
\end{thm}

The proof is contained in Appendix \ref{sec:a1} and leverages  Theorem \ref{thm:1b} to show that the set of Pareto allocations is contained with high-probability within the estimated allocations.
Theorem \ref{thm:2} characterizes the convergence rate of the UnFairness of the estimated policy relative to the lowest unfairness within the class of Pareto allocations.  
To our knowledge, this is the first result of this type of fair policy. The rate depends on the convergence rate of the estimated UnFairness. 
In the following paragraphs, we discuss examples and sufficient conditions for Assumption \ref{ass:unique2} to hold and formally characterize the rate of convergence $\eta$ and the constant $\mathcal{K}(\cdot)$.

\subsection{Regret: Examples and Rate Characterization} \label{sec:examples}

Here we discuss three examples, one based on policy predictions, a second based on the welfare-effect, and a third based on incentive compatibility.

\begin{defn}[Prediction disparity] \label{defn:pred}
Prediction disparity and its empirical counterpart take the following form
$$
C(\pi) = \mathbb{E}\Big[\pi(X, S) | S = 0\Big] - \mathbb{E}\Big[\pi(X,S) | S = 1\Big] , \quad \hat{C}(\pi) = \frac{\sum_{i=1}^n \pi(X_i) (1 - S_i)}{n(1 - \hat{p}_1)} - \frac{\sum_{i=1}^n \pi(X_i) S_i}{n \hat{p}_1},  
$$ 
\end{defn} 

Prediction disparity captures disparity in the treatment probability between groups. The second notion of UnFairness measures welfare disparities between the two groups.

\begin{defn}[Welfare disparity] \label{defn:welfare}
Define the welfare disparity and its empirical counterpart as 
 $$
 D(\pi) = W_0(\pi) - W_1(\pi) , \quad \widehat{D}(\pi) = \widehat{W}_0(\pi) - \widehat{W}_1(\pi). 
 $$ 
\end{defn}

Between-groups disparity captures the difference in \textit{welfare} between the advantaged group $(S = 0)$ and the disadvantaged group $(S = 1)$, relative to the baseline.\footnote{Recall the definition of welfare in Equation \eqref{eqn:welfare} where we only consider the effect under treatment the effect under control.}

The policymaker may also consider $|D(\pi)|$ or $|C(\pi)|$ as measures of UnFairness, in which case the policymaker treats the two groups symmetrically, whose regret bounds are discussed in Appendix \ref{sec:absolute}. 
One last example is based on the notion of incentive-compatibility, motivated by discussion in \cite{narita2021incorporating}.

\begin{defn}[Incentive compatibility] \label{defn:ic} Incentive compatibility is defined as 
$$
\small 
\begin{aligned} 
\mathcal{I}(\pi) = I_1(\pi) + I_0(\pi), \quad I_s(\pi) = \mathbb{E}\Big[\pi(X,1 - s)(Y(1) - Y(0)) | S = s\Big] - \mathbb{E}\Big[\pi(X,s)(Y(1) - Y(0)) | S = s\Big] 
\end{aligned} 
$$ 
with estimator $\widehat{\mathcal{I}}(\pi) = \hat{I}_1(\pi) +  \hat{I}_0(\pi)$, $\hat{I}_s(\pi) = \frac{1}{n} \sum_{i=1}^n  (\hat{\Gamma}_{1, s,i} - \hat{\Gamma}_{0,s,i})\pi(X_i,1 - s) - \widehat{W}_s(\pi)$.  
\end{defn} 
Here $I_s(\pi)$ captures fairness based on the incentive of an individual in revealing her sensitive attribute: $I_s(\pi)$ is positive if the welfare of an individual generated from reporting her sensitive attribute incorrectly is larger than the welfare obtained if she reported it correctly. 
 Additional notions, such as predictive parity, can also be considered and omitted for the sake of brevity, see Appendix \ref{sec:a25} for details. 
For each of the three definitions above, UnFairness linear in $\pi$, and hence optimization can be performed via MIQP.

\subsubsection{Upper and Lower Bounds: Rate Characterization}

In the following theorem, we discuss the rate of the regret-bound.

\begin{thm}[Regret bound] \label{thm:between_groups} Let Assumptions \ref{ass:unconf}, \ref{ass:moment}-\ref{ass:dr} hold. Let either \\ (i) $\mathrm{UnFairness}(\pi) = D(\pi)$, and $\widehat{\mathcal{V}}_n(\pi) = \widehat{D}(\pi)$, (ii) or $\mathrm{UnFairness}(\pi) = C(\pi)$, and $\widehat{\mathcal{V}}_n(\pi) = \widehat{C}(\pi)$, (iii) or $\mathrm{UnFairness}(\pi) = \mathcal{I}(\pi)$, and $\widehat{\mathcal{V}}_n(\pi) = \widehat{\mathcal{I}}(\pi)$. Then for some constants $0 < \underline{b}, c_0 < \infty$ independent of the sample size, for any $\gamma \in (0,1), \lambda \ge \underline{b} (\sqrt{v} + \sqrt{\log(2/\gamma)}), N = \sqrt{n}$,
with probability at least $1 - 2 \gamma$, 
$$
\mathrm{UnFairness}(\hat{\pi}_\lambda) - \inf_{\pi \in \Pi_{\textrm{\mbox{\tiny o}}}} \mathrm{UnFairness}(\pi) \le  c_0 \sqrt{\frac{  v }{n}} + c_0 \sqrt{\frac{\log(2/\gamma)}{n}}.   
$$ 
\end{thm} 

The proof is included in Appendix \ref{sec:a1}. 
Theorem \ref{thm:between_groups} characterizes the regret bound for three different notions of UnFairness. The bound scales at rate $1/\sqrt{n}$. Here $\mathcal{K}(\Pi, \gamma) = \sqrt{v} + \sqrt{\log(2/\gamma)}, \eta = 1/2$ in Theorem \ref{thm:2}. The lower bound depends, however, on the notion of unfairness. Below, we derive a lower bound for any data-dependent policy which achieves the same rate for the predictive disparity.

\begin{thm}[Lower bound] \label{thm:lower_bound} Let $\Pi$ be such that $\pi(x,s)$ is constant in its last argument $s$ for all $x \in \mathcal{X}, \pi \in \Pi$, and with finite VC-dimension $v \ge 3$. Let $\mathrm{UnFairness}(\pi) = C(\pi)$, and $\widehat{\mathcal{V}}_n(\pi) = \widehat{C}(\pi)$. Let $\mathcal{U}$ be the set of distributions of $(X,S)$ and $\mathcal{P}(X,S) = \{ P_{Y,D | (X,S)}: \text{such that } |Y| < M \text{ a.s.}, \text{ and } P(D =1 |X, S) \in (\delta, 1 - \delta) \}$. 
 Then, there exists a distribution $P_{X,S,Y,D} = P_{X,S} P_{Y,D|X,S}$ with $P_{X,S} \in \mathcal{U}, P_{Y,D|X,S} \in \mathcal{P}(X,S)$, such that for every rule $\pi_n \in \Pi_{\textrm{\mbox{\tiny o}}}$ based upon $(X_1,  S_1, Y_1, D_1), \cdots, (X_n, S_n, Y_n, D_n)$, for finite constants constant $0 < c_0, \bar{C} < \infty$ independent of $n$,  and any $\gamma \in (0, 1/4)$,  $n \ge \max\{\bar{C} \log(1/(4 \gamma)), v-1\}$, with probability at least $\gamma$
$$
\begin{aligned}
\mathrm{UnFairness}(\pi_n) -  \inf_{\pi \in \Pi_{\textrm{\mbox{\tiny o}}}} \mathrm{UnFairness}(\pi) \ge \sqrt{\frac{c_0 \log(\frac{1}{4 \gamma})}{n}}. 
\end{aligned} 
$$  
\end{thm} 

The proof is contained in Appendix \ref{sec:a1}, and, to our knowledge, it is the first result of this type for fair and Pareto optimal policies. 
The lower bound states that we can find a distribution and some positive (non-vanishing) probability $\gamma$ such that any data-dependent policy $\pi_n$ achieves a regret which scales to zero at a rate no faster than $1/\sqrt{n}$. Observe that a direct corollary of such result is that the rate of the lower bound is also achieved in expectation. The condition imposes a restriction on the set of policies $\Pi$: $\Pi$ does not contain policies that use the sensitive attribute as a covariate. This class of policies occurs if anti-discriminatory laws are enforced and incorporated over the set $\Pi$. The lower bound applies to prediction disparity, and we leave to future research a more comprehensive study of lower bounds under generic notions of fairness. The derivation modifies arguments in the empirical risk minimization literature \citep{devroye2013probabilistic}, due to the dependence of the objective function with the \textit{conditional} probability of treatment.



 Throughout this section, we have considered \textit{distributional} notions of fairness, i.e., they depend on distributional statements relative to the sensitive attribute, often used in the literature \citep{kasy2020, donini2018empirical, narita2021incorporating}.
 \textit{Counterfactual} notions depend instead on counterfactual statements relative to the sensitive attribute \citep{kilbertus2017avoiding}. We discuss one counterfactual notion in Section \ref{sec:counterfactual}.


\subsection{Computational Complexity} \label{sec:complexity} 

We conclude this section with a discussion on the computational complexity of the procedure. 
First, consider the case where $|\mathcal{X}| < \infty$, defines a \textit{finite} number of strata, as often assumed in economic applications \citep[e.g.,][]{manski2004}. Since $|\mathcal{X}| < \infty$, let $X$ be a set of dummies $X \in \{0,1\}^p, \sum_j X^{(j)} = 1$, and $\pi(X, S) = X\beta_S, \beta_0, \beta_1 \in [0,1]^p$, where $\beta$ defines the treatment probability for a given individual type. Note that the result continues to hold if we require that $\pi$ assigns treatments on a finite number of strata but $X$ is continuous. 


\begin{prop} \label{prop:polynomial} Let $|\mathcal{X}| < \infty$,  and $\pi(X, S) = X \beta_S, \beta_0, \beta_1, \in [0,1]^p$. Suppose that you can write $\mathcal{V}_n(\pi) = g(\sum_{i=1}^n \hat{F}_i \pi(X_i, S_i))$, for some arbitrary $\hat{F}_i$ and either $g(x) = x$ or $g(x) = |x|$. Let $||\hat{F}_i||_{\infty}, ||\hat{\Gamma}_i||_{\infty} < B < \infty$ be uniformly bounded.  Then there exists an algorithm which solves Equation \eqref{eqn:opt1}, with running time $\mathcal{O}(\sqrt{n} p^{\omega})$, for some finite $\omega < \infty$. 
\end{prop} 

The proof is contained in Appendix \ref{sec:proof_complexity}. Proposition \ref{prop:polynomial} states that there exists an exact algorithm with a running time that is polynomial in the number of types and which scales at a rate $\sqrt{n}$ in the number of observations. The exponent $\omega$ depends on the algorithm.\footnote{Classic examples include Vaidya's and Karmarkar's  algorithm \citep{vaidya1990algorithm, karmarkar1984new}.} The intuition is that we can represent the optimization problem in Equation \eqref{eqn:opt1} as a sequence of $\sqrt{n}$ many linear programs (see Appendix \ref{sec:proof_complexity}). 

For generic function classes, the sequence of linear programs described above is not possible. Two examples are the maximum score with continuous variables and the optimal classification trees. These methods, however, admit a mixed-integer linear program (MILP), which can be solved \textit{exactly} with, e.g., Branch and Bound (BB) algorithms \citep{wolsey1999integer}. In generic settings, MILP is known to be NP-hard in the \textit{worst-case} scenario, hence infeasible for large samples.
Here, we characterize properties when an early termination is imposed. Namely, with an early termination, the BB algorithm reports an upper bound on the distance from the best objective (gap), informative for the regret.  



\begin{prop} \label{prop:early_termination} 
Define $\bar{W}_{j,n}^\delta$ the value function which maximizes $\alpha_j \hat{W}_0(\pi) + (1 - \alpha_j) \hat{W}_1(\pi)$ with an early stopping criterion which stops when the estimated bound on the gap is $\delta$. Define $\hat{\pi}_{\lambda}^{\delta}$ the solution obtained after running the optimization algorithm in Equation \eqref{eqn:opt1} which stops whenever the estimated gap is $\delta$, and which replaces $\bar{W}_{j,n}$ in Constraint (B) with $\bar{W}_{j,n}^{\delta}$. Let the conditions in Theorem \ref{thm:between_groups} hold.
For some $0 < \underline{b}, c_0 < \infty$ independent of $n$, for any $\gamma \in (0,1), \lambda \ge \underline{b} (\sqrt{v} + \sqrt{\log(2/\gamma)}), N = \sqrt{n}$,
with probability at least $1 - 2 \gamma$, 
$$
\mathrm{UnFairness}(\hat{\pi}_\lambda^\delta) - \inf_{\pi \in \Pi_{\textrm{\mbox{\tiny o}}}} \mathrm{UnFairness}(\pi) \le  c_0 \sqrt{\frac{  v }{n}} + c_0 \sqrt{\frac{\log(2/\gamma)}{n}} + \delta.   
$$ 
\end{prop} 

The proof is in Appendix \ref{sec:proof_complexity}. Proposition \ref{prop:early_termination} shows that the effect of early termination with gap $\delta$ is informative (and can be chosen appropriately) for the policy regret.

\section{Counterfactual UnFairness} \label{sec:counterfactual}

This section is of independent interest, and it discusses a novel notion of UnFairness which connects the literature on causal fairness \citep{kilbertus2017avoiding} and the economic literature on envy-freeness \citep{varian1976two}. The notion is based on counterfactual statements relative to the \textit{sensitive} attribute. We sketch the main intuition here and defer details to Appendix \ref{sec:counterfactual_app}. This section defines $Y(d,s), X(s)$ the potential outcome and covariates as functions of the sensitive attribute $s$.  
The following causal model is considered. 

\begin{ass} \label{ass:unconfounded2} 
Let (A) $Y(d,s) \perp (D, S) |  X(s)$, (B) $X(s) \perp S$. 
\end{ass}

Assumption \ref{ass:unconfounded2} is required for estimation with a counterfactual fairness and not for notions of fairness discussed in the previous sections. Condition (A) and (B) in Assumption \ref{ass:unconfounded2} state that the sensitive attribute is independent of potential outcomes and covariates, while it allows for the dependence of \textit{observed} covariates and outcomes with the sensitive attribute. Indexing potential outcomes and covariates captures this dependence by the sensitive attribute. Dependence can also occur through \textit{unobserved} characteristics, which are dependent on both outcomes and sensitive attributes as long as observables do not \textit{causally} affect the sensitive attribute. See 
 Figure \ref{fig:confounding} for an illustration. Assumption \ref{ass:unconfounded2} holds when sensitive attributes do not have causal parents \citep[e.g.,][]{kilbertus2017avoiding}.\footnote{Whenever $S_i$ has not \textit{causal parents}, such as, for instance, age and gender in the application of interest, Assumption \ref{ass:unconf} trivially holds. The case of race represents instead an exception under which Assumption \ref{ass:unconf} may fail since an individual's race depends on parents' characteristics.  Assumption \ref{ass:unconf} can be stated after conditioning on baseline characteristics such as parents' observable characteristics to accommodate this latter case.}

Let the conditional welfare, for the policy function being assigned to the opposite attribute, i.e., the effect of $\pi(x, s_1)$, on the group $s_2$, conditional on covariates, be








\definecolor{aero}{rgb}{0.49,0.73,0.91}
\definecolor{airsuperiorityblue}{rgb}{0.45,0.63,0.76}
\definecolor{babyblueeyes}{rgb}{0.63,0.79,0.95}
\definecolor{beaublue}{rgb}{0.74,0.83,0.9}
\definecolor{glaucous}{rgb}{0.38,0.51,0.71}

\begin{figure} 
\centering
\begin{subfigure}[b]{0.4\textwidth}
\begin {tikzpicture}[-latex ,auto ,node distance =3 cm and 2cm ,on grid ,
semithick ,
state/.style ={ circle ,top color =glaucous!80 , bottom color =glaucous!80 ,
draw,glaucous!80 , text=black, minimum width =0.01 cm}]
\node[state] (C) at (0,0) {$S_i$};
\node[state] (A) [above right=of C] {$X_i$};
\node[state] (B) [below right=of A] {$D_i$};
\node[state] (D) [above right=of B] {$Y_i$};
\node[state] (E) at (6,1) {$U_i$};

\path (C) edge [bend left = 01,color=black!70] node[below =0.015 cm] {} (A);
\path (C) edge [bend right = 01,color=black!70] node[above =0.015 cm] {} (B);
\path (C) edge [bend left = 01,color=black!70] node[above =0.015 cm] {} (D);


\path (A) edge [bend right = 01, color=black!70] node[above =0.015 cm] {} (B);
\path (A) edge [bend left = 01, color=black!70] node[above =0.015 cm] {} (D);

\path (B) edge [bend left = 01, color= black!70] node[above =0.015 cm] {} (D);
\path (E) edge [bend left = 01, color= black!70] node[above =0.015 cm] {} (C);
\path (E) edge [bend left = 01, color= black!70] node[above =0.015 cm] {} (D);

\end{tikzpicture}
\end{subfigure}
\hfill 
\begin{subfigure}[b]{0.4\textwidth}
\definecolor{aero}{rgb}{0.49,0.73,0.91}
\definecolor{airsuperiorityblue}{rgb}{0.45,0.63,0.76}
\definecolor{babyblueeyes}{rgb}{0.63,0.79,0.95}
\definecolor{beaublue}{rgb}{0.74,0.83,0.9}
\definecolor{glaucous}{rgb}{0.38,0.51,0.71}

\begin {tikzpicture}[-latex ,auto ,node distance =3 cm and 2cm ,on grid ,
semithick ,
state/.style ={ circle ,top color =glaucous!80 , bottom color =glaucous!80 ,
draw,glaucous!80 , text=black, minimum width =0.01 cm}]
\node[state] (C) at (0,0) {$S_i$};
\node[state] (A) [above right=of C] {$X_i$};
\node[state] (B) [below right=of A] {$D_i$};
\node[state] (D) [above right=of B] {$Y_i$};
\node[state] (E) at (6,1) {$U_i$};

\path (C) edge [bend left = 01,color=black!70] node[below =0.015 cm] {} (A);
\path (C) edge [bend right = 01,color=black!70] node[above =0.015 cm] {} (B);
\path (C) edge [bend left = 01,color=black!70] node[above =0.015 cm] {} (D);


\path (A) edge [bend right = 01, color=black!70] node[above =0.015 cm] {} (B);
\path (A) edge [bend left = 01, color=black!70] node[above =0.015 cm] {} (D);

\path (B) edge [bend left = 01, color= black!70] node[above =0.015 cm] {} (D);
\path (E) edge [bend left = 01, color= black!70] node[above =0.015 cm] {} (C);
\path (E) edge [bend left = 01, color= black!70] node[above =0.015 cm] {} (A);

\end{tikzpicture}
  \end{subfigure}\caption{Directed acyclical graphs under which Assumption \ref{ass:unconf} does not hold in the presence of the confounder $U_i$, and holds in the absence of $U_i$. } \label{fig:confounding} 
\end{figure}

\vspace{-8mm}
\begin{align}
V_{\pi(x,s_1)}(x,s_2) &= \mathbb{E}\Big[\pi(x, s_1) Y_i(1, s_2) + (1 - \pi(x, s_1)) Y_i(0, s_2) \Big| X_i(s_2) = x\Big].
\end{align}

Envy refers to the concept that ``an allocation is equitable if and only if no agent prefers another agent’s bundle to his own'' \citep{varian1976two}. We say that the agent with attribute $s_2$ \textit{envies} the agent with attribute $s_1$, if her welfare (on the right-hand side of Equation \ref{eqn:welf}) exceeds the welfare she would have received had her covariate and policy been assigned  the opposite attribute (left-hand side of Equation \ref{eqn:welf}), namely 
\begin{equation} \label{eqn:welf}
\mathbb{E}_{X(s_1)}\Bigl[  V_{\pi(X(s_1), s_1)}\Big(X(s_1),s_2\Big) \Bigl] > 
\mathbb{E}_{X(s_2)}\Bigl[  V_{\pi(X(s_2), s_2)}\Big(X(s_2),s_2\Big) \Bigl]. 
\end{equation} 

We then measure the unfairness towards an individual with attribute $s_2$ as 
  \begin{equation} \label{eqn:ll}
  \begin{aligned} 
\mathcal{A}(s_1,s_2;\pi) =  \mathbb{E}_{X(s_1)}\Bigl[V_{\pi(X(s_1),s_1)} (X(s_1),s_2)  \Bigl]  - \mathbb{E}_{X(s_2)}\Bigl[  V_{\pi(X(s_2), s_2)}\Big(X(s_2),s_2\Big) \Bigl].
\end{aligned} 
\end{equation}
Whenever we aim not to discriminate in either direction, we take the sum of the effects $\mathcal{A}(s_1,s_2; \pi)$ and $\mathcal{A}(s_2,s_1; \pi)$.\footnote{Such an approach builds on the notion of ``social envy'' discussed in \cite{feldman1974fairness}.}
Equation \eqref{eqn:ll} connects to previous notions of \textit{counterfactual fairness} \citep{ kilbertus2017avoiding}, while, differently from previous references, (i) we provide formal justification to fairness using an envy-freeness argument; (ii) we construct the definition of fairness based on \textit{distributional} impact of the treatment allocation rule on the welfare. 
It is complementary to \cite{kusner2019making}, 
who compare the policy effects over individuals with the opposite sensitive attribute, lacking an envy-based justification. 
 On the other hand, a shortcoming of the above notion are that, similarly to the above references, it does not capture notions of incentive-compatibility differently from Definitions \ref{defn:ic}, discussed in the previous section. A second shortcoming is that it requires parametric estimation for minimax rate of convergence.  Namely, in Appendix 
\ref{sec:counterfactual_app} we show that for a suitable choice of the estimator of $\mathcal{A}(\cdot)$, with probability at least $1 - 2\gamma$
$$
\mathrm{UnFairness}(\hat{\pi}) - \inf_{\pi \in \Pi_{\textrm{\mbox{\tiny o}}}} \mathrm{UnFairness}(\pi) = \mathcal{O}\Big(\sqrt{\frac{1 }{ n^{2\zeta}} } + \sqrt{\frac{\log(2/\gamma)}{n}}\Big).   
$$ 
where here $\zeta$ denote the rate of convergence of the \textit{conditional mean function} (see Appendix \ref{sec:counterfactual_app}). The convergence rate is of order $n^{-1/2}$ for a parametric estimators and slower for non-parametric estimators compared to the notions of UnFairness discussed in Section \ref{sec:4}. The slower convergence rate is because counterfactual envy-freeness requires extrapolation on a different population. It opens new questions on the trade-offs between counterfactual and predictive notions of fairness.

\section{Empirical Application and Numerical Study} \label{sec:app}

We now discuss the empirical application. This section designs a policy that assigns students to entrepreneurial programs, while imposing fairness on gender. 
We use data that originated from \cite{lyons2017impact}. The paper studies the effect of an entrepreneurship training and incubation program for undergraduate students in  North  America on subsequent entrepreneurial activity. We have in total $335$ observations, of which $53\%$ treated and the remaining under control, and $26\%$ of applicants are women.\footnote{Data is available at \url{https://www.openicpsr.org/openicpsr/project/113492/version/V1/view}.} 
The population of interest is the pool of final applicants. We construct a targeting rule that assigns the award to the finalist based on the applicant's observable characteristics.  We maximize subsequent entrepreneurial activity, which is captured using a dummy variable, indicating whether the participant worked in the startup once the program ended. The study is a quasi-experiment, and, as noted in \cite{lyons2018does} the focus on the pool of final applicants mitigates selection on unobservables. Similarly to \cite{lyons2018does} we control for residual confounding through individual level observable characteristics and an observable quality score of the final applicant. 
Estimation of the nuisance functions is through penalized regression and discussed in Appendix \ref{sec:a3a}.

We consider three notions of UnFairness: (i)  \textit{counterfactual envy}; (ii)  (ii) \textit{predictive disparity}, which minimizes the probability of treatment between the two groups as in Definition \ref{defn:pred}; (iii)  \textit{predictive disparity} with absolute value (i.e. it denotes the absolute difference between the probability of treatment between the two groups).  While (ii) and (iii) do not impose conditions on the distribution of the sensitive attribute, counterfactual envy ((i)) assumes unconfoundedness also of the sensitive attribute. Such a condition is equivalent to assuming that the decision to change gender is exogenous. The reader may refer to Figure \ref{fig:confounding} for a graphical illustration. In case of failure of such assumption, the reader should refer to results for (ii) and (iii) only.

We consider \textit{linear} decision rules, given their large use in economics \citep{manski1975maximum}\footnote{This is estimated solving Equation \eqref{eqn:opt1} with a small slackness parameter of order $10^{-6}$. The reader may refer to Appendix \ref{sec:a24} for details. }
\begin{equation}  \label{eqn:functional_form}
\begin{aligned} 
&\Pi = \Big\{\pi(x, \mathrm{fem}) = 1\Big\{\beta_0 + \beta_1 \mathrm{fem} + x^\top \phi \ge 0 \Big\}, \quad (\beta_0, \beta_1, \phi) \in \mathcal{B}\Big\}.
\end{aligned} 
\end{equation}

We allow covariates $x$ to be either (1) the years to graduation, years of entrepreneurship, the region of the start-up, the major, the school rank, or (2) the score assigned to the candidate by the interviewer and the school rank. We refer to these two cases respectively as \textit{Case 1} and \textit{Case 2}. We consider in-sample capacity constraints imposed on the function class with at most 150 individuals selected for the treatment.\footnote{The validity of the in-sample capacity constraints follows from a uniform concentration argument of the capacity constraint around its expectation.  }



\begin{table*}[h]\centering
\caption{Empirical application. The first two columns report the welfare improvement plus the baseline value. The last column reports the importance weights assigned by the method to the welfare of female students. FTP Envy refers to the Fair Targeting rule that minimizes envy-freeness unfairness; FTP Predictive Disp (Definition \ref{defn:pred}) refers to the Pareto allocation that minimizes the difference in probability of treatment (Abs indicate in absolute value); 
Welfare Max. 1 denotes the method that maximizes the empirical welfare considering $\Pi_1$, and similarly Welfare Max. 2, 3 for the function classes, respectively 
$\Pi_2, \Pi_3$.  }    \label{tab:summaries} 
\ra{1.3}
\begin{tabular}{@{}lrrrcrrrcrrr@{}}\toprule
& \multicolumn{2}{c}{Welfare Female} & \ & \multicolumn{2}{c}{Welfare Male} & \ & \multicolumn{2}{c}{Importance Weight} \\
\cmidrule{2-3} \cmidrule{5-6}  \cmidrule{8-9}  
&Case 1& Case 2& &Case 1& Case 2 &&Case 1& Case 2  \\ \midrule
Fair Envy & $0.376$ & $0.372$ & & $0.272$ & $0.195$ & & $0.384$ & $0.487$ \\ 
FTP Pred & $0.432$ & $0.374$ & & $0.224$ & $0.180$ & & $0.847$ & $0.924$ \\ 
FTP Pred Abs & $0.433$ & $0.351$ & & $0.208$ & $0.235$ & & $0.924$ & $0.487$ \\ 
Welfare Max. 1  & $0.376$ & $0.351$ & & $0.272$ & $0.235$ & & $0.266$ & $0.266$ \\ 
Welfare Max. 2 & $0.288$ & $0.307$ & & $0.285$ & $0.238$ & & $0.266$ & $0.266$ \\ 
Welfare Max. 3 & $0.331$ & $0.307$ & & $0.265$ & $0.238$ & & $0.266$ & $0.266$ \\ 
\bottomrule
\end{tabular}
\end{table*}

We compare the proposed methodology to the method that maximizes the empirical welfare with the double robust score \citep{athey2017efficient}. We consider three nested function classes for the welfare maximization method. The first does not impose any restriction except for the functional form in Equation \eqref{eqn:functional_form}. The second, imposes that $\beta_1 = 0$. The third class imposes that $\beta_1 = 0$ \textit{and} that the average effect of the policy on females is at least as large as the one on males. The function classes are
$$
\small 
\begin{aligned} 
\Pi_1 = &\Pi, \quad \Pi_2 = \Big\{\pi(x) = 1\Big\{\beta_0  + x^\top \phi \ge 0 \Big\} \Big\}, \\
\Pi_3 = &\Big\{\pi(x) = 1\Big\{\beta_0  + x^\top \phi \ge 0 \Big\}, \quad \mathbb{E}_n\Big[(Y_i(1) - Y_i(0)) \pi(X_i)\Big| S = 1\Big] \ge \mathbb{E}_n\Big[(Y_i(1) - Y_i(0)) \pi(X_i)\Big| S = 0\Big] \Big\},  
\end{aligned} 
$$ 
where $\mathbb{E}_n[\cdot]$ denote the empirical expectation, estimated using the doubly-robust method.

Figure \ref{fig:pareto_front1} reports the Pareto frontier over each function class.\footnote{The value functions over the Pareto frontier can be exactly recovered as follows: we solve $2$ optimization problems for each $\alpha_j$, $j \in \{1, ..., N\}$. For each of these problems, we impose constraints on the welfare of one of the two groups being larger than the other and vice-versa; we then select the subset of solutions that are not Pareto dominated by the other, and we plot the corresponding welfares in the figure.} 
The figure shows that restricting the function class leads to Pareto-dominated allocations. This outlines the limitations of maximizing welfare under fairness constraints: such constraints can be harmful for both groups. Instead, the proposed method enforces Pareto optimality in the least constrained environment (red line) and selects the policy based on fairness considerations.


\begin{figure}
\centering
\includegraphics[scale=0.37]{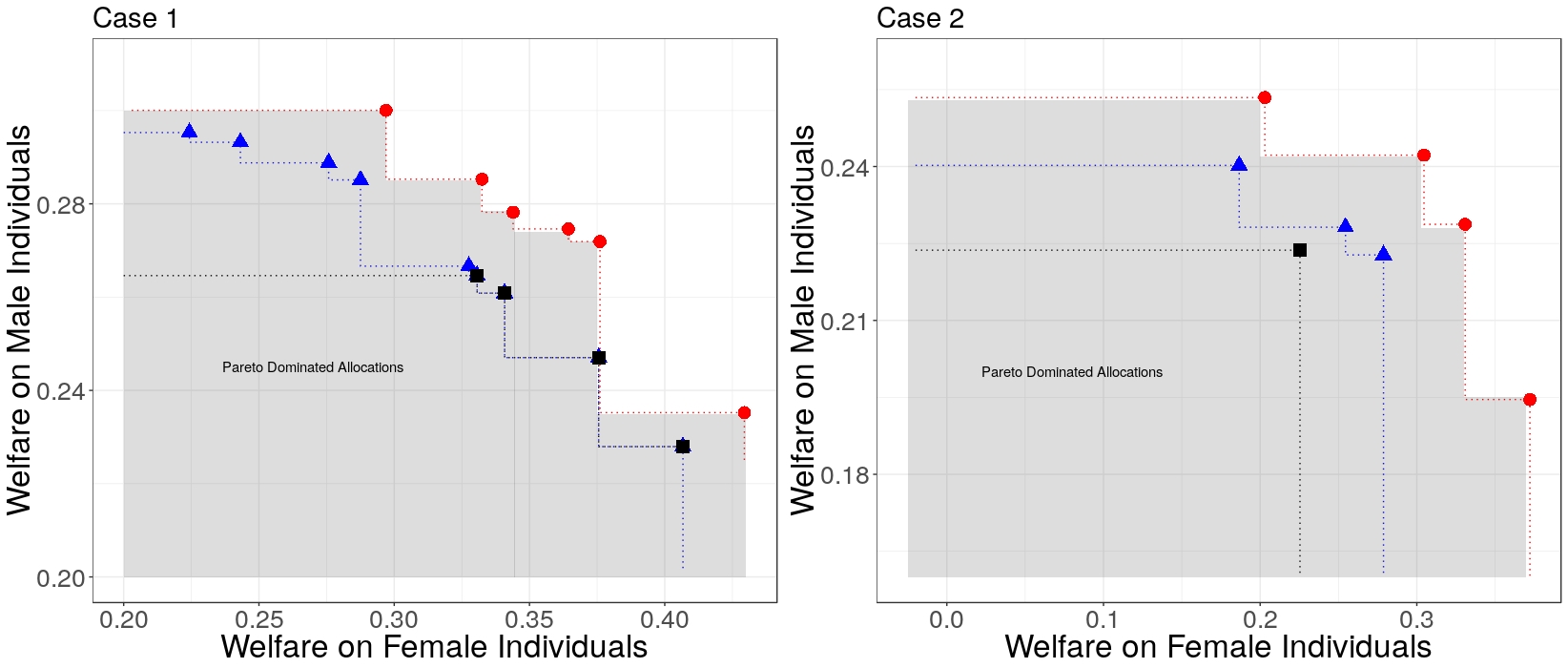}
\caption{(Discretized) Pareto frontier under deterministic linear policy rule estimated through MIQP. Dots denote Pareto optimal allocations. Red dots (circle) correspond to $\Pi_1$, blue dots (triangle) to $\Pi_2$ and black dots (square) to $\Pi_3$.
 }
\label{fig:pareto_front1} 
\end{figure}

In Table \ref{tab:summaries} we collect results\footnote{In computations, the competitors (Welfare Maximization) achieves the global optimum (dual gap equal to zero). For the proposed method we impose a maximum time limit on the MIQP.}   of the welfare on female and male students, as well as the relative importance weight assigned to each group for methods that maximize different UnFairness measures.
 In the table, we observe that minimizing Envy and Predictive Disparity leads to (weakly) larger welfare effects on the minority group. Envy leads to comparable results to welfare maximization for $Case 1$ due to the discreteness of the frontier.\footnote{Even if the weight $\alpha$ is larger for FTP Envy and FTP Parity Abs in $Case$ $1$ and $2$ respectively, this does not lead to a different result than Welfare Max. 1 due to the discreteness of the frontier.} We observe an increase in the welfare of female students when minimizing the \textit{absolute} difference between probabilities of treatments for $Case$ $1$, and comparable results to the welfare maximization method for $Case$ $2$. The table shows that the proposed method finds importance weights assigned to each group solely based on the notion of fairness provided, without requiring any prior specification of relative weights assigned to each group. The method that maximizes the empirical welfare instead assigns to the sensitive group the importance weight equal to its corresponding probability, small for minorities. In two settings only, the results coincide with the proposed method due to the discreteness of the frontier.

Figure \ref{fig:pareto_front2} reports the unfairness level for different sets of covariates, with unfairness measured as the difference in the probability of treatments between the two groups. Overall, Figure \ref{fig:pareto_front2} shows that the level of the unfairness of the proposed method is uniformly smaller than the unfairness achieved by maximizing welfare, consistently with results in Section \ref{sec:2}.   

Finally, we compare also with probabilistic decision rules, which are allowed in our framework. Figure \ref{fig:pareto_front2} also collects result also for a probabilistic policy function (in green) which is a super-set of $\Pi$ in Equation \eqref{eqn:functional_form} and assigns different probabilities of treatments to groups below and above the hyperplane in Equation \eqref{eqn:functional_form} (see Appendix \ref{sec:a24}).\footnote{Formally, the function class is $\Big\{\pi_\beta(X, S) = p_1 1\{X_i^\top \beta + S \beta_0  > 0\} + p_0 1\{X^\top \beta + S \beta_0 \le 0\}, p_1, p_0 \in [0,1], \beta \in \mathcal{B} \Big\}$.}  Results are mostly comparable across probabilistic or deterministic decisions. However, we find that a probabilistic decision enlarges the set of Pareto allocations in Appendix \ref{sec:a3a}. 



\begin{figure}
\centering
\includegraphics[scale=0.37]{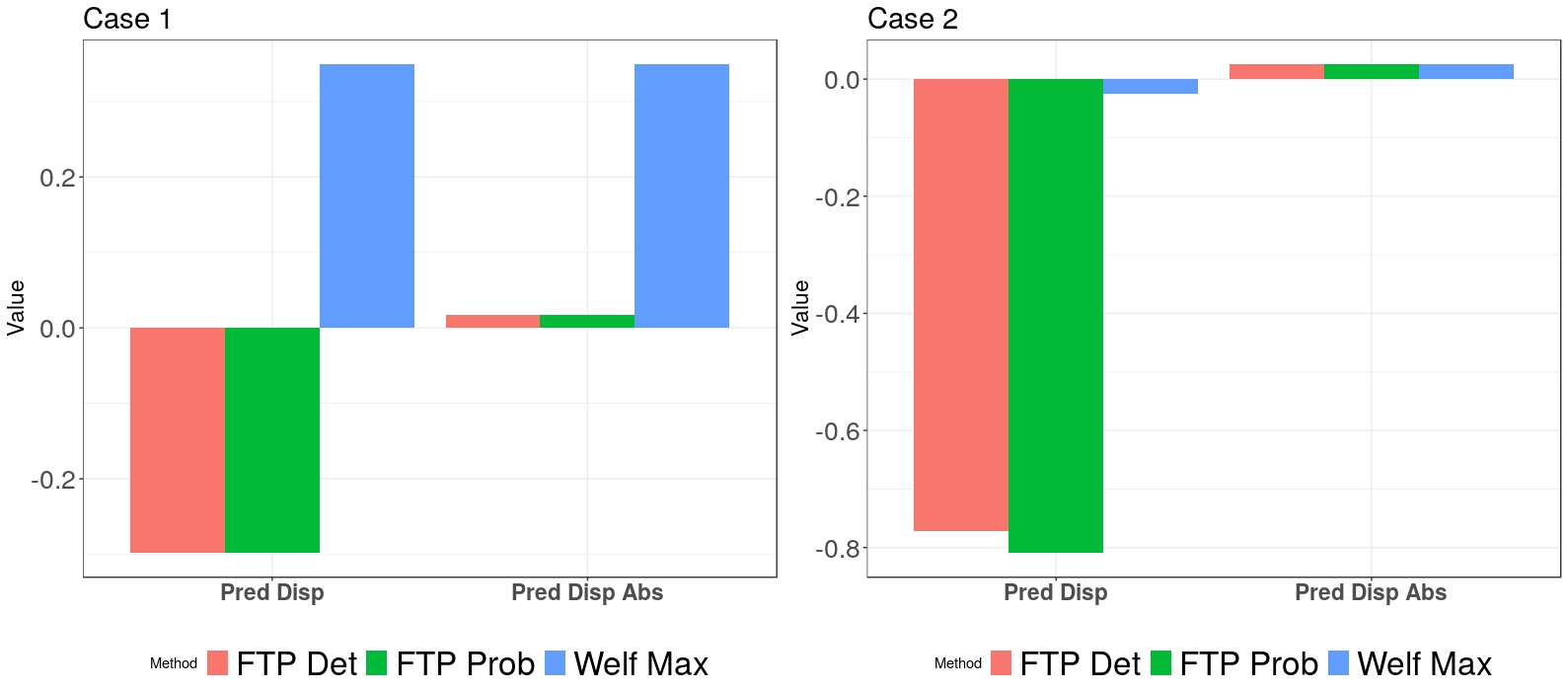}
\caption{Empirical application. Unfairness level of the Fair Policy Targeting method  with a deterministic allocation rule (in red), with a probabilistic decision rule (in green), and of the welfare maximization method (in blue). Pred disp refers to Definition \ref{defn:pred} and Pred disp abs to Definition \ref{defn:pred} in absolute value. Smaller values indicate smaller UnFairness. } 
\label{fig:pareto_front2} 
\end{figure}



\subsection{A Calibrated Experiment} \label{sec:numerics}

Next, we conduct a calibrated experiment.  We run the simulations calibrated to the same estimated model in the empirical application (using data from \citealt{lyons2017impact}). Covariates and sensitive attributes are drawn with replacement from the empirical distribution. 
Formally, we draw 
$
S_i \sim_{i.i.d.} \mathrm{Bern}(\hat{p}_1), 
$ 
where $\hat{p}_1$ is the probability of being female. We draw covariates $X | S = 1$ from the females' empirical distribution and similarly $X | S = 0$ for male applicants. We draw $D | X, S \sim \mathrm{Bern}(\hat{e}(X,S))$, and 
$
Y(d) = \hat{m}_{d,S}(X) + \varepsilon, \varepsilon \sim \mathcal{N}(0,1),  
$
where
$\hat{e}, \hat{m}$ are the estimated conditional mean and propensity score as in the application. We consider unfairness as prediction disparity in Definition \ref{defn:pred}. 

We consider three classes of policy functions: (i) probabilistic linear rule $x_1^\top \beta$, were $x_1$ is a set of binary variables
; (ii) maximum-score $\pi(x_2) = 1\{x_2^\top \beta > 0\}$;  (iii) classification tree with depth equal to two. 
For each method, we compute results over three variables  (minority, whether the student will graduate in more the one year, whether the average score exceeds the median score). We also include the average score as a continuous variable in the tree and the maximum score. We impose that the number of treated individuals does not exceed $150$ individuals across each design and report welfare per share of treated individuals.\footnote{Welfare is scaled by the unconditional treatment probability since the number of treated units is fixed.} We estimate the probabilistic rule with a linear program, the maximum score with a mixed-integer linear program, and the classification tree via exhaustive search. 
For the classification tree, we fix the number of possible splits to be four at equally spaced quantiles of each covariate distribution.\footnote{The choice of the exhaustive search follows in spirit to the discussion in \cite{zhou2018offline}, with differences due to the presence of multiple objectives and constraints here. Four splits facilitate computations. }
We run one-hundred replications, and over each replication, we correctly estimate the nuisance functions from the sampled observations.

In Figure \ref{fig:1_sim} we report the running time (in seconds) of different function classes, with the maximum score having two different stopping times (see also Appendix \ref{sec:a3b} for more results). Consistently with Section \ref{sec:complexity}, the complexity of the linear rule scales much slower than the one of the maximum score. Also, the optimal tree is much faster than the maximum score, even for a larger sample size. The maximum score presents a relatively fast growth in terms of running time, which, however, is still feasible to handle for $n = 600$.
Figure \ref{fig:1_sim} also shows that a more stringent stopping time on the maximum score does not affect its performance either in terms of fairness or welfare. This is because by passing as a starting point an ``educated guess", most of the remaining optimization time is to discard dominated solutions. We obtain such a guess by taking the best solution estimated in the first step run to estimate the Pareto frontier. Finally, Figure \ref{fig:1_sim} shows that different function classes mostly lead to non-dominated males and females welfare comparisons.

Table \ref{tab:mse2} contrasts the unfairness and welfare of five different alternative approaches. Each competitor uses the same function class and estimation procedure as the proposed method. The first two competitors maximize a weighted average of female and male welfare with weights either $\alpha = 1/2$ \citep[in the spirit of the planner's utility in][]{rambachan2020economic}, or the empirical average $\mathbb{E}_n[S]$ as a welfare maximization problem. The third approach maximizes welfare with constraints of the form $
\mathrm{UnFairness}_n(\pi) \le \kappa/n, 
$ where we choose $\kappa \in\{10, 1\}$ (Constrained Max and Constrained Max2, respectively).\footnote{This is in the spirit of fairness constraints  \citep[e.g.,][]{nabi2019learning}, with constraints on statistical parity.} The fourth approach maximizes welfare with constraints on disparate impact as in Definition \ref{defn:welfare} \citep[in the spirit of optimization in][]{donini2018empirical}. Interestingly, while the stricter constraint reduces the gap in males' and females' welfare for the competitor Disparate Impact, such a gap is large due to the estimation error of the constraint (Appendix \ref{sec:a3b}, Figure \ref{fig:impact} presents details). 
We observe that the proposed method leads to the lowest UnFairness, and it is not Pareto-dominated. Our method favors the minority group, hence leading to larger welfare for female students.   Appendix \ref{sec:a3b} provides results for a smaller sample size.


\begin{figure}[!ht]
\centering 
\includegraphics[scale=0.6]{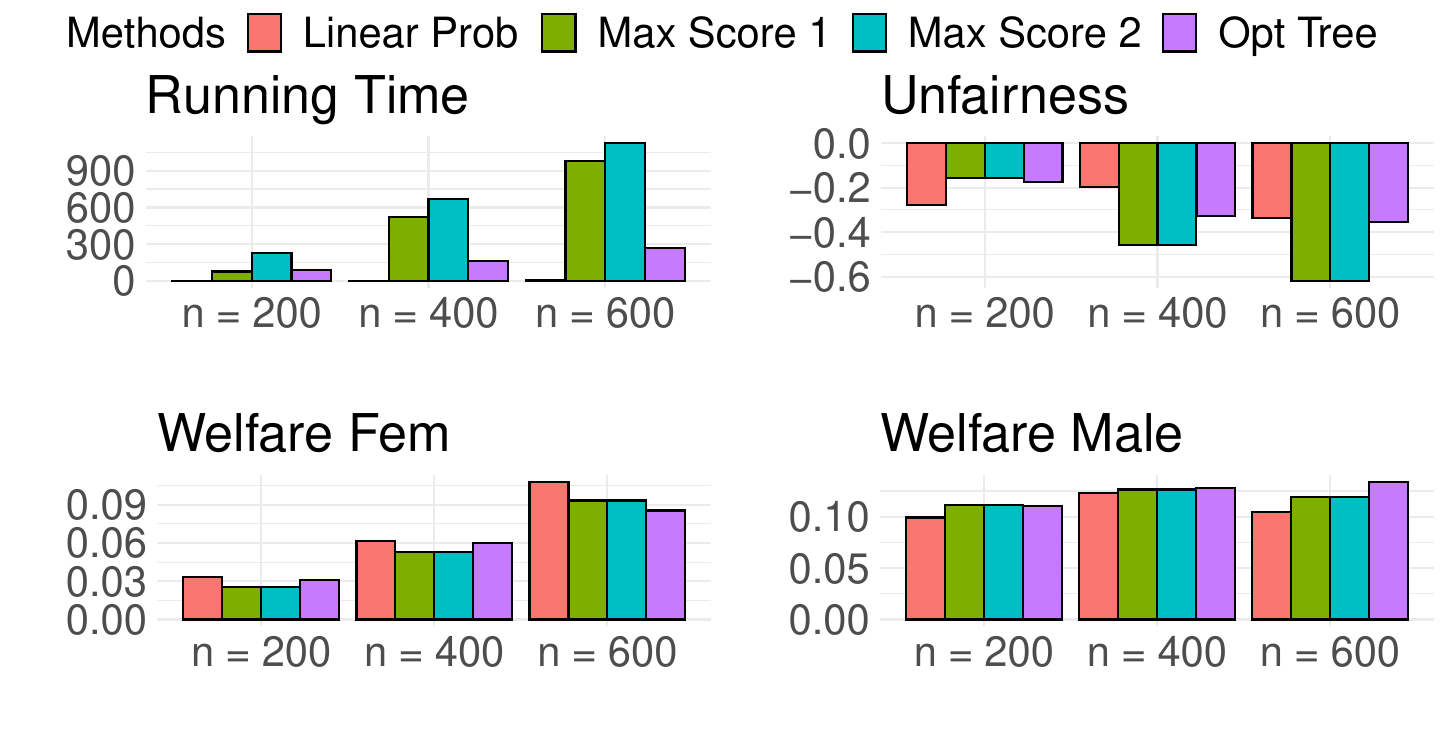}
\caption{Running time, UnFairness and welfare for $p = 4$.  Here, Linear Prob is a linear probability rule estimated via linear programming, maximum score is estimated with MILP and optimal tree via exhaustive search with depth two. The maximum score algorithm presents two different stopping times denoted as Max Score 1 and Max Score 2 (with fifty and two hundred seconds stopping after estimating the Pareto frontier). } 
\label{fig:1_sim}
\end{figure}



\begin{table*}[!ht]\centering
\caption{Statistical disparity (Unfairness), welfare of male ($W_0$) and female ($W_1$) participants of the proposed method (Fair Targeting) and of the alternative procedures in \textit{percentage points}. Weighted average maximizes a weighted average of females and males' welfare with weight $\alpha = 1/2$; Utilitarian average uses instead $\alpha = \mathbb{E}_n[S]$; Constrained Max maximizes welfare under fairness constrain and Disparate Impact maximizes welfare under constraints on disparate welfare impact between the two groups. $n = 600, p = 4$. The constraint is $\kappa = 10$ for the methods in the fourth and fifth row and $\kappa = 1$ for the last two rows. }    \label{tab:mse2} 
\begin{tabular}{@{}lrrrcrrrcrrr@{}}\toprule
& \multicolumn{3}{c}{Linear Rule} & \ & \multicolumn{3}{c}{Maximum Score} & \ & \multicolumn{3}{c}{Tree} \\
\cmidrule{2-4} \cmidrule{6-8}  \cmidrule{10-12}  
&\small UnFair&\small $W_0$ &\small $W_1$ && \small UnFair&\small  $W_0$  & \small $W_1$ && \small UnFair & \small $W_0$  & \small $W_1$  \\
 \Xhline{.8pt} 
 \cline{2-4}  \cline{6-8}  \cline{10-12} 
 \rowcolor{glaucous!60}  \small Fair Targeting & $-33.5$ & $10.4$ & $10.8$&& $-62.4$ & $12$ & $9.4$&&$-35.5$&13.4&8.5\\
\small Weighted Average  & $13.2$ & $14.7$ & $8.5$ && $14.3$ & $16.6$ & $7.7$&& $2.4$ & $15.2$ & $8.3$  \\
\small Utilitarian Average  & $25.6$ & $16.3$ & $6.4$ && $14.3$ & $16.6$ & $7.7$&& $30.9$ & $17.0$ & $6.3$  \\
\small Constrained Max & $-13.2$ & $14.4$ & $8.1$&&$-32.2$ & $14.2$ & $8.3$ &&$-13.8$ & $15.0$ & $7.1$\\
\small Disparate Impact & $3.2$ & $13.5$ & $7.9$&&$-8.2$ & $15.0$ & $8.8$ &&$6.8$ & $15.1$ & $7.9$\\
\small Constrained Max2 & $-16.3$ & $14.1$ & $8.3$&&$-34.5$ & $13.9$ & $8.4$ &&$-19.3$ & $14.7$ & $7.3$\\
\small Disparate Impact2 & $-3.8$ & $12.1$ & $8.4$&&$-13.7$ & $14.6$ & $9.3$ &&$0.6$ & $14.3$ & $8.1$\\
\bottomrule
\end{tabular}
\end{table*}

\section{Conclusion} \label{sec:conclusion}

  In this paper, we have introduced a novel method for estimating fair and optimal treatment allocation rules. We proposed a multi-objective decision problem, where the policymaker aims to select the least unfair policy in the set of Pareto optimal allocations. We discuss a set of theoretical guarantees on the estimated policy and provide an application. 
From a theoretical perspective, we open new questions on the trade-offs between predictive and causal notions of fairness and its corresponding regret bound.  
Counterfactual notions require extrapolation, hence possibly leading to a slower convergence rate. We leave to future research
a comprehensive study of properties of different notions of fairness in terms of their implied
regret. From a practical perspective, an interesting new direction is estimation with non-utilitarian within-group welfare measures. 
Finally, the decision problem considered aims to balance efficiency and fairness, and a study of such trade-offs in a decision theoretical framework remains an open research question.

\numberwithin{equation}{section}
 \numberwithin{figure}{section}
\numberwithin{algorithm}{section}
 \numberwithin{table}{section}
\makeatletter 
\newcommand{\section@cntformat}{Appendix \thesection:\ }
\makeatother

\spacingset{1.5}

\appendix
  
\section{Main Proofs} \label{sec:a1}

Throughout the rest of our discussion we define 
\begin{equation}
 \Pi_{o,n} =  \Big\{\pi_\alpha \in \Pi: \pi_\alpha \in \mathrm{arg} \sup_{\pi \in \Pi} \Big\{\alpha \hat{W}_0(\pi) + (1 - \alpha) \hat{W}_1(\pi) \Big\}, \text{ s.t. } \alpha \in (0,1)\Big\}, 
\end{equation} 
and let $N = \sqrt{n}$ as discussed in the main text. We denote $\alpha_1 - \alpha_2 = \varepsilon$, where, recall, the grid of $(\alpha_i)_{i=1}^N$ contains elements equally spaced. We say that $x \lesssim y$ if $y \le c_0 x$ for a finite constant $c_0$ independent of $n$.

\subsection{Auxiliary Lemmas}

\begin{lem} \label{prop:1} Under Assumption \ref{ass:unconf}, \ref{ass:overlap} for any sensitive attribute $s \in \{0,1\}$ 
\begin{equation} \label{eqn:wel}
\begin{aligned} 
&W_s(\pi) = \mathbb{E}\Big[ \frac{1\{S_i = s\}}{p_s} \Big(\frac{ Y_i D_i }{e(X_i, s)} - \frac{ Y_i (1 - D_i)}{1 - e(X_i, s)}\Big) \pi(X_i, s)\Big].
\end{aligned}   
\end{equation} 
\end{lem} 
\vspace{5 mm}
\begin{proof}[Proof of Lemma \ref{prop:1}]
Assumption \ref{ass:overlap} guarantees existence of the expectation.
By definition of the conditional expectation
$$
\begin{aligned} 
\eqref{eqn:wel} = \mathbb{E}\Big[ \Big(\frac{ Y_i D_i }{e(X_i, s)} - \frac{ Y_i (1 - D_i)}{1 - e(X_i, s)}\Big) \pi(X_i, s) \Big|S_i = s\Big]. 
\end{aligned} 
$$ 
 Using the law of iterated expectations and Assumption \ref{ass:unconf} the result directly follows.    
\end{proof} 

\vspace{5 mm}

\begin{lem} \label{lem:helper2}  Let $W_{s,n} = \frac{1}{n} \sum_{i=1}^n (\Gamma_{1,s,i} - \Gamma_{0,s,i}) \pi(X_i,s)$, where $\Gamma_{d,s,i}$ is defined as in Equation \eqref{eqn:gamma}. Let Assumptions \ref{ass:unconf}, \ref{ass:moment}, and \ref{ass:overlap} hold. Then with probability at least $1 - \gamma$, 
\begin{equation} 
\sup_{\alpha \in (0,1)} \sup_{\pi \in \Pi} \Big| \alpha W_{0}(\pi) + (1 - \alpha) W_{1}(\pi) - \alpha W_{0, n}(\pi) + (1 - \alpha) W_{1, n}(\pi) \Big| \le \bar{C} \frac{M}{\delta^2} \sqrt{v/n} + \frac{\bar{C} M}{\delta^2} \sqrt{\log(2/\gamma)/n}
\end{equation} 
for a universal constant $\bar{C} < \infty$. In addition,
\begin{equation} 
\mathbb{E}\Big[\sup_{\alpha \in (0,1)}  \sup_{\pi \in \Pi} \Big| \alpha W_{0}(\pi) + (1 - \alpha) W_{1}(\pi) - \alpha W_{0, n}(\pi) + (1 - \alpha) W_{1, n}(\pi) \Big| \Big] \le \bar{C} \frac{M}{\delta^2} \sqrt{v/n}. 
\end{equation} 

\end{lem} 

\vspace{5 mm}

\begin{proof}[Proof of Lemma \ref{lem:helper2}] Throughout the proof we refer to $\bar{C} < \infty$ as a universal constant. 
Observe first that under Assumption \ref{ass:overlap} and Assumption \ref{ass:moment}, we have
\begin{equation}
\sup_{\alpha \in (0,1)}  \sup_{\pi \in \Pi} \Big| \alpha W_{0}(\pi) + (1 - \alpha) W_{1}(\pi) - \alpha W_{0, n}(\pi) + (1 - \alpha) W_{1, n}(\pi) \Big|,  
\end{equation} 
satisfies the bounded difference assumption \citep{boucheron2013concentration} with constant $\frac{2 M}{\delta^2 n}$.\footnote{This follows from the triangular inequality and the fact that under Assumption \ref{ass:overlap} the inverse probability weight is uniformly bounded by $1/\delta^2$ and under Assumption \ref{ass:moment} the conditional mean function is bounded by $M$.} See for instance \cite{boucheron2005theory}.   
By the bounded difference inequality, with probability at least $1- \gamma$, 
\begin{equation} \label{eqn:thm1}
\begin{aligned} 
&\sup_{\alpha \in (0,1)}  \sup_{\pi \in \Pi} \Big| \alpha W_{0}(\pi) + (1 - \alpha) W_{1}(\pi) - \alpha W_{0, n}(\pi) + (1 - \alpha) W_{1, n}(\pi) \Big| \\ &\le \mathbb{E}\Big[\sup_{\alpha \in (0,1)}  \sup_{\pi \in \Pi} \Big| \alpha W_{0}(\pi) + (1 - \alpha) W_{1}(\pi) - \alpha W_{0, n}(\pi) + (1 - \alpha) W_{1, n}(\pi) \Big|\Big] + \bar{C} \frac{M}{\delta^2}\sqrt{\log(2/\gamma)/n}.
\end{aligned} 
\end{equation}  
We now move to bound the expectation in the right-hand side of Equation \eqref{eqn:thm1}. Under Assumption \ref{ass:unconf}, we obtain by Lemma \ref{prop:1} and trivial rearrangments, that 
\begin{equation}
\mathbb{E}\Big[\alpha W_{0}(\pi) + (1 - \alpha) W_{1}(\pi) - \alpha W_{0, n}(\pi) + (1 - \alpha) W_{1, n}(\pi) \Big] = 0. 
\end{equation} 

Using the symmetrization argument \citep{van1996weak}, we can now bound the above supremum with the Radamacher complexity of the function class of interest, which combined with the triangle inequality reads as follows: 
\begin{equation}
\begin{aligned}
&\mathbb{E}\Big[\sup_{\alpha \in (0,1)}  \sup_{\pi \in \Pi} \Big| \alpha W_{0}(\pi) + (1 - \alpha) W_{1}(\pi) - \alpha W_{0, n}(\pi) + (1 - \alpha) W_{1, n}(\pi) \Big|\Big] \\ 
&\le \mathbb{E}\Big[\sup_{\alpha \in (0,1)}  \sup_{\pi \in \Pi} \Big| \alpha W_{0}(\pi)  - \alpha W_{0, n}(\pi)\Big| \Big] + \mathbb{E}\Big[\sup_{\alpha \in (0,1)}  \sup_{\pi \in \Pi} \Big|(1 - \alpha)W_1(\pi) -  (1 - \alpha) W_{1, n}(\pi) \Big|\Big] \\ 
&\le \mathbb{E}\Big[ \sup_{\pi \in \Pi} \Big| W_{0}(\pi)  - W_{0, n}(\pi) \Big|\Big] + \mathbb{E}\Big[ \sup_{\pi \in \Pi} \Big|W_1(\pi) -  W_{1, n}(\pi) \Big|\Big] \\ 
&\le \mathbb{E}\Big[\sup_{\pi \in \Pi}  \Big|\frac{1}{n} \sum_{i=1}^n \sigma_i\pi(X_i, 1)\Gamma_{1,1,i} \Big| \Big] + \mathbb{E}\Big[\sup_{\pi \in \Pi}  \Big|\frac{1}{n} \sum_{i=1}^n \sigma_i \pi(X_i, 1) \Gamma_{0,1,i} \Big| \Big]\\ &+\mathbb{E}\Big[\sup_{\pi \in \Pi}  \Big|\frac{1}{n} \sum_{i=1}^n \sigma_i\pi(X_i, 0)\Gamma_{1,0,i} \Big|  \Big] + \mathbb{E}\Big[\sup_{\pi \in \Pi}  \Big|\frac{1}{n} \sum_{i=1}^n \sigma_i \pi(X_i, 0) \Gamma_{0,0,i} \Big|  \Big],
\end{aligned}  
\end{equation} 
where here $\sigma_i$ are independent Radamacher random variables. 
We can study each component of the above expression separately. By the Dudley's entropy integral bound, since the VC-dimension of the function class $\Pi$ is bounded by Assumption \ref{ass:moment}, and since each $\Gamma_{d,s,}$ is bounded, we obtain (see for instance \cite{wainwright2019high}), under Assumption \ref{ass:moment} (A) and (B), with trivial rearrangement
\begin{equation}
\mathbb{E}\Big[\sup_{\pi \in \Pi}  \Big|\frac{1}{n} \sum_{i=1}^n \sigma_i\pi(X_i, s)\Gamma_{d,s, i} \Big| \Big] \le \frac{M \bar{C}}{\delta^2} \sqrt{v/n} .
\end{equation} 
for each $d,s$. The remaining terms follow similarly. The proof is complete. 
\end{proof} 

\begin{lem} \label{lem:dr} Let Assumptions \ref{ass:unconf}, \ref{ass:moment}-\ref{ass:dr} hold. Then with probability at least $1 - \gamma$,
\begin{equation} 
\sup_{\alpha \in (0,1)} \sup_{\pi \in \Pi} \Big| \alpha W_{0}(\pi) + (1 - \alpha) W_{1}(\pi) - \alpha \hat{W}_{0}(\pi) + (1 - \alpha) \hat{W}_{1}(\pi) \Big| \le \bar{C} \frac{M}{\delta^2} \sqrt{v/n} + \frac{\bar{C} M}{\delta^2} \sqrt{\log(2/\gamma)/n}
\end{equation} 
for a universal constant $\bar{C} < \infty$. 
\end{lem} 

\vspace{5 mm}
\begin{proof}[Proof of Lemma \ref{lem:dr}]
First observe that we can bound the above expression as 
\begin{equation} 
\begin{aligned} 
& \sup_{\alpha \in (0,1)}  \sup_{\pi \in \Pi} \Big| \alpha W_{0}(\pi) + (1 - \alpha) W_{1}(\pi) - \alpha \hat{W}_{0}(\pi) + (1 - \alpha) \hat{W}_{1}(\pi) \Big| \le \\ &
\sup_{\alpha \in (0,1)}  \underbrace{\sup_{\pi \in \Pi} \Big| \alpha W_{0}(\pi) + (1 - \alpha) W_{1}(\pi) - \alpha W_{0,n}(\pi) + (1 - \alpha) W_{1, n}(\pi) \Big|}_{(I)} \\ &+ \sup_{\alpha \in (0,1)}  \underbrace{\sup_{\pi \in \Pi} \Big| \alpha W_{0, n}(\pi) + (1 - \alpha) W_{1, n}(\pi) - \alpha \hat{W}_{0}(\pi) + (1 - \alpha) \hat{W}_{1}(\pi) \Big|}_{(II)}. 
\end{aligned} 
\end{equation} 
Here $W_{s,n}$ is as defined in Lemma \ref{lem:helper2}. The term (I) is bounded as in Lemma \ref{lem:helper2}. Therefore, we are only left to discuss (II). 

Using the triangular inequality, we only need to bound
\begin{equation}
\sup_{\pi \in \Pi} \Big|W_{0,n}(\pi)  - \hat{W}_{0,n}(\pi)\Big| + \sup_{\pi \in \Pi} \Big|W_{1,n}(\pi)  - \hat{W}_{1,n}(\pi)\Big|.
\end{equation} 
We bound the first term while the second term follows similarly. 
We write 
\begin{equation}
\begin{aligned} 
&\sup_{\pi \in \Pi} \Big|W_{s,n}(\pi)  - \hat{W}_s(\pi)\Big| \\ &\le 
\Big|\frac{1}{n} \sum_{i=1}^n \frac{1\{S_i = s\}}{p_s} \frac{D_i (Y_i - m_{1,s}(X_i))}{e(X_i, s)} \pi(X_i, s) +  \frac{1\{S_i = s\}}{p_s} m_{1,s}(X_i)\pi(X_i, s) \\ &- \frac{1}{n} \sum_{i=1}^n \frac{1\{S_i = s\}}{\hat{p}_s} \frac{D_i (Y_i - \hat{m}_{1,s}(X_i))}{\hat{e}(X_i, s)} \pi(X_i, s) -  \frac{1\{S_i = s\}}{\hat{p}_s} \hat{m}_{1,s}(X_i)\pi(X_i, s)\Big| \\
&+ \Big|\frac{1}{n} \sum_{i=1}^n \frac{1\{S_i = s\}}{p_s} \frac{(1 - D_i) (Y_i - m_{0,s}(X_i))}{1 - e(X_i, s)} \pi(X_i, s) +  \frac{1\{S_i = s\}}{p_s} m_{0,s}(X_i)\pi(X_i, s)\\ &- \frac{1}{n} \sum_{i=1}^n \frac{1\{S_i = s\}}{\hat{p}_s} \frac{(1 - D_i) (Y_i - \hat{m}_{s,0}(X_i))}{1 - \hat{e}(X_i, s)} \pi(X_i, s) -  \frac{1\{S_i = s\}}{ \hat{p}_s} \hat{m}_{0, s}(X_i) \pi(X_i, s)\Big| .
\end{aligned} 
\end{equation} 
We discuss the first component while the second follows similarly. 

With trivial re-arrengment, using the triangular inequality, we obtain that the following holds 
\begin{equation} \label{eqn:i_and_ii}
\begin{aligned} 
&\Big|\frac{1}{n} \sum_{i=1}^n \frac{1\{S_i = s\}}{p_s} \frac{D_i (Y_i - m_{1,s}(X_i))}{e(X_i, s)} \pi(X_i, s) +  \frac{1\{S_i = s\}}{p_s} m_{1,s}(X_i)\pi(X_i, s) \\ &- \frac{1}{n} \sum_{i=1}^n \frac{1\{S_i = s\}}{\hat{p}_s} \frac{D_i (Y_i - \hat{m}_{1,s}(X_i))}{\hat{e}(X_i, s)} \pi(X_i, s) -  \frac{1\{S_i = s\}}{ \hat{p}_s} \hat{m}_{1,s}(X_i)\pi(X_i, s)\Big| \\
&\le \underbrace{\sup_{\pi \in \Pi} \Big|\frac{1}{n} \sum_{i=1}^n 1\{S_i = s\} D_i (Y_i - m_{1,s}(X_i)) \Big(\frac{1}{p_s e(X_i, s)} - \frac{1}{\hat{p}_s \hat{e}(X_i, s)}\Big) \pi(X_i, s)\Big|}_{(i)} \\ &+ \underbrace{\sup_{\pi \in \Pi}\Big| \frac{1}{n} \sum_{i=1}^n 
 \Big( \frac{1\{S_i =s\} D_i}{\hat{e}(X_i, s) \hat{p}_s}  - \frac{1\{S_i =s\}}{ p_s}\Big) (m_{1,s}(X_i) - \hat{m}_{1,s}(X_i)) \pi(X_i, s)\Big|}_{(ii)}. 
\end{aligned}  
\end{equation} 
We study $(i)$ and $(ii)$ separately. We start from $(i)$. 
Recall, that by cross fitting $\hat{e}(X_i, s) = \hat{e}^{-k(i)}(X_i, s)$, where $k(i)$ is the fold containing unit $i$. Therefore, observe that given the $K$ folds for cross-fitting, we have  
\begin{equation} \label{eqn:summanda} 
\begin{aligned} 
 &\Big|\frac{1}{n} \sum_{i=1}^n 1\{S_i = s\} D_i (Y_i - m_{1,s}(X_i)) \Big(\frac{1}{p_s e(X_i, s)} - \frac{1}{\hat{p}_s \hat{e}(X_i, s)}\Big) \pi(X_i, s)\Big| \\
 &\le  \sum_{k \in \{1, ..., K\}} \Big|\frac{1}{n} \sum_{i \in \mathcal{I}_k} 1\{S_i = s\} D_i (Y_i - m_{1,s}(X_i)) \Big(\frac{1}{p_s e(X_i, s)} - \frac{1}{\hat{p}_s^{(-k(i))} \hat{e}^{(-k(i))}(X_i, s)}\Big) \pi(X_i, s)\Big|. 
 \end{aligned} 
\end{equation} 
In addition, we have that 
\begin{equation}
\begin{aligned} 
&\mathbb{E}\Big[\sum_{i \in \mathcal{I}_k} 1\{S_i = s\} D_i (Y_i - m_{1,s}(X_i)) \Big(\frac{1}{p_s e(X_i, s)} - \frac{1}{\hat{p}_s^{(-k(i))} \hat{e}^{(-k(i))}(X_i, s)}\Big) \pi(X_i, s) \Big] \\ 
&= \mathbb{E}\Big[\mathbb{E}\Big[\sum_{i \in \mathcal{I}_k} 1\{S_i = s\} D_i (Y_i - m_{1,s}(X_i)) \Big(\frac{1}{p_s e(X_i, s)} - \frac{1}{\hat{p}_s^{(-k(i))} \hat{e}^{(-k(i))}(X_i, s)}\Big) \pi(X_i, s) \Big| \hat{p}^{(-k(i))}, \hat{e}^{(-k(i))} \Big] \Big] \\ &= 0,  
\end{aligned} 
\end{equation} 
by cross-fitting. 
By Assumption \ref{ass:dr}, we know that
\begin{equation}
\sup_{x \in \mathcal{X} ,s \in \mathcal{S}} \Big|\frac{1}{p_s e(x, s)} - \frac{1}{\hat{p}_s^{(-k(i))} \hat{e}^{(-k(i))}(x, s)}\Big| \le 2/\delta^2
\end{equation} 
and therefore each summand in Equation \eqref{eqn:summanda} is bounded by a finite constant $2M \bar{C} /\delta^2$, for a universal constant $\bar{C}$. 
We now obtain, using the symmetrization argument \citep{van1996weak}, and the Dudley's entropy integral \citep{wainwright2019high}
\begin{equation}
\begin{aligned} 
&\mathbb{E}\Big[\sup_{\pi \in \Pi} |\frac{1}{n} \sum_{i \in \mathcal{I}_k} 1\{S_i = s\} D_i (Y_i - m_{1,s}(X_i)) \Big(\frac{1}{p_s e(X_i, s)} - \frac{1}{\hat{p}_s^{(-k(i))} \hat{e}^{(-k(i))}(X_i, s)}\Big) \pi(X_i, s)| \Big| \hat{p}^{(-k(i))}, \hat{e}^{(-k(i))}\Big] \\ &\lesssim \frac{M}{\delta^2} \sqrt{v/n}.
\end{aligned} 
\end{equation} 
In addition, by the bounded difference inequality \citep{boucheron2005theory}, with probability at least $1 - \gamma$, for a universial constant $c < \infty$
\begin{equation}
\begin{aligned} 
&\sup_{\pi \in \Pi} \Big|\frac{1}{n} \sum_{i \in \mathcal{I}_k} 1\{S_i = s\} D_i (Y_i - m_{1,s}(X_i)) \Big(\frac{1}{p_s e(X_i, s)} - \frac{1}{\hat{p}_s^{(-k(i))} \hat{e}^{(-k(i))}(X_i, s)}\Big) \pi(X_i, s)\Big| \le  \\
&\mathbb{E}\Big[\sup_{\pi \in \Pi} |\frac{1}{n} \sum_{i \in \mathcal{I}_k} 1\{S_i = s\} D_i (Y_i - m_{1,s}(X_i)) \Big(\frac{1}{p_s e(X_i, s)} - \frac{1}{\hat{p}_s^{(-k(i))} \hat{e}^{(-k(i))}(X_i, s)}\Big) \pi(X_i, s)| \Big| \hat{p}^{(-k(i))}, \hat{e}^{(-k(i))}\Big] \\ &+ c \frac{M}{\delta^2} \sqrt{\frac{\log(2/\gamma)}{n}}. 
\end{aligned} 
\end{equation} 

We now consider the term $(ii)$. Observe that we can write 
\begin{equation} \label{eqn:j_and_jj}
\begin{aligned} 
(ii) \le & \underbrace{\sup_{\pi \in \Pi}\Big| \frac{1}{n} \sum_{i=1}^n 
\Big(\frac{D_i 1\{S_i = s\}}{\hat{p}_s \hat{e}(X_i, s)}  - \frac{D_i 1\{S_i = s\}}{ p_s e(X_i,s)}\Big) (m_{1,s}(X_i) - \hat{m}_{1,s}(X_i)) \pi(X_i, s)\Big|}_{(j)} \\ &+ \underbrace{\sup_{\pi \in \Pi}\Big| \frac{1}{n} \sum_{i=1}^n 
\Big(\frac{D_i 1\{S_i = s\}}{p_s e(X_i, s)}  - \frac{1\{S_i = s\}}{p_s}\Big) (m_{1,s}(X_i) - \hat{m}_{1,s}(X_i)) \pi(X_i, s)\Big|}_{(jj)}. 
\end{aligned} 
\end{equation} 
We consider each term seperately. Consider $(jj)$ first. Using the cross-fitting argument we obtain 
\begin{equation} 
\begin{aligned} 
&\sup_{\pi \in \Pi}\Big| \frac{1}{n} \sum_{i=1}^n 
\Big(\frac{D_i 1\{S_i = s\}}{p_s e(X_i, s)}  - \frac{1\{S_i = s\}}{p_s}\Big)( m_{1,s}(X_i) - \hat{m}_{1,s}(X_i)) \pi(X_i, s)\Big| \\ &\le \sum_{k \in \{1, ..., K\}} \sup_{\pi \in \Pi}\Big| \frac{1}{n} \sum_{i \in \mathcal{I}_k} 
\Big(\frac{D_i 1\{S_i = s\}}{p_s e(X_i, s)}  - \frac{1\{S_i = s\}}{p_s}\Big) (m_{1,s}(X_i) - \hat{m}_{1,s}^{(-k(i))}(X_i)) \pi(X_i, s)\Big|. 
\end{aligned} 
\end{equation} 
Observe now that 
\begin{equation} \label{eqn:helper100}
\begin{aligned} 
\mathbb{E}\Big[\Big(\frac{D_i 1\{S_i = s\}}{p_s e(X_i, s)}  - \frac{1\{S_i = s\}}{p_s} \Big) (m_{1,s}(X_i) - \hat{m}_{1,s}^{(-k(i))}(X_i)) \pi(X_i, s)\Big| \hat{m}_{1,s}^{(-k(i))} \Big] = 0, 
\end{aligned} 
\end{equation} 
since by cross-fitting, $\hat{m}_{1, s}^{(-k(i))}$ is independent of $(D_i, S_i, X_i)$ and, as a result, the conditional expectation of the left-hand side in Equation \eqref{eqn:helper100}, also conditional on $X_i$ equals zero. 
Therefore, following the same argument used for $(i)$ in Equation \eqref{eqn:i_and_ii}, we obtain that with probability at least $1 - \gamma$ 
\begin{equation}
\small 
\begin{aligned} 
& \sum_{k \in \{1, \cdots, K\}} \sup_{\pi \in \Pi}\Big| \frac{1}{n} \sum_{i \in \mathcal{I}_k} 
\Big(\frac{D_i 1\{S_i = s\}}{p_s e(X_i, s)}  - \frac{1\{S_i = s\}}{p_s}\Big) (m_{1,s}(X_i) - \hat{m}_{1,s}^{(-k(i))}(X_i)) \pi(X_i, s)\Big| \\ &\lesssim  \frac{K M}{\delta^2} \sqrt{\frac{vK}{n}} + \frac{M K}{\delta^2}\sqrt{\frac{\log(2 K/\gamma)}{n}}, 
\end{aligned} 
\end{equation} 
where the number of folds $K$ is a constant. 
We are now left to bound $(j)$ in Equation \eqref{eqn:j_and_jj}. We obtain that 
\begin{equation}
(j) \le \sqrt{\frac{1}{n} \sum_{i=1}^n \Big(\frac{1}{\hat{p}_s \hat{e}(X_i,s)} - \frac{1}{ p_s e(X_i,s)}\Big)^2} \sqrt{\frac{1}{n} \sum_{i=1}^n (m_{1,s}(X_i) - \hat{m}_{1,s}(X_i))^2}. 
\end{equation}  
Such a bound does not depend on $\pi$. Observe now that we can write by Assumption \ref{ass:dr}
\begin{equation} 
\begin{aligned}
&\sqrt{\frac{1}{n} \sum_{i=1}^n \Big(\frac{1}{\hat{p}_s \hat{e}(X_i,s)} - \frac{1}{p_s e(X_i,s)}\Big)^2} \sqrt{\frac{1}{n} \sum_{i=1}^n (m_{1,s}(X_i) - \hat{m}_{1,s}(X_i))^2} \\
&\le \frac{1}{\delta} \sqrt{ \sum_{k \in \{1, ..., K\}} \frac{1}{n} \sum_{i \in \mathcal{I}_k} \Big(\frac{1}{ \hat{e}^{(-k(i))}(X_i,s) \hat{p}^{-k(i)}} - \frac{1}{ e(X_i,s) p_s}\Big)^2} \sqrt{ \sum_{k \in \{1, ..., K\}} \frac{1}{n} \sum_{i \in \mathcal{I}_k} (m_{1,s}(X_i) - \hat{m}_{1,s}^{(-k(i))}(X_i))^2} .  
\end{aligned} 
\end{equation} 
By the bounded difference inequality, and the union bound we obtain that the following holds: 
\begin{equation}
\begin{aligned} 
&\sqrt{ \sum_{k \in \{1, ..., K\}} \frac{1}{n} \sum_{i \in \mathcal{I}_k} \Big(\frac{1}{\hat{e}^{(-k(i))}(X_i,s) \hat{p}^{-k(i)}} - \frac{1}{e(X_i,s) p_s}\Big)^2} \sqrt{ \sum_{k \in \{1, ..., K\}} \frac{1}{n} \sum_{i \in \mathcal{I}_k} (m_{1,s}(X_i) - \hat{m}_{1,s}^{(-k(i))}(X_i))^2} \\ 
&\le K\sqrt{\mathbb{E}\Big[\Big(\frac{1}{ \hat{e}(X_i,s) \hat{p}} - \frac{1}{ e(X_i,s) p_s}\Big)^2\Big]} \sqrt{\mathbb{E}\Big[ (m_{1,s}(X_i) - \hat{m}_{1,s}(X_i))^2\Big]} \\ &+ 2\sqrt[4]{\log(2K/\gamma)/n} \sqrt{\mathbb{E}\Big[\Big(\frac{1}{\hat{e}(X_i,s) \hat{p}} - \frac{1}{ e(X_i,s) p_s}\Big)^2\Big]} + 2\sqrt[4]{\log(2K/\gamma)/n} \sqrt{\mathbb{E}\Big[ (m_{1,s}(X_i) - \hat{m}_{1,s}(X_i))^2\Big]} \\ &+ 2\sqrt{\log(2K/\gamma)/n},
\end{aligned}
\end{equation} 
with probability at least $1 - \gamma$. Under Assumption \ref{ass:dr} and the union bound, the result completes since $K$ is a finite number. 
\end{proof}

\vspace{5 mm}

\begin{lem} \label{lem:cover} Let 
\begin{equation}
G(\alpha) =\sup_{\pi \in \Pi} \Big\{\alpha W_{0}(\pi) + (1 - \alpha) W_{1}(\pi) \Big\} - \sup_{ \pi \in \hat{\Pi}_{\mbox{\tiny o}}} \Big\{\alpha W_{0}(\pi) + (1 - \alpha) W_{1}(\pi) \Big\}. 
\end{equation} 
Define 
$$\mathcal{G} = \{G(\alpha), \alpha \in (0,1)\}.
$$
Under Assumption \ref{ass:moment}, for any $\varepsilon > 0$, there exist a set $\{\alpha_1, ..., \alpha_{N(\varepsilon)}\}$, such that for all $\alpha \in (0,1)$, 
\begin{equation}
|G(\alpha) - \max_{j \in \{1, ..., N(\varepsilon)\}} G(\alpha_j)| \le 4 \varepsilon M, 
\end{equation} 
and $N(\varepsilon) \le 1 + 1/\varepsilon$. 
\end{lem} 
\vspace{5 mm}
\begin{proof}[Proof of Lemma \ref{lem:cover}]

We denote $\{\alpha_1, ..., \alpha_{N(\varepsilon)}\}$ an $\varepsilon$-cover of the interval $(0,1)$ with respect to the L1 norm. Namely, $\{\alpha_1, ..., \alpha_{N(\varepsilon)}\}$ are equally spaced numbers between $(0,1)$.  Clearly, we have that the covering number $N(\varepsilon) \le 1 + 1/\varepsilon$. We denote 
\begin{equation}
G(\alpha) =  \sup_{\pi \in \Pi} \alpha W_{0}(\pi) + (1 - \alpha) W_{1}(\pi) - \sup_{\pi \in \hat{\Pi}_{\mbox{\tiny o}}} \Big\{\alpha W_{0}(\pi) + (1 - \alpha) W_{1}(\pi)\Big\} .
\end{equation} 
To characterize the corresponding cover of the function class 
$$\mathcal{G} = \{G(\alpha), \alpha \in (0,1)\},
$$ we claim that for any $\alpha \in (0,1)$, there exist an $\alpha_j$ in the $\varepsilon$ cover such that 
\begin{equation}
|G(\alpha) - G(\alpha_j)| \le 4 \varepsilon M. 
\end{equation} 
Such a result follows by the argument outlined in the following lines. 

\textit{Take $\alpha_j$ closest to $\alpha$}. Consider
\begin{equation} 
\begin{aligned}
&|G(\alpha) - G(\alpha_j)| \\
&= \Big|\sup_{\pi \in \Pi} \Big\{\alpha W_0(\pi) + (1 - \alpha) W_1(\pi) \Big\} - \sup_{\pi \in \hat{\Pi}_{\mbox{\tiny o}}} \Big\{\alpha W_0(\pi) + (1 - \alpha) W_1(\pi) \Big\}  \\ &- \sup_{\pi \in \Pi} \Big\{\alpha_j W_0(\pi) + (1 - \alpha_j) W_1(\pi) \Big\} + \sup_{\pi \in \hat{\Pi}_{\mbox{\tiny o}}} \Big\{ \alpha_j W_0(\pi) + (1 - \alpha_j) W_1(\pi) \Big\} \Big| \\
&\le  \underbrace{\Big|\sup_{\pi \in \Pi} \Big\{ \alpha W_0(\pi) + (1 - \alpha) W_1(\pi) \Big\} - \sup_{\pi \in \Pi} \Big\{\alpha_j W_0(\pi) + (1 - \alpha_j) W_1(\pi) \Big\} \Big|}_{(i)} \\ &+ \underbrace{\Big|\sup_{\pi \in \hat{\Pi}_{\mbox{\tiny o}}} \Big\{\alpha W_0(\pi) + (1 - \alpha) W_1(\pi) \Big\} - \sup_{\pi \in \hat{\Pi}_{\mbox{\tiny o}}} \Big\{\alpha_j W_0(\pi) + (1 - \alpha_j) W_1(\pi) \Big\}\Big|}_{(ii)}.
\end{aligned} 
\end{equation} 

We study $(i)$ and $(ii)$ separately. Consider first $(i)$. We observe the following fact: whenever 
\begin{equation}
\sup_{\pi \in \Pi} \alpha W_0(\pi) + (1 - \alpha) W_1(\pi) - \sup_{\pi \in \Pi} \alpha_j W_0(\pi) + (1 - \alpha_j) W_1(\pi)  > 0 
\end{equation} 
then we can bound 
\begin{equation}
(i) \le \Big|\alpha W_0(\pi^*) + (1 - \alpha) W_1(\pi^*) -  \alpha_j W_0(\pi^*) + (1 - \alpha_j) W_1(\pi^*) \Big|. 
\end{equation} 
Here $\pi^* \in \mathrm{arg} \sup_{\pi \in \Pi} \alpha W_0(\pi) + (1 - \alpha) W_1(\pi)$. When instead 
\begin{equation}
\sup_{\pi \in \Pi} \alpha W_0(\pi) + (1 - \alpha) W_1(\pi) - \sup_{\pi \in \Pi} \alpha_j W_0(\pi) + (1 - \alpha_j) W_1(\pi)  \le 0  
\end{equation} 
we can use the same argument by switching sign, which, with trivial rearrengment reads as 
\begin{equation} 
(i) \le \Big|\alpha W_0(\pi^{**}) + (1 - \alpha) W_1(\pi^{**}) -  \alpha_j W_0(\pi^{**}) + (1 - \alpha_j) W_1(\pi^{**}) \Big|. 
\end{equation} 
Here $\pi^{**} \in \mathrm{arg} \sup_{\pi \in \Pi} \alpha_j W_0(\pi) + (1 - \alpha_j) W_1(\pi)$.
Therefore we obtain, 
\begin{equation}
(i) \le  \sup_{\pi \in \Pi} \Big|\alpha W_0(\pi) + (1 - \alpha) W_1(\pi) -  \alpha_j W_0(\pi) + (1 - \alpha_j) W_1(\pi) \Big| \le 2|\alpha - \alpha_j| M
\end{equation} 
where the last inequality follows by Assumption \ref{ass:moment} and the triangle inequality. Similar reasoning also applies to $(ii)$. Since $\alpha_j$ was chosen to be the closest to $\alpha$, we have $|\alpha_j - \alpha| \le \varepsilon$. 
\end{proof}

\subsection{Proof of Lemma \ref{lem:hyperplane} }

The proof follows similarly to standard microeconomic textbook \citep{mas1995microeconomic}. Let 
\begin{equation}
\tilde{\Pi} = \{\pi_\alpha: \pi_\alpha \in \mathrm{arg} \sup_{\pi \in \Pi} \alpha_1 W_0(\pi) + \alpha_2 W_1(\pi), \quad \alpha \in \mathbb{R}_+^2, \alpha_1 + \alpha_2 > 0\}. 
\end{equation} 
Then we want to show that $\Pi_{\mbox{\tiny o}} = \tilde{\Pi}$. Trivially  
 $\tilde{\Pi} \subseteq \Pi_{\mbox{\tiny o}}$, since otherwise the definition of Pareto optimality would be violated. 
 Consider now some $\pi^* \in \Pi_{\mbox{\tiny o}}$. Then we show that there exist a vector $\alpha \in \mathbb{R}_+^2$, such that $\pi^*$ maximizes the expression 
\begin{equation}
\sup_{\pi \in \Pi} \alpha_1 W_0(\pi) + \alpha_2 W_1(\pi).  
\end{equation}

 Denote the set 
 \begin{equation}
 \mathcal{F} = \{(\tilde{W}_0, \tilde{W}_1) \in \mathbb{R}^2: \exists \pi \in \Pi: \tilde{W}_0 \le W_0(\pi)\text{ and } \tilde{W}_1 \le W_1(\pi) \}.  
 \end{equation}
 
 Since $(0,0) \in \mathcal{F}$, such a set is non-empty.  Notice now that $W_s(\pi)$ is linear is $\pi$  for $s \in \{0,1\}$. Therefore, we obtain that the set $\mathcal{F}$ is a convex set, since it denotes the sub-graph of a concave functional. We denote $\bar{W} = (W_0(\pi^*), W_1(\pi^*))$ and $\mathcal{G} = \mathbb{R}^2_{++} + \bar{W}$ the set of welfares that strictly dominates $\pi^*$. Then $\mathcal{G}$ is non-empty and convex. Since $\pi^* \in \Pi_{\mbox{\tiny o}}$, we must have that $\mathcal{F} \cap \mathcal{G} = \emptyset$. Therefore, by the separating hyperplane theorem, there exist an $\alpha \in \mathbb{R}^2$, with $\alpha \neq 0$, such that $\alpha^\top F \le \alpha^\top (\bar{W} + d)$ for any $F \in \mathcal{F}$, $d \in \mathbb{R}_{++}^2$. Let $d_1 \rightarrow \infty$, it must be that $\alpha_1 \in \mathbb{R}_+$, and similarly for $\alpha_2$. So $\alpha \in \mathbb{R}^2_{+}$. By letting $d \rightarrow 0$, we have that   $\alpha^\top F \le \alpha^\top \bar{W}$. This implies that 
 \begin{equation}
 \alpha_1  W_0(\pi) + \alpha_2 W_1(\pi) \le 
 \alpha_1  W_0(\pi^*) + \alpha_2 W_1(\pi^*) 
 \end{equation}  
 for any $\pi \in \Pi$ (since it is true for any $F \in \mathcal{F}$). Hence $\pi^*$ maximizes welfare over all possible feasible allocations once reweighted by $(\alpha_1, \alpha_2)$. Since the maximizer is invariant to multiplication of the objective function by constants, the result follows after dividing the objective function by the sums of the coefficients, which is non-zero by the separating hyperplane theorem. This completes the proof.  

\subsection{Proof of Proposition \ref{lem:soc_p}} 

First, observe that by rationality, preferences are complete and transitive. Observe also that the preference function equivalently correspond to lexico-graphic  with $\pi \succ \pi'$ if $\pi$ Pareto dominates $\pi'$. If instead neither $\pi, \pi'$, Pareto dominates the other, then $\pi \succ \pi'$ is UnFairness ($\pi$) $<$  UnFairness ($\pi'$). 
Therefore, it must be that 
$\mathcal{C}(\Pi) \subseteq \Pi_{\mbox{\tiny o}}$, with $\pi^\star \in \mathcal{C}(\Pi)$ if and only if  
$$
\pi^\star \in \mathrm{arg} \min_{\pi \in \Pi_{\mbox{\tiny o}}} \mathrm{UnFairness}(\pi). 
$$ 
By Lemma \ref{lem:hyperplane} the result directly follows.

\subsection{Proof of Corollary \ref{cor:prop}}

 Define $\widetilde{\Pi} \subseteq \Pi$ the set of policies that satisfy the constraint in Equation \eqref{eqn:main_eq} (i.e., feasible allocations). By Proposition \ref{lem:soc_p} $\widetilde{\Pi} = \Pi_{o}$. Observe now that $\pi_{\omega}$ is a feasible allocation under the constraint in Equation \eqref{eqn:main_eq}. This directly implies the conclusion for $\pi_{\omega}$. 

Consider now $\tilde{\pi}$, and fairness constraints not being binding. 
 If $\widetilde{\pi}$ is Pareto optimal, then it represents a feasible allocation (i.e. it satisfies the constraint in Equation \eqref{eqn:main_eq}). If it is not, then any other allocation that \textit{is} Pareto optimal and Pareto dominates $\widetilde{\pi}$ is feasible under the constraint in Equation \eqref{eqn:main_eq} completing the proof.  
Finally, whenever fairness constraints are binding, the estimated policy contains as one possible solution the policy which maximizes the utilitarian welfare under fairness constraints. This follows from the fact that in such case 
$$
\widetilde{\pi} \in \Big\{\mathrm{arg} \max_{\pi \in \Pi} p_1 W_1(\pi) + (1 - p_1) W_0(\pi)\Big\} \subseteq \Pi_o, 
$$  
since $\Pi = \Pi(\kappa)$. 

\subsection{Proof of Theorem \ref{lem:cons_dr}} 

Throughout the proof we refer to $\bar{C} < \infty$ as a universal constant. We write 

\begin{equation} 
\begin{aligned} 
&\sup_{\alpha \in (0,1)} \sup_{\pi \in \Pi} \Big| \alpha W_{0}(\pi) + (1 - \alpha) W_{1}(\pi) - \max_{\alpha_j \in \{\alpha_1, ..., \alpha_N\}} \alpha_j \hat{W}_{0}(\pi) - (1 - \alpha_j) \hat{W}_{1}(\pi) - \lambda/\sqrt{n} \Big|  \\ 
&\le \underbrace{\sup_{\alpha \in (0,1)} \sup_{\pi \in \Pi} \Big| \alpha W_{0}(\pi) + (1 - \alpha) W_{1}(\pi) - \alpha \hat{W}_{0}(\pi) - (1 - \alpha) \hat{W}_{1}(\pi) \Big|}_{(I)}  + \frac{\lambda}{\sqrt{n}} \\ &+ \underbrace{\sup_{\alpha \in (0,1)} \sup_{\pi \in \Pi} \Big| \alpha \hat{W}_{0}(\pi) + (1 - \alpha) \hat{W}_{1}(\pi) - \max_{\alpha_j \in \{\alpha_1, ..., \alpha_N\}} \alpha_j \hat{W}_{0}(\pi) - (1 - \alpha_j) \hat{W}_{1}(\pi) \Big|}_{(II)}.  
\end{aligned} 
\end{equation} 

$(I)$ is bounded as in Lemma \ref{lem:dr}. $(II)$ is bounded as follows. 
\begin{equation}
(II) \le  \varepsilon \sup_{\pi \in \Pi} |  \hat{W}_{0}(\pi)| + \varepsilon \sup_{\pi \in \Pi} |\hat{W}_{1}(\pi)|. 
\end{equation} 
Under Assumption \ref{ass:dr}, the estimated conditional mean and propensity score are uniformly bounded. Therefore we obtain that 
$$
\varepsilon \sup_{\pi \in \Pi} |  \hat{W}_{0}(\pi)| + \sup_{\pi \in \Pi} \varepsilon |\hat{W}_{1}(\pi)| \le \bar{C} \varepsilon \frac{M}{\delta^2} \le  \bar{C} \frac{M}{N \delta^2} . 
$$

\subsection{Proof of Theorem \ref{thm:1b} }

Recall the definition of $\bar{W}_\alpha$ in Equation \eqref{eqn:W_a}. The set of Pareto optimal policies reads as follows 
$$
\pi: \alpha W_1(\pi) + (1 - \alpha) W_0(\pi) \ge \bar{W}_\alpha \text{ for some } \alpha \in (0,1).  
$$  
Now it suffices to show  
for the claim to hold that 
$$
P\Big(\forall \alpha \in (0,1), \quad \max_{j \in \{1, \cdots, N\}} \bar{W}_\alpha - \bar{W}_{j,n} + \lambda(\gamma)/\sqrt{n} + \frac{\underline{b}}{N} \ge 0\Big) \le \gamma, 
$$
where $\lambda(\gamma) = \underline{b} (\sqrt{v} + \sqrt{ \log(2/\gamma) })$, whenever $N = \sqrt{n}$ (and hence $\lambda = \lambda(\gamma) + \underline{b}$). Observe that since $\{\alpha_1, \cdots, \alpha_N\}$ are equally spaced, we have that for all $\alpha \in (0,1)$ 
$$
 \sup_{\pi \in \Pi} \alpha W_1(\pi) + (1 - \alpha) W_0(\pi) \ge 
\sup_{\pi \in \Pi} \alpha_j W_1(\pi) + (1 - \alpha_j) W_0(\pi) + M \varepsilon   
$$ 
for some $j \in \{1, \cdots, N\}$ by Assumption \ref{ass:overlap} (ii). Taking $\underline{b} \ge M,\varepsilon = 1/N$, we have
$$
\begin{aligned} 
& P\Big(\forall \alpha \in (0,1), \quad \max_{j \in \{1, \cdots, N\}} \bar{W}_\alpha - \bar{W}_{j,n} + \lambda(\gamma)/\sqrt{n} + \frac{\underline{b}}{N} \ge 0\Big) \\ &\le P\Big( \max_{j \in \{1, \cdots, N\}} \bar{W}_{\alpha_j} - \bar{W}_{j,n} + \lambda(\gamma)/\sqrt{n} \ge 0 \Big).
\end{aligned} 
$$
We now observe that the following inequality holds:  
$$
\begin{aligned} 
& \sup_{\pi \in \Pi}  \alpha_j W_1(\pi) + (1 - \alpha_j) W_0(\pi) - \bar{W}_{j,n} \\ &=\sup_{\pi \in \Pi} \Big\{ \alpha_j W_1(\pi) + (1 - \alpha_j) W_0(\pi)\Big\} - \sup_{\pi \in \Pi} \Big\{ \alpha \hat{W}_1(\pi) + (1 - \alpha) \hat{W}_0(\pi) \Big\} \\
&\le 2 \sup_{\pi \in \Pi} \Big|\alpha_j W_1(\pi) + (1 - \alpha_j) W_0(\pi) - \alpha_j  \hat{W}_1(\pi)  + (1 - \alpha_j) \hat{W}_0(\pi)\Big|. 
\end{aligned} 
$$ 
By Lemma \ref{lem:dr}, with probability at least $1 - \gamma$, 
$$
 \sup_{\pi \in \Pi} \max_{ \alpha_j, j \in \{1, \cdots, N\}} \Big|\alpha_j W_1(\pi) + (1 - \alpha_j) W_0(\pi) - \alpha_j  \hat{W}_1(\pi)  + (1 - \alpha_j) \hat{W}_0(\pi)\Big| \le \bar{C} \sqrt{\frac{v}{n}} + \bar{C} \sqrt{\frac{\log(2/\gamma)}{n}}
$$ 
for a finite constant $\bar{C}$ independent of $n$. By choosing $\underline{b} \ge 2 \bar{C} + M$, the proof completes.

\subsection{Proof of Theorem \ref{thm:2}} 

By Theorem \ref{thm:1b} with probability at least $1 - \gamma$,  $\Pi_{\mbox{\tiny o}} \subseteq \hat{\Pi}_{\mbox{\tiny o}}(\lambda)$ with $\hat{\Pi}_{\mbox{\tiny o}}(\lambda)$ in Equation \eqref{eqn:set_const2}. As a result, we can write with probability $1 - \gamma$, 
$$
\mathrm{UnFairness}(\hat{\pi}) - \mathrm{inf}_{\pi \in \Pi_{\mbox{\tiny o}}} \mathrm{UnFairness}(\pi) \le \mathrm{UnFairness}(\hat{\pi}) - \mathrm{inf}_{\pi \in  \hat{\Pi}_{\mbox{\tiny o}}(\lambda)} \mathrm{UnFairness}(\pi). 
$$   
We then write 
$$
\small 
\begin{aligned} 
& \mathrm{UnFairness}(\hat{\pi}) - \mathrm{inf}_{\pi \in  \hat{\Pi}_{\mbox{\tiny o}}(\lambda)} \mathrm{UnFairness}(\pi)  =  \mathrm{UnFairness}(\hat{\pi}) - \hat{\mathcal{V}}_n(\hat{\pi}) + \hat{\mathcal{V}}_n(\hat{\pi}) -  \mathrm{inf}_{\pi \in  \hat{\Pi}_{\mbox{\tiny o}}(\lambda)} \mathrm{UnFairness}(\pi). 
\end{aligned} 
$$ 
Since $\hat{\pi}_\lambda \in \hat{\Pi}_{\mbox{\tiny o}}(\lambda)$, we have 
$$
\begin{aligned}
& \mathrm{UnFairness}(\hat{\pi}) - \hat{\mathcal{V}}_n(\hat{\pi}) + \hat{\mathcal{V}}_n(\hat{\pi}) -  \mathrm{inf}_{\pi \in  \hat{\Pi}_{\mbox{\tiny o}}(\lambda)} \mathrm{UnFairness}(\pi)   \\ 
& \le 2 \sup_{\pi \in   \hat{\Pi}_{\mbox{\tiny o}}(\lambda)}\Big|\mathrm{UnFairness}(\pi) - \hat{\mathcal{V}}_n(\pi) \Big| \le 2 \sup_{\pi \in   \Pi}\Big|\mathrm{UnFairness}(\pi) - \hat{\mathcal{V}}_n(\pi) \Big| 
\end{aligned} 
$$ 
where the last equality follows from the fact that $\hat{\Pi}_{\mbox{\tiny o}}(\lambda) \subseteq \Pi$. Assumption \ref{ass:unique2} bounds $\sup_{\pi \in   \Pi}\Big|\mathrm{UnFairness}(\pi) - \hat{\mathcal{V}}_n(\pi) \Big|$ completing the proof. 

\subsection{Proof of Theorem \ref{thm:between_groups}} 

For $\widehat{D}(\pi)$ it suffices to observe that 
\begin{equation} \label{eqn:bb} 
\sup_{\pi \in \Pi} \Big|\widehat{W}_1(\pi_1) - \widehat{W}_0(\pi) - W_1(\pi) + W_0(\pi)\Big| \le \sup_{\pi \in \Pi} \Big|\widehat{W}_1(\pi_1) - W_1(\pi)\Big| + \sup_{\pi \in \Pi}  \Big|W_0(\pi) - \widehat{W}_0(\pi)\Big| 
\end{equation}  
with each term being bounded with probability at least $1 - 2 \gamma$\footnote{$2 \gamma$ follows by the union bound.}, by $\bar{C} \sqrt{v/n} + \bar{C}\sqrt{\log(2/\gamma)/n}$ for a finite constant $\bar{C} < \infty$, similarly to what discussed in the proof of Lemma \ref{lem:dr}. 

The UnFairness bound follows as a corollary of Theorem \ref{thm:2}, where here Assumption \ref{ass:unique2} holds with $\mathcal{K}(\Pi, \gamma) n^{-\eta} \lesssim \bar{C} \sqrt{v/n} + \bar{C}\sqrt{\log(2/\gamma)/n}$ for a finite constant $\bar{C} < \infty$, i.e., the bound of Equation \eqref{eqn:bb}.  

For $\widehat{C}(\pi)$ the argument follows similarly, after noticing that we can bound 
$$ 
\begin{aligned} 
&\sup_{\pi \in \Pi} \Big| \frac{1}{n \hat{p}_1} \sum_{i=1}^n \pi(X_i) S_i -   \mathbb{E}[\pi(X) | S = 1] + \frac{1}{n (1 - \hat{p}_1)} \sum_{i=1}^n \pi(X_i) (1 - S_i) -   \mathbb{E}[\pi(X) | S = 0] \Big| \\ &\le  
\underbrace{\sup_{\pi \in \Pi} \Big| \frac{1}{n \hat{p}_1} \sum_{i=1}^n \pi(X_i) S_i -   \mathbb{E}[\pi(X) | S = 1]\Big|}_{(A)} + \underbrace{\sup_{\pi \in \Pi} \Big|\frac{1}{(1 - \hat{p}_1)n} \sum_{i=1}^n \pi(X_i) (1 - S_i) -   \mathbb{E}[\pi(X) | S = 0] \Big| }_{(B)}      .     
\end{aligned} 
$$              
We proceed by bounding $(A)$,while $(B)$ follows similarly. We have 
$$
(A) \le \underbrace{\sup_{\pi \in \Pi} \Big| \frac{1}{p_1 n} \sum_{i=1}^n \pi(X_i) S_i -   \mathbb{E}[\pi(X) | S = 1]\Big|}_{(i)} +  \underbrace{\Big|\frac{1}{ p_1}  -   \frac{1}{\hat{p}_1}  \Big|}_{(ii)}, 
$$ 
where the second component follows by the triangular inequality and the fact that $\pi(X_i) S_i \in \{0,1\}$. We now observe that each summand in $(i)$ is centered around its expectation. Therefore, we can bound $(i)$ using the Radamacher complexity of $\Pi$, with 
$$
\mathbb{E}[(i)]  \le \frac{2}{\delta} \mathbb{E}\Big[\sup_{\pi \in \Pi} \Big|\frac{1}{n} \sum_{i=1}^n \sigma_i \pi(X_i) S_i\Big|\Big],  
$$ 
with $\sigma_1, \cdots, \sigma_n$ being independent Radamacher random variables. Using the Dudley's entropy bound (see \cite{wainwright2019high}) it is easy to show that the right-hand side is bounded by $\bar{C} \sqrt{v/n}$ for a constant $\bar{C} < \infty$. Finally, using the bounded difference inequality \citep{boucheron2003concentration}, with probability at least $1 - \gamma$, 
$$
|(i) - \mathbb{E}[(i)]| \le \bar{C} \sqrt{\frac{\log(2/\gamma)}{n}},  
$$ 
for a finite constant $\bar{C}$. 
The bound on the second component (ii) follows from standard property of the
sample mean and the assumption that $\hat{p}_1 \ge \delta$.  
The final statement follows as a direct corollary of Theorem \ref{thm:2}.

For $\mathcal{I}(\pi)$ the claim holds since 
\begin{equation}
\begin{aligned} 
&\sup_{\pi \in \Pi} \Big| I_s(\pi) - \hat{I}_s(\pi) \Big| \le \\ &\underbrace{\sup_{\pi \in \Pi} \Big|\frac{1}{n} \sum_{i=1}^n  (\hat{\Gamma}_{1, s,i} - \hat{\Gamma}_{0,s,i})\pi(X_i,s')  - \mathbb{E}\Big[(\Gamma_{1, s,i} - \Gamma_{0,s,i})\pi(X_i,s')\Big]\Big|}_{(A)} \\& + \underbrace{\sup_{\pi \in \Pi} \Big|\frac{1}{n} \sum_{i=1}^n  (\hat{\Gamma}_{1, s,i} - \hat{\Gamma}_{0, s,i})\pi(X_i,s) - \mathbb{E}[(\Gamma_{1, s,i} - \Gamma_{0,s,i})\pi(X_i,s)] \Big|}_{(B)}. 
\end{aligned} 
\end{equation} 
Observe now that under Assumption \ref{ass:dr}, following the same argument in Lemma \ref{lem:dr}, we can bound $(A)$ and $(B)$ as follows 
$$
(A) \lesssim \sqrt{\frac{v}{n}} + \sqrt{\frac{ \log(2/\gamma)}{n}}, \quad (B) \lesssim \sqrt{\frac{v}{n}} +  \sqrt{\frac{\log(2/\gamma)}{n}}. 
$$ 
with probability at least $1 - \gamma$. The reader may refer to the proof of Lemma \ref{lem:dr} for details.

\subsection{Proof of Theorem \ref{thm:lower_bound}} 

First, since $\pi(x,s)$ is constant in $s$ with an abuse of notation we can write $\pi(x)$ as a function of $x$ only.  
We first observe that we can write 
$$
\begin{aligned} 
C(\pi) = \mathbb{E}\Big[(\frac{(1 - S)}{1 - p_1} - \frac{S}{p_1}) \pi(X)\Big] =  \mathbb{E}\Big[\frac{(p_1 - S)}{(1 - p_1)p_1}  \pi(X)\Big] 
\end{aligned} 
$$
For the lower bound it suffices to find one distribution which satisfies the condition. 
We choose $Y(1) = 0$, and $Y(0) = 0$ almost surely, which satisfies the bounded assumption on $Y$. This condition implies that any $\pi \in \Pi$ satisfies Pareto optimality, hence $\Pi_o = \Pi$.

Observe that the expression for $C(\pi)$ corresponds to the risk associated with a classifier $\pi(X)$ for classifying the sensitive attribute $S$ with loss 
$$
l(S, \pi(X)) \propto (p_1 - S) \pi(X)  =  \begin{cases}
& p_1 - 1 \text{ if } S = 1, \pi(X) = 1 \\ 
 & p_1 \text{ if } S = 0, \pi(X) = 1 \\
 & 0 \text{ otherwise }. 
\end{cases} .
$$ 

We now proceed following some of the steps in Theorem 14.5 and Theorem 14.6 in \cite{devroye2013probabilistic}, but introducing modifications in the construction of the set of distributions under consideration and in the data-generating process due to the different loss function and its dependence with $P(S = 1)$ (which itself depends on the distribution of $(X,S)$).\footnote{The lack of restriction on the error of the classifier represents a further difference.}  
We start by choosing $D$ to be distributed as a Bernoulli random variable independent of $(X, S)$. As a result, $(Y, D)$ are independent of $(X,S)$. Therefore, since $(Y,D)$ is independent of $(X,S)$ it suffices to focus on classifiers $\pi_n(X)$ constructed using information $(X_1, S_1), \cdots, (X_n, S_n)$ only. The rest of the proof consists in constructing a distribution of $(X,S)$ such that the lower bound is attained. Recall that classifiers depend on $X$ only and not on $S$ by assumption. 

Consider first the case where $(v-1)/2$ is an integer. The case where it is not follows similarly to below and discussed at the end of the proof.   
We construct a family of distributions for $(X,S)$, defined $\mathcal{F}$ as follows: first we find points $x_1, \cdots x_{v}$ that are shattered by $\Pi_o$. Each distribution in $\mathcal{F}$ is concentrated on the set of these points. A member in $\mathcal{F}$ is described by $v - 1$ bits $b_1, \cdots, b_{v-1}$. This is representated as a bit vector $b  \subset \{0,1\}^{v-1}$. Each bit vector that we consider is assumed to sum to $(v-1)/2$, namely 
$$
\sum_{i=1}^{v-1} b_i = \frac{v-1}{2}.
$$ 
 Assume that $v - 1 \le n$. For each vector $b$, we let $X$ put mass $m$ at $x_i, i < v$, and mass $1 - (v-1)m$ at $x_v$. This imposes the condition $(v - 1)m \le 1$, which will be satisfied.
We choose for all $b$ that we consider $P(S = 1) = p_1 \in (\delta, 1- \delta)$ which we choose later in the proof. 
 Next, introduce the constant $c \in (0,p_1)$. 
Let $U$ a uniform random variable on $[0,1]$,
$$
S = \begin{cases}
&1 \text{ if } U \le p_1 - c + 2cb_i, X= x_i, i < v \\ 
&1 \text{ if } U \le p_1, X= x_v \\ 
&0 \text{ otherwise} 
\end{cases}.
$$ 
 Thus for $X = x_i, i < v$, $S$ is one with probability $p_1 - c$ or $p_1 + c$, while for $X = x_v$ $S$ is one with probability $p_1$. Now observe that the choice of $S$ and the fact that $P(S = 1) = p_1$ implies that 
\begin{equation} \label{eqn:nn} 
p_1 = \sum_{i=1}^{v-1} m(p_1 - c + 2cb_i) + p_1(1 - m(v-1)) = (v-1) m p_1 + p_1(1 - m(v-1)),
 \end{equation}  
 since $c\sum_{i=1}^{v-1} b_i = c \frac{v-1}{2}$ by the restriction on $b \in \mathcal{B}$. The above expression is satisfied for any $m$, so no restrictions on $m$ are implied by the Equation \eqref{eqn:nn}. 
 With a simple argument, it is easy to show that one of the best rules\footnote{A different which leads to the same objective is the one that classifies one also for $X = x_v$. This would be indifferent with respect to $f_b$ since the loss function at $X = x_v$ is always zero in expectation for either prediction.} for $b$ is the one which sets 
$$
f_b(x) = \begin{cases}
&1 \text{ if } x = x_i, i < v, b_i = 1 \\ 
&0 \text{ otherwise}. 
\end{cases} 
$$ 
Such rule is feasible since it has VC-dimension $v$. 
Notice now that we can write for the decision rule $f_b(x)$, 
$
\mathbb{E}[l(S, f_b(X)) | X = x_i] = -c
$  for $i < v$, for fixed $b$. Observe now that we can write for any $\pi_n, X \in \{x_1, \cdots, x_{v-1}\}$ , for fixed $b$, 
$$
\begin{aligned} 
\mathbb{E}[l(S, \pi_n(X)) | X] - \mathbb{E}[l(S, f_b(X)) | X] &  \ge 
2 c 1\{\pi_n(X) \neq f_b(X)\}, 
\end{aligned} 
$$ 
since if $\pi_n(X) = 1 - f_b(X)$, then $\mathbb{E}[l(S, 1 - f_b(X)) | X] = c$. 
Therefore we can bound for any $\pi_n$, and a fixed $b$
\begin{equation} \label{eqn:hhg} 
\begin{aligned} 
\mathrm{UnFairness}(\pi_n) - \inf_{\pi \in \Pi_o} \mathrm{UnFairness}(\pi) &\propto \mathbb{E}[l(S, \pi_n(X))] - \inf_{\pi \in \Pi} \mathbb{E}[l(S, \pi(X))]  \\ 
&\ge \sum_{j=1}^{v-1} 2 m c 1\{\pi_n(x_j) = 1 - f_b(x_j)\} \\ 
&\ge  \sum_{j=1}^{v-1} 2 m c 1\{\pi_n(x_j) = 1 - f_b(x_j)\}.
\end{aligned} 
\end{equation}   
Since we take the supremum over the class of distribution $P_b \in \mathcal{F}$ indexed by the bit-vector $b$, it suffices to provide upper bound with respect to $b$ being a random variable and take expectations over $b$. We replace $b$ by a uniformly distributed random variable $B$ over $\mathcal{B} \subset \{0,1\}^{v-1}$, where $\mathcal{B}$ is the set of bit vectors which sum to $(v-1)/2$. We observe that for any $t \ge 0$, 
$$
\begin{aligned} 
 & \sup_{(X,S) \in \mathcal{F}} P\Big(\mathrm{UnFairness}(\pi_n) - \inf_{\pi \in \Pi_o} \mathrm{UnFairness}(\pi) > t\Big) \\ 
 & = \sup_b  P\Big(\mathrm{UnFairness}(\pi_n) - \inf_{\pi \in \Pi_o} \mathrm{UnFairness}(\pi) > t\Big) \\ 
 & \ge \mathbb{E}_b\Big[1\{\mathrm{UnFairness}(\pi_n) - \inf_{\pi \in \Pi_o} \mathrm{UnFairness}(\pi) > t\}\Big] \text{ (with random b)} 
\\ &\ge \mathbb{E}_b\Big[1\Big\{\sum_{j=1}^{v-1} 2 m  c 1\{\pi_n(x_j) = 1 - f_b(x_j)\} > t\Big\}\Big]
\end{aligned} 
$$ 
where the last inequality uses Equation \eqref{eqn:hhg} and the monotonicity of the indicator function. We can now write 
$$
\begin{aligned} 
& \mathbb{E}_b\Big[1\Big\{\sum_{j=1}^{v-1} 2 m   c 1\{\pi_n(x_j) = 1 - f_b(x_j)\} > t\Big\}\Big] \\ 
& = \frac{1}{|\mathcal{B}|} \sum_{(x_1, \cdots, x_n', s_1, \cdots, s_n) 
\in (\{x_1, \cdots, x_v\} \times \{0,1\})^2} \\ &\sum_{b \in \mathcal{B}}  1\Big\{\sum_{j=1}^{v-1} 2 m  c 1\{\pi_n(x_j) = 1 - f_b(x_j)\} > t\Big\} \prod_{j=1}^n p_b(x_j', s_j)
\end{aligned} 
$$ 
with $p_b(x_j', s_j)$ denoting the joint probability of $x_j', s_j$. For a fixed $b$, define $b^c = (1 - b_1, \cdots, 1 - b_{v-1})$. Observe that if $b \in \mathcal{B}$, then $b^c \in \mathcal{B}$ since we assumed that $(v-1)/2$ is an integer. Now observe that if 
\begin{equation} \label{eqn:jhgf}
\frac{t}{2 m c} \le (v-1)/2,  
\end{equation}  
then 
$$
1\Big\{\sum_{j=1}^{v-1} 2 m  c 1\{\pi_n(x_j) = 1 - f_b(x_j)\} > t\Big\} + 1\Big\{\sum_{j=1}^{v-1} 2 m  c 1\{\pi_n(x_j) = 1 - f_{b^c}(x_j)\} > t\Big\} \ge 1
$$ 
since it must be that either (or both) indicators are equal to one. Therefore for $t \Big/ 2 m c \le (v-1)/2$, the last expression in the lower bound above is bounded from below by 
$$
\begin{aligned}
\frac{1}{|\mathcal{B}|} \sum_{(x_1, \cdots, x_n', s_1, \cdots, s_n) 
\in (\{x_1, \cdots, x_v\} \times \{0,1\})^2}  \sum_{b \in \mathcal{B}} \frac{1}{2} \min\Big\{\prod_{j=1}^n p_b(x_j', s_j), \prod_{j=1}^n p_{b^c}(x_j', s_j)\Big\}.  
\end{aligned} 
$$
By LeCam's inequality, we have that the above expression is bounded from below by (see Page 244 in \citealt{devroye2013probabilistic})
$$
\frac{1}{4 |\mathcal{B}|} \sum_{b \in \mathcal{B}} \Big(\sum_{(x,s)} \sqrt{p_b(x,s)p_{b^c}(x,s)}\Big)^{2n}. 
$$ 
Observe that we have for $x = x_v$, 
$$
p_b(x,1) = p_{b^c}(x,1) = p_1(1 - m(v-1)), \quad  p_b(x,1) = p_{b^c}(x,1) = (1 - p_1)(1 - m(v-1)). 
$$ 
 For $x = x_i, i< v$, we have 
$$
p_b(x,s)p_{b^c}(x,s) = m^2(p_1^2 - c^2), \quad s \in \{0,1\}. 
$$ 
Therefore, we obtain 
$$
\begin{aligned} 
\sum_{(x,s)} \sqrt{p_b(x,s)p_{b^c}(x,s)} &=  (1 - m(v-1)) + 2 (v-1)m \sqrt{(p_1^2 -c^2)} \\ &= (1 - (v-1)m) + 2 (v-1)m \sqrt{(p_1^2 -c^2)}.
\end{aligned} 
$$ 
Hence we can write 
$$
\begin{aligned} 
\frac{1}{4 |\mathcal{B}|}  \sum_{b \in \mathcal{B}} \sum_{(x,s)} \sqrt{p_b(x,s)p_{b^c}(x,s)} &= \frac{1}{4} \Big\{ (1 - (v-1)m) + 2 (v-1)m \sqrt{(p_1^2 -c^2)}\Big\}. 
\end{aligned} 
$$ 
Define $F = m (v-1)(p_1 - c)$. Then we can write  
$$
\begin{aligned} 
\frac{1}{4 |\mathcal{B}|} \sum_{b \in \mathcal{B}} \Big(\sum_{(x,s)} \sqrt{p_b(x,s)p_{b^c}(x,s)}\Big)^{2n} &= \frac{1}{4} \Big\{ (1 - (v-1)m) + 2 (v-1)m \sqrt{(p_1^2 -c^2)}\Big\} \\ 
& = \frac{1}{4} \Big\{ 1 - \frac{F}{p_1 - c}\Big(1 - \sqrt{4 p_1^2 - 4 c^2}\Big)\Big\}^{2n}. 
\end{aligned} 
$$ 
We now choose $p_1 = 1/2$. We can now follow \cite{devroye2013probabilistic}, end of Page 244 and write
$$
\begin{aligned} 
\frac{1}{4} \Big\{ 1 - \frac{F}{p_1 - c}\Big(1 - \sqrt{4 p_1^2 - 4 c^2}\Big)\Big\}^{2n} & = \frac{1}{4} \Big\{ 1 - \frac{F}{p_1 - c}\Big(1 - \sqrt{1 - 4 c^2}\Big)\Big\}^{2n} \\ & \ge
\frac{1}{4} \Big\{ 1 - \frac{F}{p_1 - c} 4c^2\Big\}^{2n}
\\ &\ge  \frac{1}{4} \exp\Big(- \frac{16 n F c^2}{1 - 2c} \Big/ \Big(1 - \frac{8Fc^2}{1 - 2c}\Big)\Big), 
\end{aligned} 
$$ 
where we used $1  - x \ge e^{-x/(1 - x)}$. 

We now choose $c = \frac{t}{(v-1)m}$, which satisfies Equation \eqref{eqn:jhgf}, and where we need the condition that $0 < t \le \frac{(v-1)m}{2}$ which we check later in the proof.  We write 
$$
\frac{16 n F c^2}{1 - 2c} \Big/ \Big(1 - \frac{8Fc^2}{1 - 2c} \Big) =  
\frac{16 n F c^2}{1 - 2c - 8F c^2}. 
$$ 
Fix a constant $h \in (0,1)$ whose conditions will be discussed below together with the conditions for $t$.
Take $t, h$ such that $1 - 2c - 8F c^2 \ge h \in (0,1)$. 
Then it follows that (since $c = \frac{t}{(v-1)m}$)
$$
\frac{16 n F c^2}{1 - 2c - 8F c^2} \le \frac{16 n t^2 F}{(v-1)^2 m^2 h}. 
$$ 
Hence, the lower bound reads as follows: 
$$
\sup_{(X,S) \in \mathcal{F}} P\Big(\mathrm{UnFairness}(\pi_n) - \inf_{\pi \in \Pi_o} \mathrm{UnFairness}(\pi) > t\Big) \ge \frac{1}{4} \exp \Big(-\frac{16 n t^2 F}{(v-1)^2 m^2 h}\Big).
$$ 
Let $ \frac{1}{4} \exp \Big(-\frac{16 n t^2 F}{(v-1)^2 m^2 h}\Big) = \kappa$. By re-arranging the expression, we write with probability at least $\kappa$, for some distribution in $\mathcal{F}$, for all $\pi_n$, 
\begin{equation} \label{eqn:lower}  
\mathrm{UnFairness}(\pi_n) - \inf_{\pi \in \Pi_o} \mathrm{UnFairness}(\pi) \ge \sqrt{\frac{F (v-1)^2 m^2 \log(\frac{1}{4 \kappa})}{16 n h}}
\end{equation}  
where we chose $t = \sqrt{\frac{F (v-1)^2 m^2 \log(\frac{1}{4 \kappa})}{16 n h}}$.

Next, we check the condition for $t, h$, and characterize the constants $m,h, F$. Recall that the conditions are the following: 
$$
\begin{aligned} 
&0 <  t \le \frac{(v - 1) m}{2}, \quad 1 - 2 c - 8 F c^2 \ge h, \quad c = \frac{t}{(v - 1)m}, \quad F = m(v-1)(\frac{1}{2} - c), \quad  0 < m \le \frac{1}{v - 1}, \\
 t & = \sqrt{\frac{F (v-1)^2 m^2 \log(\frac{1}{4 \kappa})}{16 n h}}, \quad h \in (0,1),  
\end{aligned} 
$$ 
where the first condition on $t$ follows from  Equation \eqref{eqn:jhgf}. 
Take first $h = F/8$. Then the first condition on $t$ implies that $n \ge \log(1/4 \kappa)$. The second condition on $h$ (with $h = F/8$) is satisfied if the first inequality holds 
$$
1 - F/8 \ge c(2 + 4F) \ge  c(2 + 8Fc) 
$$ 
since $c \in (0,1/2)$. Now, observe that $F \le 1/2$, hence it suffices to show that $$
c \le \frac{1 - 1/16}{4} \Rightarrow \sqrt{\frac{  \log(\frac{1}{4 \kappa})}{2 n}} \le \frac{15}{64} \Rightarrow n \ge \bar{C} \log(1/4 \kappa),  
$$ 
for a finite constant $\bar{C}$. The proof completes since the remaining conditions can be satisfied for an arbitrary choice of $0 < m < 1/(v-1)$.

We are left to show that the claim holds if $(v-1)/2$ is not an integer. For this case we follow the same steps of the proof where we construct a set of distributions $\mathcal{F}$ which puts mass $m$ on $v - 2$ $x_i, i < v-1$ and mass $\frac{1 - (v-2)m}{2}$ on the remaining $x_{v-1}, x_v$. We construct a bit vector $b \in \mathcal{B} \subset \{0,1\}^{v-2}$ with $\sum_{i=1}^{v-2} b_i = \frac{v-2}{2}$ which must be equal to an integer since $\frac{v-1}{2}$ is not. We construct (since $v \ge 3$)  
$$
S = \begin{cases}
&1 \text{ if } U \le p_1 - c + 2cb_i, X= x_i, i < v - 1 \\ 
&1 \text{ if } U \le p_1, X= x_i, i \in \{v-1, v\} \\ 
&0 \text{ otherwise} 
\end{cases},
$$ 
while the remaining part of the proof follows similarly to above.

\subsection{Regret bounds for $|D(\pi)|$, and $|C(\pi)|$} \label{sec:absolute}

To obtain UnFairness bounds for unfairness being defined as either $D(\pi)$ or $C(\pi)$ \textit{in absolute value} it suffices to bound the following empirical processes 
$$
\sup_{\pi \in \Pi}\Big| |\hat{C}(\pi)| - |C(\pi)|\Big|, \quad 
\sup_{\pi \in \Pi}\Big| |\hat{D}(\pi)| - |D(\pi)|\Big|.  
$$ 
We bound the first on the left-hand side while the second follows similarly. We write by the reverse triangular inequality
$$
\sup_{\pi \in \Pi}\Big| |\hat{C}(\pi)| - |C(\pi)|\Big| \le 
\sup_{\pi \in \Pi}\Big| \hat{C}(\pi) - C(\pi)\Big|.
$$ 
The rest of the proof follows similarly to Theorem \ref{thm:parity}.

\subsection{Proofs in Section \ref{sec:complexity} } \label{sec:proof_complexity}

\subsubsection{Proof of Proposition \ref{prop:polynomial}}  For simplicity we assume that $\beta_0 = \beta_1 = \beta \in [0,1]^p$, while our reasoning directly extend to different $\beta_0, \beta_1$. To analyze the computational complexity of the algorithm, we first, need to compute the computational complexity of each operation needed to estimated $\bar{W}_{j,n}$. Note that each optimization problem to estimate $\bar{W}_{j,n}$ is a linear program with $p$ variables and constraint $\beta^{(j)} \in [0,1], 1 \le j \le p$. Therefore, using standard arguments \citep[][Theorems 8.2, 8.5]{papadimitriou1998combinatorial}, each program admits an exact solution in $\mathcal{O}(p^{\omega})$ running time, for a finite constant $\omega$. There are $\sqrt{n}$ many of such programs, with overall running time $\mathcal{O}(\sqrt{n} p^{\omega})$. Consider now the optimization program in Equation \eqref{eqn:opt1}. Suppose first that $g(x) = x$. Then we can write the program as follow: for each constraint (i.e., each $\alpha_j$ with corresponding $\bar{W}_{j,n}$) in (B), we construct one program where (B) must hold for a single $\alpha_j$ only, and where we drop (C), (E) and replace (D) with $\beta^{(j)} \in [0,1], 1 \le j \le p$. We have in total $\sqrt{n}$ many of such programs. These programs are linear programs with $p + 1$ constraints and $p$ variables. Similarly to what discussed above, each of this program can be solved with running time $\mathcal{O}(p^{\omega})$ for some finite $\omega$. Once we solve each of this program, the solution to Equation \eqref{eqn:opt1} is obtained by finding the smallest objective among the $\sqrt{n}$ many programs. The running time for finding the minimum from $\sqrt{n}$ many elements is $\mathcal{O}(\sqrt{n})$. Therefore, the overall complexity of the optimization is $\mathcal{O}(\sqrt{n} p^{\omega})$. Consider now the case where $g(x) = |x|$. In such a case, for each sub-problem which substitute (B) in Equation \eqref{eqn:opt1} with a single constraint for a given $\alpha_j$, we can write two sub-problems. The first, is the optimization under the constraint that $x \ge 0$ and the second is the optimization under the constraint that $x \le 0$, with objective function multiplied by $-1$, where $x$ denotes the argument of the function $g(\cdot)$ (i.e., $\sum_{i=1}^n \hat{F}_i \pi(X_i)$). Again, each subproblem is a linear program with computational complexity $\mathcal{O}(p^{\omega})$ which completes the proof.

\subsection{Proof of Proposition \ref{prop:early_termination}}

Note first that  $\bar{W}_{j,n}^{\delta} \le \bar{W}_{j,n}$ for all $j, \delta$. Therefore, we obtain that the constraint in (B) imposed when estimating $\hat{\pi}_{\lambda}^{\delta}$ is less restrictive than the constraint when solving Equation \eqref{eqn:opt1}. It follows that 
$$
\mathcal{V}_n(\hat{\pi}_{\lambda}^{\delta}) - 
\mathcal{V}_n(\hat{\pi}_{\lambda}) \le \delta.  
$$ 
We can then write 
$$
\small 
\begin{aligned} 
\mathrm{UnFairness}(\hat{\pi}_{\lambda}^{\delta}) - \inf_{\pi \in \Pi_{\textrm{\mbox{\tiny o}}}} \mathrm{UnFairness}(\pi) & =  
\underbrace{\mathrm{UnFairness}(\hat{\pi}_{\lambda}^{\delta}) - \mathcal{V}_n(\hat{\pi}_{\lambda}^{\delta})}_{(A)} + \underbrace{\mathcal{V}_n(\hat{\pi}_{\lambda}^{\delta}) - \mathcal{V}_n(\hat{\pi}_{\lambda})}_{(B)}  \\ & \quad + \underbrace{\mathcal{V}_n(\hat{\pi}_{\lambda})  - \mathrm{UnFairness}(\hat{\pi}_{\lambda})}_{(C)} \\ & \quad + \underbrace{\mathrm{UnFairness}(\hat{\pi}_{\lambda}) - \inf_{\pi \in \Pi_{\textrm{\mbox{\tiny o}}}} \mathrm{UnFairness}(\pi)}_{(D)}.
\end{aligned} 
$$ 
We can now bound 
$$
\begin{aligned} 
(A) + (C) & \le 2 \sup_{\pi \in \Pi} \Big|\mathrm{UnFairness}(\pi) - \mathcal{V}_n(\pi)\Big|, \quad (B) \le \delta. 
\end{aligned} 
$$ 
The rest of the proof follows similarly to the one of Theorem \ref{thm:between_groups}.

\section{Extensions and Mathematical Details} \label{sec:a2}

\subsection{Comparison under strong duality} \label{sec:a22} 
We now sketch the differences in the optimization problem with the one in Equation \eqref{eqn:pi_c} assuming strong duality for expositional convenience, and providing an intuition on the result in Corollary \ref{cor:prop} for $\tilde{\pi}$. Assuming strong-duality, the optimization problem of maximizing welfare under fairness constraint in Equation \eqref{eqn:pi_c} can be equivalently re-written as:  
\begin{equation} \label{eqn:dual} 
\tilde{\pi} \in \arg \min_{\pi \in \Pi} \mathrm{UnFairness}(\pi), \quad \text{ such that } p_1 W_1(\pi) + (1 - p_1) W_0(\pi) \ge \lambda(\kappa)
 \end{equation} 
 for some constant $ \lambda(\kappa) \le \bar{W}_{p_1}$ which depends on $\kappa$. We now constrast Equation \eqref{eqn:dual} to our proposed approach (Equation \eqref{eqn:main_eq}). Suppose first that $ \lambda(\kappa) = \bar{W}_{p_1}$, i.e., $\widetilde{\pi}$ \textit{ is } Pareto optimal. Then the constraint in Equation \eqref{eqn:dual} is \textit{stricter} than the constraint in Equation \eqref{eqn:main_eq}, since the latter case imposes that $\alpha W_1(\pi) + (1 - \alpha) W_0(\pi) \ge \bar{W}_\alpha$, for \textit{some} $\alpha$, instead of for a particular chosen weight (e.g., $p_1$). As a result, $\pi^\star$ leads to a lower level of UnFairness whenever $\widetilde{\pi}$ \textit{is} Pareto optimal, since $\pi^\star$ minimizes UnFairness under weaker constraints compared to $\widetilde{\pi}$. When instead $\widetilde{\pi}$ is \textit{not} Pareto optimal, i.e., $ \lambda(\kappa) < \bar{W}_{p_1}$, $\widetilde{\pi}$ is Pareto dominated by some other allocation $\widetilde{\widetilde{\pi}}$. However $\widetilde{\widetilde{\pi}}$ leads to a larger UnFairness than $\pi^\star$, while not Pareto dominating $\pi^\star$.  

The key intuition is the following: under strong duality, the dual of $\widetilde{\pi}$ corresponds to minimize UnFairness for \textit{one particular} weighted combination of welfare exceeding a certain threshold. In contrast, our decision problem imposes the constraint that \textit{some} weighted combination of welfares exceeding a certain threshold. This difference reflects the difference between the lexicographic preferences that we propose as opposed to an additive social planner's utility. It guarantees that whenever $\widetilde{\pi}$ is Pareto optimal, its fairness is dominated from the one under $\pi^\star$.

\subsection{Cross-fitting with UnFairness}  \label{sec:a23} 

In this section we discuss cross-fitting with fairness. Two alternative cross-fitting procedures are available to the researcher. The first one, consists in dividing the sample into $K$ folds and estimating the conditional mean $\hat{m}_{d,s}^{(-k(i))}(X_i)$ using observations for which $S = s$ only, after excluding the fold $k$ corresponding to unit $i$ (panel on the right in Figure \ref{fig:cross_fitting}). Formally, 
 let $ i \in \mathcal{I}_k \cap \mathcal{S}_1$ where $\mathcal I_k$ is the $k$-th fold of the data and $\mathcal{S}_1 = \left\{ i: S_i =s_1\right\}$.
Let $\hat m^{(-k(i))}_{d,s_2}$ be an estimator obtained using samples not in the fold $k$, $ \mathcal{I}_k^c \cap \mathcal{S}_1^c$ for which $\mathcal{S}_1^c = \left\{ i: S_i =s_2\right\}$;  for example by a random forest or linear regression of $Y_j$ onto $X_j$ for  $S_j =s_2$, and $j \notin \mathcal{I}_k$.  Such an approach does not impose parametric restrictions on the dependence of $m_{d,s}$ on the attribute $s$, at the expense of shrinking the effective sample size used for estimation. The second approach consists in further imposing additional parametric restrictions on the depends of $m_{d,s}$ on $s$ and using all observations in all folds except $k$ for estimating $\hat{m}_{d,s}^{(-k(i))}(X_i)$ (panel on the right in Figure \ref{fig:cross_fitting}).

 \definecolor{glaucous}{rgb}{0.38,0.51,0.71}
 
 \begin{figure}
 \centering
 \caption{Graphical representation of cross-fitting under two alternative model formulation. The light gray area is the training set, used to construct an estimator of $\hat m_{d, s = 1}$, whereas the darker gray area is an evaluation set,    area in which a prediction of  $\hat m_{d, s = 1}$ is computed.
 }  \label{fig:cross_fitting}
  \begin{tikzpicture}[scale = 0.8]
  \draw[semithick,gray] (0,0) rectangle (4,3);
  \draw[red,thick,dashed] (2,-0.25) -- (2,3.25) ;
  \node at (3,3.25) {$S = 1$};
    \node at (1,3.25) {$S = 0$};
      \draw[blue,thick,dashed] (-0.25,1) -- (3.25,1) ;
         \draw[blue,thick,dashed] (-0.25,2) -- (3.25,2) ;
        \node at (-0.25,1.5) {$\mathcal{I}_k$};
         \filldraw[fill=gray!25] (2,2) rectangle (4,3);       
          \filldraw[fill=gray!25] (2,0) rectangle (4,1);  
                \filldraw[fill=gray!75] (0,1) rectangle (4,2);   
            \draw[red,thick,dashed] (2,-0.25) -- (2,3.25) ;   
                  \draw[blue,thick,dashed] (-0.25,1) -- (4.25,1) ;
         \draw[blue,thick,dashed] (-0.25,2) -- (4.25,2) ;
                           
  \draw[semithick,gray] (6,0) rectangle (10,3);
  \draw[red,thick,dashed] (8,-0.25) -- (8,3.25) ;
  \node at (9,3.25) {$S = 1$};
    \node at (7,3.25) {$S = 0$};
      \draw[blue,thick,dashed] (5.75,1) -- (9.25,1) ;
         \draw[blue,thick,dashed] (5.75,2) -- (9.25,2) ;
        \node at (5.75,1.5) {$\mathcal{I}_k$};
         \filldraw[fill=gray!25] (8,2) rectangle (10,3);       
          \filldraw[fill=gray!25] (8,0) rectangle (10,1); 
           \filldraw[fill=gray!25] (6,2) rectangle (10,3);       
          \filldraw[fill=gray!25] (6,0) rectangle (10,1); 
                \filldraw[fill=gray!75] (6,1) rectangle (10,2);   
            \draw[red,thick,dashed] (8,-0.25) -- (8,3.25) ;   
                  \draw[blue,thick,dashed] (5.75,1) -- (10.25,1) ;
         \draw[blue,thick,dashed] (5.75,2) -- (10.25,2) ;
                            \end{tikzpicture}
  \end{figure}

\subsection{Linear or quadratic constraints for the policy function space representation} \label{sec:a24}

In this section we discuss mixed integer formualtions of probabilistic and deterministic decisions rules. 
Consider first a deterministic decision rule of the form 
$$
\Pi = \Big\{\pi_\beta(X, S) = 1\{X^\top \beta  + S \beta_0 > 0\}, \quad \beta \in \mathcal{B} \Big\}. 
$$  
Then we can write the constraint (A) in Equation \eqref{eqn:opt1} as \citep{KitagawaTetenov_EMCA2018}
$$
\frac{X_i^\top \beta + s \beta_0 }{|C_i|} < z_{s,i} \le \frac{X_i^\top \beta + s \beta_0}{|C_i|} + 1, \quad C_i > \text{sup}_{\beta \in \mathcal{B}} |X_i^\top \beta| + |\beta_0|, \quad  z_{s,i} \in \{0,1\}. 
$$

Consider now the following probabilistic decision 
\begin{equation} \label{eqn:probabilistic} 
\Pi = \Big\{\pi_\beta(X, S) = p_1 1\{X_i^\top \beta + S \beta_0  > 0\} + p_0 1\{X^\top \beta + S \beta_0 \le 0\}, p_1, p_0 \in [0,1], \beta \in \mathcal{B} \Big\}. 
\end{equation} 
Then we can represent each decision variable as follows 
$$
\begin{aligned}
& z_{s,i} = p_1 \xi_{s,i} + p_0 (1 - \xi_{s,i}) \\  
& \frac{X_i^\top \beta + s \beta_0 }{|C_i|} < \xi_{s,i} \le \frac{X_i^\top \beta + s \beta_0}{|C_i|} + 1, \quad C_i > \text{sup}_{\beta \in \mathcal{B}} |X_i^\top \beta| + |\beta_0|, \quad  \xi_{s,i} \in \{0,1\}. 
\end{aligned} 
$$
where we introduced the additional variables $\xi_{s,i}$. We use this probablistic rule in the empirical application.

One last type of function class of interest is a linear probability rule of the following form
$$
z_{s,i} = X_i^\top \beta + \beta_0 S_i, \quad z_{s,i} \in [0,1]
$$ 
which leads to fast computations due lack of integer variables in the program.

\subsection{Extension: Additional Notions of UnFairness} \label{sec:a25}

\subsubsection{Predictive Parity} 

Predictive parity has been discussed in \cite{kasy2020} among others. Here we consider its definition within the context of policy-targeting. Its notion requires additional assumption for its implementation, assuming \textit{deterministic} treatment assignments $\pi(X_i) \in \{0,1\}$ (i.e., $\mathcal{T} = \{0,1\}$). The notion reads as follows: 
$$
P_s(\pi) = \Big|\mathbb{E}\Big[Y(1) \Big| \pi(X) = 1, S = s\Big] -
\mathbb{E}\Big[Y(1) \Big| \pi(X) = 1\Big]\Big| . 
$$  
Larger values of $P_s(\pi)$ increase UnFairness. 
Using the definition of the conditional expectation, and using consistency of potential outcomes, the following lemma holds. 
\begin{lem} \label{lem:parity} Let $\mathcal{T} = \{0,1\}$. Then following holds. 
$$
P_s(\pi) = (1 - p_s) \Big| \frac{\mathbb{E}\Big[Y(1) 1\{S_i = s\}\pi(X)\Big]}{p_s \mathbb{P}(\pi(X) = 1 | S = s)} - \frac{\mathbb{E}\Big[Y(1) \pi(X) 1\{S = s'\}\Big]}{(1 - p_s) \mathbb{P}(\pi(X) = 1| S = s')} \Big| . 
$$  
\end{lem}

\begin{proof}[Proof of Lemma \ref{lem:parity}] 
Using the definition of conditional expectation: 
\begin{equation} \label{eqn:kk} 
\begin{aligned} 
&\mathbb{E}\Big[Y \Big| \pi(X) = 1, S = s\Big] = \mathbb{E}\Big[ \frac{Y(1) 1\{S = s\} \pi(X)}{p_s P(\pi(X) = 1| S = s)}\Big]. 
\end{aligned} 
\end{equation}  
We also write 
$$
\begin{aligned} 
&\mathbb{E}\Big[Y \Big| \pi(X) = 1\Big] = p_s \mathbb{E}\Big[Y \Big| \pi(X) = 1, S = s\Big] + (1 - p_s)\mathbb{E}\Big[Y \Big| \pi(X) = 1, S = s'\Big]. 
\end{aligned} 
$$ 
Combining the expression with Equation \eqref{eqn:kk} completes the proof. 
\end{proof} 

Given two sensitive groups $\mathcal{S} = \{0,1\}$, the corresponding notion of UnFairness we consider takes the following form: 
\begin{equation} \label{eqn:unf1}
P(\pi) \propto \frac{P_1(\pi)}{1 - p_1} = \frac{P_{0}(\pi)}{p_1}. 
\end{equation}
We consider a double-robust estimator which takes the following form: 
\begin{equation} \label{eqn:v_hat1}
\small 
\begin{aligned} 
&\widehat{\mathcal{V}}_n(\pi) = \Big| \frac{\sum_{i=1}^n \pi(X_i) S_i \Big\{ \frac{ (Y_i - \hat{m}_1(X_i, S_i))D_i }{\hat{e}(X_i, S_i)} +  \hat{m}_1(X_i, S_i) \Big\} }{n p_1 \mathbb{P}(\pi(X_i) = 1|S_i = 1) } -  \frac{\sum_{i=1}^n (1 - S_i)  \pi(X_i) \Big\{ \frac{(Y_i - \hat{m}_1(X_i, S_i)) D_i }{\hat{e}(X_i, S_i)} + \hat{m}_1(X_i, S_i) \Big\} }{n (1 - p_1)\mathbb{P}(\pi(X_i) = 1|S_i = 0) } \Big|. 
\end{aligned} 
\end{equation}   
Observe that the estimator depends on the estimated conditional mean function and propensity score, whereas $p_s$ and $\mathbb{P}(\pi(X) = 1 | S = s)$ are assumed to be known. These two components can be obtained, for instance from census data, since $p_s$ and $\mathbb{P}(\pi(X) = 1 | S = s)$ only depend on the distribution of covariates and sensitive attributes. Whenever $\mathbb{P}(\pi(X_i) = 1| S_i = s)$ is replaced by its sampled analog $\mathbb{P}_n(\pi(X_i) = 1| S_i = s) = \frac{1}{n p_s} \sum_{i=1}^n \pi(X_i) 1\{S_i = s\}$, we require that $\mathbb{P}_n(\pi(X_i) = 1| S_i = s)$ is bounded away from zero almost surely.

\begin{thm}[Predictive parity] \label{thm:between_groups2} Let Assumptions \ref{ass:unconf}, \ref{ass:moment},\ref{ass:overlap}, \ref{ass:dr} hold. Let either $\mathrm{UnFairness}(\pi)$ be defined using the notion of Predictive (dis)-parity. Assume that $P(\pi(X,S) = 1 | S = 1), P(\pi(X,S) = 1 | S = 0) \in (\kappa, 1 - \kappa)$ for all $\pi \in \Pi$, $\kappa \in (0,1)$. Then for some constant $c_0 < \infty$, for any $\gamma \in (0,1), \lambda \ge \underline{b} \sqrt{\frac{v \log(2/\gamma)}{n}}$, for a constant $\underline{b} > 0$, independent of the sample size
with probability at least $1 - 2 \gamma$, 
$$
 \mathrm{UnFairness}(\hat{\pi}) - \inf_{\pi \in \Pi_{\textrm{\mbox{\tiny o}}}} \mathrm{UnFairness}(\pi) \le  c_0 \sqrt{\frac{  \log(2/\gamma)}{n}} + c_0 \sqrt{\frac{v}{n}},  
$$ 
for a finite constant $c_0 < \infty$. 
\end{thm}

 \begin{rem}[Mixed-integer linear representation of Predictive Parity] \label{thm:kj} Let $\mathbb{P}(\pi(X_i) | S_i = 1) = \frac{1}{p_s N} \sum_{i=1}^N \pi(X_i) S_i$ where $N$ denotes the number of individuals whose census-information (i.e., baseline covariates and sensitive attributes) are observed.
The optimization problem can be formulated as a mixed-integer \textit{fractional} linear program for $\pi(X)$ satisfying a linear representation. This follows after the linearization of the constraint (B), which can be achieved by introducing $2 N \times n$ many additional binary variables. Since fractional linear programs admit a mixed-integer linear program representations \citep{charnes1962programming}, the optimization problem can be solved as a mixed integer linear program. 
 \end{rem}

\begin{proof}[Proof of Theorem \ref{thm:between_groups2}]

We write 
$$
\begin{aligned} 
&\mathbb{E}\Big[\sup_{\pi \in \Pi} \Big|\frac{\widehat{P}_s(\pi)}{1 - p_s}  -  \frac{P_s(\pi)}{1 - p_s}\Big|\Big]\\ & \le \underbrace{\mathbb{E}\Big[\sup_{\pi \in \Pi} \Big| \frac{\sum_{i=1}^n \pi(X_i) S_i \Big\{ \frac{ (Y_i - \hat{m}_1(X_i, S_i))D_i }{\hat{e}(X_i, S_i)} +  \hat{m}_1(X_i, S_i) \Big\} }{n p_1 \mathbb{P}(\pi(X) = 1| S =1) } - \mathbb{E}\Big[Y | \pi(X) = 1, S = 1\Big] \Big|\Big]}_{(A)} \\
&+ \underbrace{\mathbb{E}\Big[\sup_{\pi \in \Pi} \Big| \frac{\sum_{i=1}^n \pi(X_i) (1 - S_i) \Big\{ \frac{ (Y_i - \hat{m}_1(X_i, S_i))D_i }{\hat{e}(X_i, S_i)} +  \hat{m}_1(X_i, S_i) \Big\} }{n \mathbb{P}(\pi(X) = 1| S = 0) (1 - p_1) } - \mathbb{E}\Big[Y | \pi(X) = 1, S = 0\Big] \Big|\Big]}_{(B)}. 
\end{aligned} 
$$  

We study $(A)$ while $(B)$ follows similarly. First, we write 
$$
\begin{aligned} 
(A) &\le  \underbrace{\mathbb{E}\Big[\sup_{\pi \in \Pi} \Big| \frac{\sum_{i=1}^n \pi(X_i) S_i \Big\{ \frac{ (Y_i - \hat{m}_1(X_i, S_i))D_i }{\hat{e}(X_i, S_i)} +  \hat{m}_1(X_i, S_i) - \frac{ (Y_i - m_1(X_i, S_i))D_i }{e(X_i, S_i)} -  m_1(X_i, S_i) \Big\} }{n p_1 \mathbb{P}(\pi(X) = 1| S =1) } \Big| \Big] }_{(I)} \\ 
& + \underbrace{\mathbb{E}\Big[\sup_{\pi \in \Pi} \Big| \frac{\sum_{i=1}^n \pi(X_i) S_i \Big\{ \frac{ (Y_i - m_1(X_i, S_i))D_i }{e(X_i, S_i)} +  m_1(X_i, S_i) \Big\} }{n p_1 \mathbb{P}(\pi(X) = 1| S =1) } - \mathbb{E}\Big[Y | \pi(X) = 1, S = 1\Big] \Big| \Big] }_{(II)}.
\end{aligned} 
$$ 
We study $(I)$ first. Define 
$$
\begin{aligned} 
V_n(\pi) &= \frac{1}{n p_1} \sum_{i=1}^n \pi(X_i) S_i \Big\{ \frac{ (Y_i - \hat{m}_1(X_i, S_i))D_i }{\hat{e}(X_i, S_i)} +  \hat{m}_1(X_i, S_i) - \frac{ (Y_i - m_1(X_i, S_i))D_i }{e(X_i, S_i)} -  m_1(X_i, S_i) \Big\}.
\end{aligned} 
$$ 

We have 
$$
(I) \le \frac{1}{\kappa} \mathbb{E}\Big[  \underbrace{\sup_{\pi \in \Pi} |V_n(\pi)|}_{(a)} \Big].
$$ 
We write
\begin{equation}
\begin{aligned} 
(a)
&\le \frac{1}{\delta} \underbrace{\sup_{\pi \in \Pi} \Big|\frac{1}{n} \sum_{i=1}^n S_i D_i (Y_i - m_{1,S_i}(X_i)) \Big(\frac{1}{ e(X_i, S_i)} - \frac{1}{ \hat{e}(X_i, S_i)}\Big) \pi(X_i, S_i)\Big|}_{(j)} \\ &+ \frac{1}{\delta} \underbrace{\sup_{\pi \in \Pi}\Big| \frac{1}{n} \sum_{i=1}^n 
\Big(\frac{D_i}{\hat{e}(X_i, S_i)}  - 1\Big) (m_{1,S_i}(X_i) - \hat{m}_{1,S_i}(X_i)) \pi(X_i, S_i) S_i \Big|}_{(jj)}. 
\end{aligned}  
\end{equation}

We study $(j)$ and $(jj)$ separately. We start from $(j)$. 
Recall, that by cross fitting $\hat{e}(X_i, S_i) = \hat{e}^{-k(i)}(X_i, S_i)$, where $k(i)$ is the fold containing unit $i$. Therefore, observe that given the $K$ folds for cross-fitting, we have  
\begin{equation} \label{eqn:summand} 
\begin{aligned} 
 &\Big|\frac{1}{n} \sum_{i=1}^n S_i D_i (Y_i - m_{1,S_i}(X_i)) \Big(\frac{1}{e(X_i, S_i)} - \frac{1}{ \hat{e}(X_i, S_i)}\Big) \pi(X_i, S_i)\Big| \\
 &\le  \sum_{k \in \{1, ..., K\}} \Big|\frac{1}{n} \sum_{i \in \mathcal{I}_k} S_i D_i (Y_i - m_{1,S_i}(X_i)) \Big(\frac{1}{ e(X_i, S_i)} - \frac{1}{ \hat{e}^{(-k(i))}(X_i, S_i)}\Big) \pi(X_i, S_i)\Big|. 
 \end{aligned} 
\end{equation} 
In addition, we have that 
\begin{equation}
\begin{aligned} 
&\mathbb{E}\Big[\sum_{i \in \mathcal{I}_k} S_i D_i (Y_i - m_{1,S_i}(X_i)) \Big(\frac{1}{e(X_i, S_i)} - \frac{1}{\hat{e}^{(-k(i))}(X_i, S_i)}\Big) \pi(X_i, S_i) \Big] \\ 
&= \mathbb{E}\Big[\mathbb{E}\Big[\sum_{i \in \mathcal{I}_k} S_i D_i (Y_i - m_{1,S_i}(X_i)) \Big(\frac{1}{ e(X_i, S_i)} - \frac{1}{ \hat{e}^{(-k(i))}(X_i, S_i)}\Big) \pi(X_i, S_i) \Big|  \hat{e}^{(-k(i))} \Big] \Big] = 0,  
\end{aligned} 
\end{equation} 
by cross-fitting. 
By Assumption \ref{ass:dr}, we know that 
\begin{equation}
\sup_{x \in \mathcal{X} ,s \in \mathcal{S}} \Big|\frac{1}{e(x, s)} - \frac{1}{ \hat{e}^{(-k(i))}(x, s)}\Big| \le 2/\delta^2
\end{equation} 
and therefore each summand in Equation \eqref{eqn:summand} is bounded by a finite constant $2/\delta^2$. 
We now obtain, using the symmetrization argument \citep{van1996weak}, and the Dudley's entropy integral \citep{wainwright2019high}
\begin{equation}
\mathbb{E}\Big[\sup_{\pi \in \Pi} |\frac{1}{n} \sum_{i \in \mathcal{I}_k} S_i D_i (Y_i - m_{1,S_i}(X_i)) \Big(\frac{1}{ e(X_i, S_i)} - \frac{1}{\hat{e}^{(-k(i))}(X_i, S_i)}\Big) \pi(X_i, S_i)| \Big| \hat{e}^{(-k(i))}\Big] \lesssim \frac{M}{\delta^2} \sqrt{v/n}.
\end{equation}

We now consider the term $(jj)$. Observe that we can write 
\begin{equation}
\begin{aligned} 
(jj) \le & \underbrace{\sup_{\pi \in \Pi}\Big| \frac{1}{n} \sum_{i=1}^n 
\Big(\frac{D_i }{\hat{e}(X_i, S_i)}  - \frac{D_i }{e(X_i, S_i)}\Big) (m_{1,S_i}(X_i) - \hat{m}_{1,S_i}(X_i)) S_i \pi(X_i, S_i)\Big|}_{(v)} \\ &+ \underbrace{\sup_{\pi \in \Pi}\Big| \frac{1}{n} \sum_{i=1}^n 
\Big(\frac{D_i }{ e(X_i, S_i)}  - 1\Big) (m_{1,S_i}(X_i) - \hat{m}_{1,S_i}(X_i)) \pi(X_i, S_i) S_i \Big|}_{(vv)}. 
\end{aligned} 
\end{equation} 
We consider each term seperately. Consider $(vv)$ first. Using the cross-fitting argument we obtain 
\begin{equation} 
\begin{aligned} 
&\sup_{\pi \in \Pi}\Big| \frac{1}{n} \sum_{i=1}^n 
\Big(\frac{D_i }{e(X_i, S_i)}  - 1\Big) (m_{1,S_i}(X_i) - \hat{m}_{1,S_i}(X_i)) \pi(X_i, S_i) S_i \Big| \\ &\le \sum_{k \in \{1, ..., K\}} \sup_{\pi \in \Pi}\Big| \frac{1}{n} \sum_{i \in \mathcal{I}_k} 
\Big(\frac{D_i }{p_s e(X_i, S_i)}  - 1\Big) (m_{1,S_i}(X_i) - \hat{m}_{1,S_i}^{(-k(i))}(X_i)) \pi(X_i, S_i) S_i\Big|. 
\end{aligned} 
\end{equation} 
Observe now that 
\begin{equation}
\mathbb{E}\Big[\Big(\frac{D_i}{e(X_i, S_i)}  - 1\Big) (m_{1,S_i}(X_i) - \hat{m}_{1,S_i}^{(-k(i))}(X_i)) \pi(X_i, S_i) S_i\Big| \hat{m}_{1,S_i}^{(-k(i))} \Big] = 0.  
\end{equation} 
Therefore, following the same argument discussed before, 
\begin{equation}
\mathbb{E}\Big[\sup_{\pi \in \Pi}\Big| \frac{1}{n} \sum_{i \in \mathcal{I}_k} 
\Big(\frac{D_i }{e(X_i, S_i)}  - 1\Big) (m_{1,S_i}(X_i) - \hat{m}_{1,S_i}^{(-k(i))}(X_i)) \pi(X_i, S_i) S_i \Big|\Big] \lesssim  \frac{M}{\delta^2} \sqrt{\frac{v}{n}} . 
\end{equation} 

We are now left to bound $(v)$. We obtain that 
\begin{equation}
(v) \le \sqrt{\frac{1}{n} \sum_{i=1}^n \Big(\frac{1}{ \hat{e}(X_i,S_i)} - \frac{1}{ e(X_i,S_i)}\Big)^2} \sqrt{\frac{1}{n} \sum_{i=1}^n (m_{1,S_i}(X_i) - \hat{m}_{1,S_i}(X_i))^2}. 
\end{equation}  
Using Jensen inequality and Assumption \ref{ass:dr} $\mathbb{E}[(v)] \lesssim n^{-1/2}$.

We now move to bound the expectation of $(II)$. First, observe that by Lemma \ref{lem:parity}, and standard properties of the double-robust estimator, we have that 
$$
\mathbb{E}\Big[\frac{\frac{1}{p_1 n} \sum_{i=1}^n \pi(X_i) S_i \Big\{ \frac{ (Y_i - m_1(X_i, S_i))D_i }{e(X_i, S_i)} +  m_1(X_i, S_i) \Big\} }{\mathbb{P}(\pi(X) = 1| S = s) }\Big] = \mathbb{E}\Big[Y(1) | \pi(X) = 1, S = 1\Big]. 
$$ 

Using the symmetrization argument (see \cite{van1996weak}), we have 
$$
(II) \le 2 \mathbb{E}\Big[\sup_{\pi \in \Pi} \Big| \frac{\frac{1}{p_1 n} \sum_{i=1}^n \pi(X_i) \sigma_i S_i \Big\{ \frac{ (Y_i - m_1(X_i, S_i))D_i }{e(X_i, S_i)} +  m_1(X_i, S_i) \Big\} }{\mathbb{P}(\pi(X) = 1| S = s) } \Big| \Big], 
$$ 
where $\{\sigma_i\}$ are $i.i.d.$ exogenous Radamacher random variables. Using the assumption that $P(\pi(X) = 1|S = s) \in (\kappa, 1 - \kappa)$, we write 
$$
\begin{aligned} 
&\mathbb{E}\Big[\sup_{\pi \in \Pi} \Big| \frac{\frac{1}{p_1 n} \sum_{i=1}^n \pi(X_i) \sigma_i S_i \Big\{ \frac{ (Y_i - m_1(X_i, S_i))D_i }{e(X_i, S_i)} +  m_1(X_i, S_i) \Big\} }{\mathbb{P}(\pi(X) = 1| S = s) } \Big| \Big] \\ &\le \frac{1}{\kappa} \mathbb{E}\Big[\sup_{\pi \in \Pi} \Big| \frac{1}{p_1 n} \sum_{i=1}^n \pi(X_i) \sigma_i S_i \Big\{ \frac{ (Y_i - m_1(X_i, S_i))D_i }{e(X_i, S_i)} +  m_1(X_i, S_i) \Big\}  \Big| \Big]. 
\end{aligned} 
$$  
We now proceed using a standard argument. Using the fact that each summand in the above expression are uniformly bounded, and $\Pi$ has finite VC-dimension, using the Dudley's entropy integral bound, it directly follows that 
$$
\mathbb{E}\Big[\sup_{\pi \in \Pi} \Big| \frac{1}{p_1 n} \sum_{i=1}^n \pi(X_i) \sigma_i S_i \Big\{ \frac{ (Y_i - m_1(X_i, S_i))D_i }{e(X_i, S_i)} +  m_1(X_i, S_i) \Big\}  \Big| \Big] \lesssim \sqrt{\frac{v}{n}}
$$  
which concludes the proof. 

\end{proof} 

\subsubsection{Counterfactual envy-freeness} \label{sec:counterfactual_app}

In this section we discuss estimation and guarantees for the counterfactual notion of fairness in Section \ref{sec:counterfactual}. 

We estimate $\mathcal{A}(\cdot)$ as: 
\begin{equation} \label{eqn:opt2}
\small 
\begin{aligned} 
\mathcal{A}_{n}(s,s'; \pi) = &\frac{1}{n \hat{p}_s} \sum_{i: S_i =s }  \Big\{\hat{m}_{1,s'}(X_i) \pi(X_i, s) +  \hat{m}_{0,s'}(X_i) (1 - \pi(X_i, s)) \Big\} \\ &- \frac{1}{n} \sum_{i=1}^n \Big\{ \hat{\Gamma}_{1, s,i}\pi(X_i,s) - \hat{\Gamma}_{0,s,i}(1 - \pi(X_i,s)) \Big\}. 
\end{aligned} 
\end{equation}  Whenever we aim not to discriminate in either direction, we take the sum of the effects $\mathcal{A}(s_1,s_2; \pi)$ and $\mathcal{A}(s_2,s_1; \pi)$,\footnote{Such an approach builds on the notion of ``social envy'' discussed in \cite{feldman1974fairness}.} 
and define counterfactual envy-freeness and its estimator as  
\begin{equation} \label{eqn:envy_free} 
\small 
\begin{aligned} 
 \mathrm{E}(\pi) = \mathcal{A}(1,0;\pi) + \mathcal{A}(0,1;\pi), \quad \hat{\mathrm{E}}(\pi) = \mathcal{A}_{n}(1,0; \pi) + \mathcal{A}_{n}(0,1; \pi). 
\end{aligned}  
 \end{equation}  

 \begin{ass} \label{ass:consistency} Assume that for some $\zeta > 0$, 
$
\mathbb{E}\Big[ \Big( \hat{m}_{d,s_1}(X_i(s_2)) - m_{d,s_1}(X_i(s_2))\Big)^2 \Big] = \mathcal{O}(n^{-2\zeta}), \forall s_1,s_2 \in \{0,1\}, d \in \{0,1\}. 
$
\end{ass} 
Assumption \ref{ass:consistency} states that the estimator of the conditional mean function for each sensitive attribute and treatment status $s, d \in \{0,1\}$, must converge to the true conditional mean function in mean-squared error at some arbitrary rate $2\zeta > 0$. Here, we require convergence in $l_2$ for a given sensitive attribute conditional on the \textit{opposite} sensitive attribute, due to the particular notion of fairness considered.\footnote{Namely, to estimate fairness, we need to extrapolate relative to the \textit{opposite} group.} 
 Examples include (i) linear regression models, of the form $m_{d,s}(x) = x\beta_s$, with $\beta_s$ being potentially high-dimensional, and bounded covariates; (ii) local polynomial estimators \citep{fan1996local}.

\begin{thm} \label{thm:parity}  Let Assumptions \ref{ass:moment}-\ref{ass:dr}, \ref{ass:unconfounded2} and \ref{ass:consistency} hold. Let $\mathrm{UnFairness}(\cdot) = \mathrm{E}(\cdot)$ and $\hat{\mathcal{V}}_n(\cdot) = \hat{E}(\cdot)$. Then for some constants $0 < \underline{b}, c_0 < \infty$ independent of the sample size, for any $\gamma \in (0,1),  \lambda \ge \underline{b} (\sqrt{v} + \sqrt{\log(2/\gamma)} + 1), N = \sqrt{n}$, 
with probability at least $1 - 2 \gamma$,  
$$
\mathrm{UnFairness}(\hat{\pi}) - \inf_{\pi \in \Pi_{\textrm{\mbox{\tiny o}}}} \mathrm{UnFairness}(\pi) \le c_0 \sqrt{\frac{v }{ n^{2\zeta}} } + c_0 \sqrt{\frac{\log(2/\gamma)}{n}}.   
$$ 
\end{thm}

The proof is in Appendix \ref{sec:proof_envy_freeness}. 
A corollary of Theorem \ref{thm:parity} is that under the parametric rate of convergence of the conditional mean function, the regret bound scales at rate $n^{-1/2}$. Interestingly, the convergence rate is of order slower than $n^{-1/2}$ for non-parametric estimators compared to the notions of UnFairness discussed in Section \ref{sec:4}. The slower convergence rate is because counterfactual envy-freeness requires estimating the conditional mean function on the population with attribute $S = s_1$ while averaging over the covariates' distribution with the opposite attribute, therefore requiring extrapolation.  
This result showcases the \textit{trade-off} in the choice of a counterfactual notion of unfairness relative to predictive ones.

\subsubsection{Envy-freeness UnFairness: Proofs} \label{sec:proof_envy_freeness}

\begin{lem} \label{lem:helper1} Under Assumption\ref{ass:moment}, \ref{ass:overlap}, \ref{ass:dr}, \ref{ass:unconfounded2}, \ref{ass:consistency}, the following holds:  with probability at least $1 - \gamma$ ,  
\begin{equation}
\sup_{\pi \in \Pi} \Big| \mathcal{A}(s,s';\pi) - \mathcal{A}_{n}(s,s';\pi) \Big| \le \frac{c M}{\delta^2} \sqrt{\frac{\log(2/\gamma)}{n}}  + \frac{c}{\delta} n^{-\eta} + \sqrt{\frac{v}{n}}
\end{equation} 
for a universal constant $c < \infty$. 
\end{lem} 
\vspace{5 mm}
\begin{proof}[Proof of Lemma \ref{lem:helper1}] We consider the case where $s' \neq s$, whereas $s' = s$ follows trivially. 
Observe that we can write 
\begin{equation}
\small 
\begin{aligned} 
&\sup_{\pi \in \Pi} \Big| \mathcal{A}(s,s';\pi) - \mathcal{A}_{n}(s,s';\pi) \Big| \le \\ &\underbrace{\sup_{\pi \in \Pi} \Big|\mathbb{E}_{X(s)}\Big[V_{\pi(X(s),s)}(X(s), s')\Big] - \frac{1}{n} \sum_{i=1}^n \Big(\frac{1\{S_i = s\}}{\hat{p}_s} \hat{m}_{1,s'}(X_i) \pi(X_i,s) + \frac{1\{S_i = s\}}{\hat{p}_s}  \hat{m}_{0,s'}(X_i)(1 -  \pi(X_i,s)) \Big) \Big|}_{(A)} \\ 
&+  \underbrace{\sup_{\pi \in \Pi} |\hat{W}_{s'}(\pi) - W_{s'}(\pi)|}_{(B)}. 
\end{aligned} 
\end{equation} 
The term (B) is bounded as discussed in Lemma \ref{lem:dr}. Therefore, we are only left to discuss bounds on (A). To derive bounds in such a scenario, we first observe that we can write 
$$
\small 
\begin{aligned} 
&\sup_{\pi \in \Pi} \Big|\mathbb{E}_{X(s)}\Big[V_{\pi(X(s),s)}(X(s), s')\Big] - \frac{1}{n} \sum_{i=1}^n \Big(\frac{1\{S_i = s\}}{\hat{p}_s} \hat{m}_{1,s'}(X_i) \pi(X_i,s) + \frac{1\{S_i = s\}}{\hat{p}_s}  \hat{m}_{0,s'}(X_i)(1 -  \pi(X_i,s)) \Big) \Big| \\&\le \underbrace{\sup_{\pi \in \Pi} \Big|\frac{1}{n} \sum_{i=1}^n \frac{1\{S_i = s\}}{p_s} m_{1,s'}(X_i) \pi(X_i,s) - \mathbb{E}\Big[\frac{1\{S_i = s\}}{p_s} m_{1,s'}(X_i) \pi(X_i,s)\Big] \Big|}_{(I)}  \\&+\underbrace{\sup_{\pi \in \Pi} \Big|\frac{1}{n} \sum_{i=1}^n \frac{1\{S_i = s\}}{p_s} m_{0,s'}(X_i) (1 - \pi(X_i,s)) - \mathbb{E}\Big[\frac{1\{S_i = s\}}{p_s} m_{0,s'}(X_i) (1 - \pi(X_i,s))\Big] \Big|}_{(II)} \\&+\underbrace{\sup_{\pi \in \Pi} \Big|\frac{1}{n} \sum_{i=1}^n \Big(\frac{1\{S_i = s\}}{\hat{p}_s} \hat{m}_{1,s'}(X_i) - \frac{1\{S_i = s\}}{p_s} m_{1,s'}(X_i)\Big) \pi(X_i,s)\Big|}_{(III)} \\ &+ \underbrace{\sup_{\pi \in \Pi} \Big|\frac{1}{n} \sum_{i=1}^n \Big(\frac{1\{S_i = s\}}{\hat{p}_s} \hat{m}_{0,s'}(X_i) - \frac{1\{S_i = s\}}{p_s} m_{0,s'}(X_i)\Big) (1 - \pi(X_i,s))\Big|}_{(IV)}.
\end{aligned} 
$$
We discuss (I) and (III), whereas (II) and (IV) follow similarly. Observe first that by Assumption \ref{ass:moment} and the bounded difference inequality, with probability $1  -\gamma$, 
$$
\begin{aligned} 
&\sup_{\pi \in \Pi} \Big|\frac{1}{n} \sum_{i=1}^n \frac{1\{S_i = s\}}{p_s} m_{1,s'}(X_i) \pi(X_i,s) - \mathbb{E}\Big[\frac{1\{S_i = s\}}{p_s} m_{1,s'}(X_i) \pi(X_i,s)\Big] \Big|\\ 
&\le \mathbb{E}\Big[\sup_{\pi \in \Pi} \Big|\frac{1}{n} \sum_{i=1}^n \frac{1\{S_i = s\}}{p_s} m_{1,s'}(X_i) \pi(X_i,s) - \mathbb{E}\Big[\frac{1\{S_i = s\}}{p_s} m_{1,s'}(X_i) \pi(X_i,s)\Big] \Big|\Big] + \bar{C} \frac{M}{\delta} \sqrt{\log(2/\gamma)/n}
\end{aligned} 
$$
for a constant $\bar{C} < \infty$. Under Assumption  \ref{ass:unconfounded2} each summand is centered around zero. Using the symmetrization argument \citep{van1996weak}, we have 
$$
\begin{aligned} 
&\mathbb{E}\Big[\sup_{\pi \in \Pi} \Big|\frac{1}{n} \sum_{i=1}^n \frac{1\{S_i = s\}}{p_s} m_{1,s'}(X_i) \pi(X_i,s) - \mathbb{E}\Big[\frac{1\{S_i = s\}}{p_s} m_{1,s'}(X_i) \pi(X_i,s)\Big] \Big|\Big] \le \\ 
&2\mathbb{E}\Big[\sup_{\pi \in \Pi} \Big|\frac{1}{n} \sum_{i=1}^n \sigma_i \frac{1\{S_i = s\}}{p_s} m_{1,s'}(X_i) \pi(X_i,s)\Big|\Big]
\end{aligned} 
$$
where $\sigma_i$ are $i.i.d.$ Radamacher random variables. Since $m_{1,s}$ is uniformly bounded and similarly $p_s$ is bounded, and by Assumption \ref{ass:moment}, we obtain by the properties of the Dudley's entropy integral \citep{wainwright2019high}, 
$$
\mathbb{E}\Big[\sup_{\pi \in \Pi} \Big|\frac{1}{n} \sum_{i=1}^n \sigma_i \frac{1\{S_i = s\}}{p_s} m_{1,s'}(X_i) \pi(X_i,s)\Big|\Big] \le \bar{C}\frac{M}{\delta} \sqrt{v/n}
$$
for a universal constant $\bar{C} < \infty$. We now move to bound (III). 
Using the triangular inequality and Holder's inequality, we obtain 
\begin{equation}
\begin{aligned}  
(III) \le \frac{1}{n} \sum_{i=1}^n \frac{1\{S_i = s\}}{p_s} \Big| m_{1,s'}(X_i) - \hat{m}_{1,s'}(X_i)\Big| 
\end{aligned} 
\end{equation} 
The above bound is deterministic and it does not depend on $\pi$. 
Observe now that by consistency of potential outcomes and covariates
\begin{equation} \label{eqn:khg}
\begin{aligned} 
 &\frac{1}{n} \sum_{i=1}^n \frac{1\{S_i = s\}}{p_s} \Big| m_{1,s'}(X_i) - \hat{m}_{1,s'}(X_i)\Big|   \\ =&\frac{1}{n} \sum_{i=1}^n \frac{1\{S_i = s\}}{p_s} \Big| m_{1,s'}(X_i(s)) - \hat{m}_{1,s'}(X_i(s))\Big| \le \frac{1}{n \delta} \sum_{i=1}^n \Big| m_{1,s'}(X_i(s)) - \hat{m}_{1,s'}(X_i(s))\Big|. 
 \end{aligned} 
\end{equation} 
We now separate the contribution of each of the $K$ folds using in the cross-fitting algorithm. Namely, we define  
\begin{equation} 
\frac{1}{n \delta} \sum_{i=1}^n \Big| m_{1,s'}(X_i(s)) - \hat{m}_{1,s'}(X_i(s))\Big|\le \sum_{k \in \{1, ..., K\}} \frac{1}{n \delta} \sum_{i \in \mathcal{I}_k} \Big| m_{1,s'}(X_i(s)) - \hat{m}_{1,s'}^{(-k(i))}(X_i(s))\Big| 
\end{equation} 
where $\mathcal{I}_k$ denotes the set of indexes in fold $k$, and $\hat{m}_{1,s'}^{(-k(i))}$ denotes the estimator obtained from all folds except $k$. Next, we bound the following term using Liaponuv inequality: 
\begin{equation}
\begin{aligned}  
&\frac{1}{n} \sum_{i \in \mathcal{I}_k}  \mathbb{E}\Big[|m_{1,s'}(X_i(s)) - \hat{m}_{1,s'}(X_i(s))| \Big] \lesssim \sqrt{\mathbb{E}\Big[|m_{1,s'}(X_i(s)) - \hat{m}_{1,s'}(X_i(s))|^2 \Big]} \le c n^{-\eta}.
\end{aligned}   
\end{equation} 
The last inequality follows by Assumption \ref{ass:consistency}, for a universal constant $c < \infty$. Finally, we discuss exponential concentration of the empirical counterpart. By boundeness of $\hat{m}$ in Assumption \ref{ass:dr}, we have  
\begin{equation} 
\sup_{x \in \mathcal{X}} \big| m_{d,s'}(x) - \hat{m}_{d,s'}(x) \Big| \le 2M. 
\end{equation}
By the bounded difference inequality, with probability at least $1 - \gamma$, 
\begin{equation}
\frac{1}{n} \sum_{i \in \mathcal{I}_k} \Big|m_{1,s'}(X_i(s)) - \hat{m}_{1,s'}(X_i(s)) \Big| \le \mathbb{E}\Big[\Big|m_{1,s'}(X_i(s)) - \hat{m}_{1,s'}(X_i(s)) \Big|\Big] + 4 M \sqrt{\log(2/\gamma)/n}.   
\end{equation}  
Combining the above bounds, the proof completes. 
\end{proof} 

\begin{cor} Theorem \ref{thm:parity} holds.  
\end{cor} 

\begin{proof} This follows from Theorem \ref{thm:parity} and Lemma \ref{lem:helper2}.  
\end{proof} 

\subsection{Multi-action policies} \label{sec:multi_actions}

In this subsection, we discuss how our results extend to multi-action policies. Using Theorem 1 in \cite{zhou2018offline}, assuming a bounded entropy integral with respect to $\Pi$ (Assumption 3 in \cite{zhou2018offline}), and assumptions in Theorem \ref{thm:1b}, it is easy to show that with probability at least $1 - \gamma$ 
$$
\Big|W_d(\pi) - \hat{W}_d(\pi)\Big| = \mathcal{O}\Big(\frac{\kappa(\Pi)}{\sqrt{n}}\Big) + o(1/\sqrt{n}) + \mathcal{O}\Big(\sqrt{\frac{\log(2/\gamma)}{n}}\Big)
$$ 
where $\kappa(\Pi) = \int_0^1 \sqrt{\log(\mathcal{N}(\Pi, \epsilon^2))} d\epsilon$ and $\mathcal{N}(\Pi, \epsilon)$ is the covering number for the function class $\Pi$.\footnote{The reader may refer to \cite{wainwright2019high} and Definition 4 in \cite{zhou2018offline} for more discussion.} Here the first two terms follow directly from the bound on the Rademacher complexity in \cite{zhou2018offline} and standard symmetrization arguments \citep{van1996weak}, while the last term follows from the bounded difference inequality \citep{wainwright2019high} and leverages bounded estimated conditional mean and propensity score.\footnote{Bounded estimated nuisance functions is not required if we interpret our results as asymptotic in the spirit of \cite{athey2017efficient}.}  Given such concentration result, it is easy to show that the rest of our proofs of Theorems \ref{lem:cons_dr}, \ref{thm:1b}, \ref{thm:2}, do not require binary actions, and follow without additional modifications for a finite number of actions.

 \section{Numerical Studies and Empirical Application: Further Results and Details} \label{sec:a3}

 \subsection{Empirical Application} \label{sec:a3a}

 \paragraph{Estimation details}  We control for confounding of the treatment assignment by estimating the probability of treatment using a penalized logistic regression, where we condition on the non-Caucasian attribute, gender, the average score, years to graduation, whether the individual had previously had entrepreneurship activities, the startup region (which a dummy since only two regions are considered), the degree (either engineer or business) and the school rank. We estimate the outcome using a penalized logistic regression, after conditioning on the above covariates, and any interaction term between gender, treatment assignment, and a vector of covariates, which include years to graduation, prior entrepreneurship, startup region, and the school rank. We estimate treatment effects using a doubly robust estimator. We use cross-fitting with five folds in our estimation. 
 
 \paragraph{Additional results for probabilistic treatment assignments} We also consider in our analysis the class of \textit{probabilistic} assignment rules, which assign treatments with a probability decision as in Equation \eqref{eqn:probabilistic}.   Results are collected in Figure \ref{fig:front_prob}, where we observe that the set of probabilistic decision Pareto dominates the determinitic ones up-to a small optimization error.

\begin{figure}
\centering 
\includegraphics[scale=0.4]{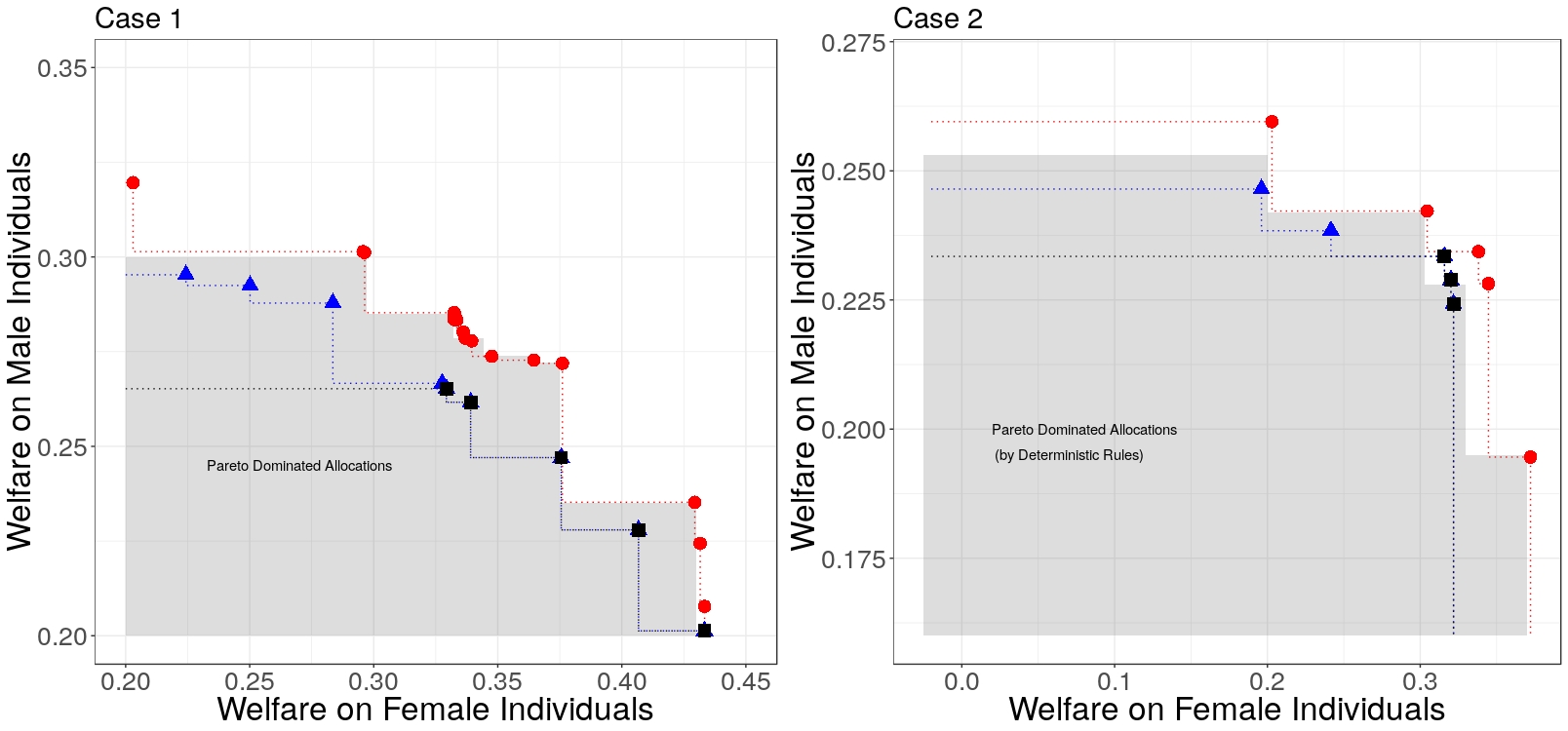} 
\caption{Empirical application. (Discretized) Pareto frontier under \textit{probabilistic} linear policy rule estimated through MIQP. Dots denote Pareto optimal allocations. Red dots (circle) correspond to $\Pi_1$, blue dots (triangle) to $\Pi_2$ and black dots (square) to $\Pi_3$. The gray area denotes the set of allocations dominated by a \textit{deterministic} decision rule. 
} \label{fig:front_prob}
\end{figure} 
 
 \subsection{Numerical Studies} \label{sec:a3b}

In this subsection, we include additional results for the numerical studies. In Figure \ref{fig:stopping_time} we report the computational time for different number of covariates. The figure shows that the linear probability rule and optimal tree present better scalability than the maximum score, while the four methods can still be feasibly implemented for $n = 600$. Figure \ref{fig:p_3} presents results for the different function classes for $p =3$ (instead of $p = 4$) covariates. 
Figure \ref{fig:impact} reports welfare comparisons for the disparate impact method for female and male participants. Table \ref{tab:mse2} reports comparisons for $n = 400$. In the table, we observe that a smaller sample size tends to decrease the performance of each method, as expected. We still observe that the proposed method leads to the largest fairness across all designs. While for $n = 600$, the proposed method is never Pareto dominated, here also the method is never Pareto dominated with a single exception occurring for the maximum score, where we observe a slight dominance for welfare but not fairness which might occur with small sample sizes. In all remaining cases, the proposed method leads to strictly larger welfare for the female students with respect to all competitors.

   \begin{figure}
\centering 
\includegraphics[scale=0.6]{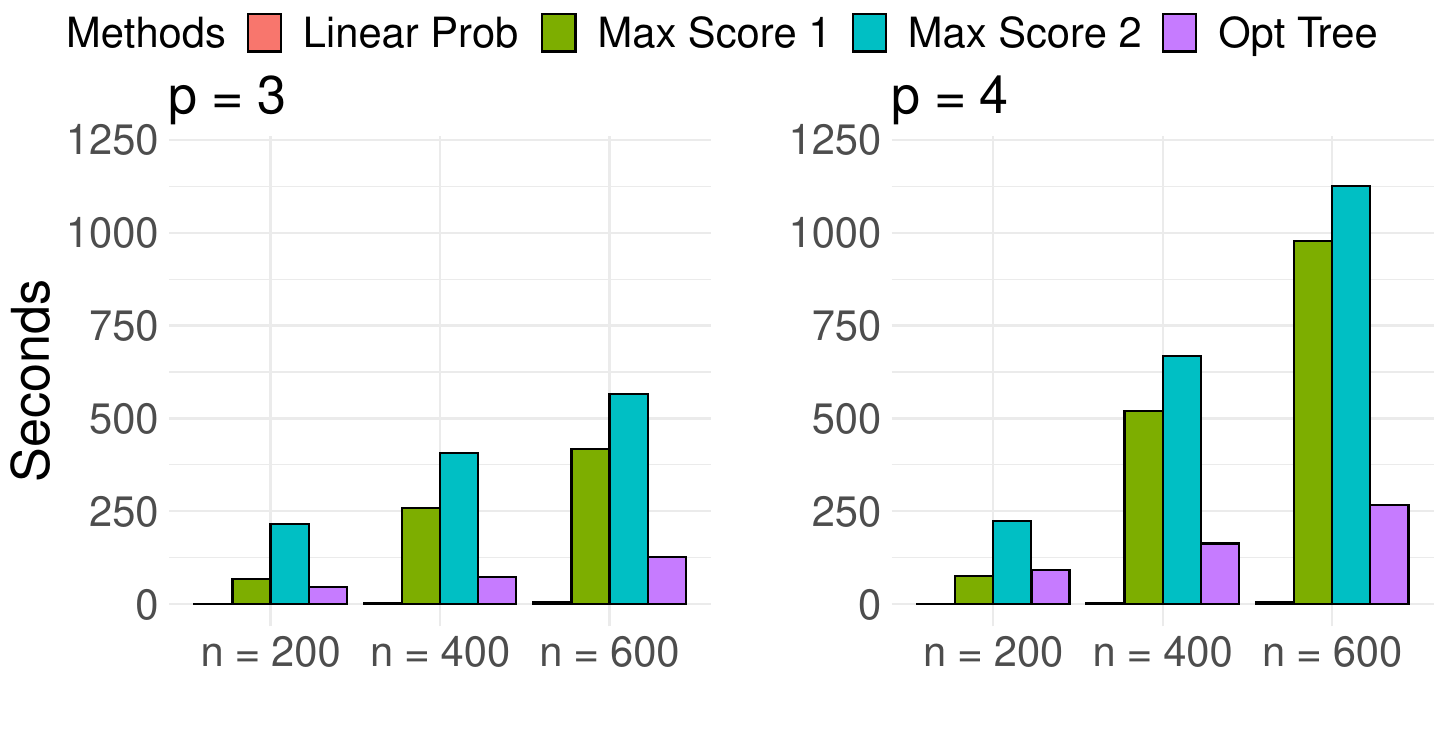}
\caption{Computational time in seconds for different number of covariates $p$ and sample sizes $n$. Here, Linear Prob is a linear probability rule estimated via linear programming, maximum score is estimated with mixed-integer linear program and optimal tree via exhaustive search. The maximum score algorithm presents two different stopping times (Max Score 1 and Max Score 2).} \label{fig:stopping_time}
\end{figure}

   \begin{figure}
\centering 
\includegraphics[scale=0.6]{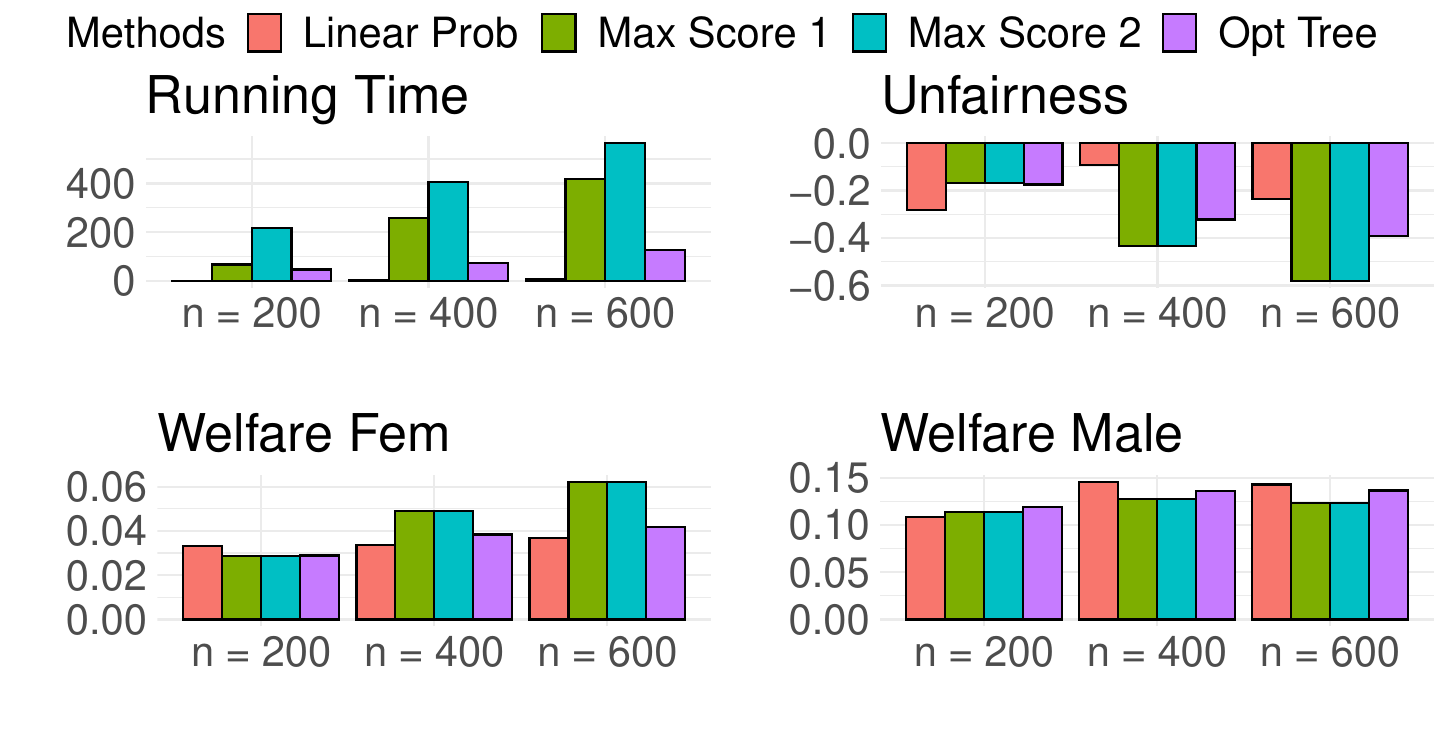}
\caption{$p = 3$. Running time, unfairness and welfare  as a function of the sample size and covariates. Here, Linear Prob is a linear probability rule estimated via linear programming, maximum score is estimated with mixed-integer linear program and optimal tree via exhaustive search. The maximum score algorithm presents two different stopping times (Max Score 1 and Max Score 2).} \label{fig:p_3}
\end{figure}

\begin{figure}
\centering
\includegraphics[scale=0.5]{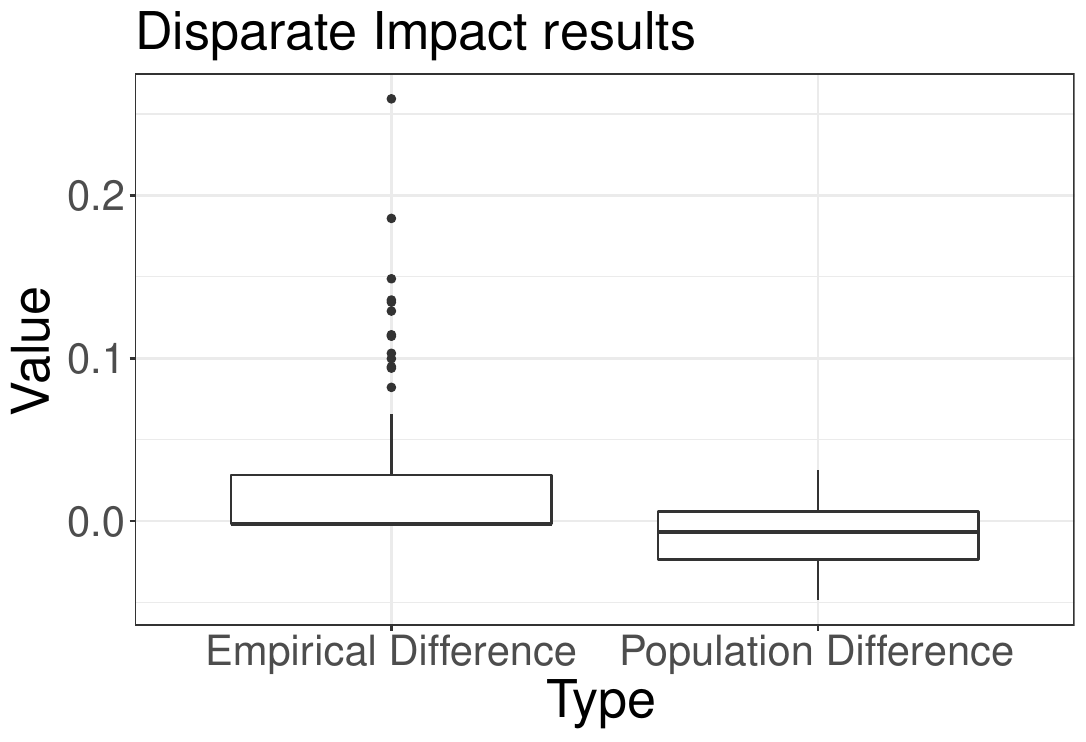}
\caption{Difference between females and males' welfare for the method Disparate impact in Table \ref{tab:mse2}, with $\kappa = 1$. The left-hand side panel shows the empirical difference (showing that the constraint is attained) and the right-hand side panel the population counterpart.} 
\label{fig:impact} 
\end{figure}

\begin{table*}[!ht]\centering
\caption{Statistical disparity (Unfairness), welfare of male ($W_0$) and female ($W_1$) participants of the proposed method (Fair Targeting) and of the alternative procedures in \textit{percentage points}. Weighted average maximizes a weighted average of females and males' welfare with weight $\alpha = 1/2$; Utilitarian average uses instead $\alpha = \mathbb{E}_n[S]$; Constrained Max maximizes welfare under fairness constrain and Disparate Impact maximizes welfare under constraints on disparate welfare impact between the two groups. $n = 400, p = 4$. The constraint is $\kappa = 10$ for the methods in the fourth and fifth row and $\kappa = 1$ for the last two rows. }    \label{tab:mse3} 
\begin{tabular}{@{}lrrrcrrrcrrr@{}}\toprule
& \multicolumn{3}{c}{Linear Rule} & \ & \multicolumn{3}{c}{Maximum Score} & \ & \multicolumn{3}{c}{Tree} \\
\cmidrule{2-4} \cmidrule{6-8}  \cmidrule{10-12}  
&\small UnFair&\small $W_0$ &\small $W_1$ && \small UnFair&\small  $W_0$  & \small $W_1$ && \small UnFair & \small $W_0$  & \small $W_1$  \\
 \Xhline{.8pt} 
 \cline{2-4}  \cline{6-8}  \cline{10-12} 
 \rowcolor{glaucous!60}  \small Fair Targeting & $-19.7$ & $12.2$ & $6.1$&& $-46.4$ & $12.7$ & $5.3$&&$-32.7$&12.7&6.0\\
\small Weighted Average  & $20.7$ & $15$ & $4.4$ && $7.5$ & $15.3$ & $4.3$&& $-1.9$ & $14.5$ & $5.5$  \\
\small Utilitarian Average  & $7.5$ & $15.3$ & $4.3$ && $7.5$ & $15.3$ & $4.3$&& $7.5$ & $15.3$ & $4.3$  \\
\small Constrained Max & $-8.5$ & $14.6$ & $4.6$&&$-18.7$ & $14$ & $5$ &&$-16.8$ & $14.5$ & $5$\\
\small Disparate Impact & $6.3$ & $12.8$ & $4.6$&&$-0.2$ & $14.4$ & $5.5$ &&$2.9$ & $14.1$ & $5$\\
\small Constrained Max2 & $-12.3$ & $14.2$ & $4.7$&&$-21.7$ & $13.7$ & $4.9$ &&$-19.3$ & $14.5$ & $5$\\
\small Disparate Impact2 & $-1.6$ & $11.2$ & $4.8$&&$-6.2$ & $14.2$ & $5.4$ &&$-2.2$ & $13.9$ & $5.2$\\
\bottomrule
\end{tabular}
\end{table*}

 \bibliography{my_bib2}

\begin{thebibliography}{}

\bibitem[\protect\citeauthoryear{Armstrong and Shen}{Armstrong and
  Shen}{2015}]{armstrong2015inference}
Armstrong, T. and S.~Shen (2015).
\newblock Inference on optimal treatment assignments.
\newblock {\em Available at SSRN 2592479\/}.

\bibitem[\protect\citeauthoryear{Athey and Wager}{Athey and
  Wager}{2021}]{athey2017efficient}
Athey, S. and S.~Wager (2021).
\newblock Policy learning with observational data.
\newblock {\em Econometrica\/}~{\em 89\/}(1), 133--161.

\bibitem[\protect\citeauthoryear{Balashankar, Lees, Welty, and
  Subramanian}{Balashankar et~al.}{2019}]{balashankar2019fair}
Balashankar, A., A.~Lees, C.~Welty, and L.~Subramanian (2019).
\newblock What is fair? exploring pareto-efficiency for fairness constrained
  classifiers.
\newblock {\em arXiv preprint arXiv:1910.14120\/}.

\bibitem[\protect\citeauthoryear{Bhattacharya and Dupas}{Bhattacharya and
  Dupas}{2012}]{bhattacharya2012inferring}
Bhattacharya, D. and P.~Dupas (2012).
\newblock Inferring welfare maximizing treatment assignment under budget
  constraints.
\newblock {\em Journal of Econometrics\/}~{\em 167\/}(1), 168--196.

\bibitem[\protect\citeauthoryear{Boucheron, Bousquet, and Lugosi}{Boucheron
  et~al.}{2005}]{boucheron2005theory}
Boucheron, S., O.~Bousquet, and G.~Lugosi (2005).
\newblock Theory of classification: A survey of some recent advances.
\newblock {\em ESAIM: probability and statistics\/}~{\em 9}, 323--375.

\bibitem[\protect\citeauthoryear{Boucheron, Lugosi, and Massart}{Boucheron
  et~al.}{2013}]{boucheron2013concentration}
Boucheron, S., G.~Lugosi, and P.~Massart (2013).
\newblock {\em Concentration inequalities: A nonasymptotic theory of
  independence}.
\newblock Oxford university press.

\bibitem[\protect\citeauthoryear{Boucheron, Lugosi, Massart, et~al.}{Boucheron
  et~al.}{2003}]{boucheron2003concentration}
Boucheron, S., G.~Lugosi, P.~Massart, et~al. (2003).
\newblock Concentration inequalities using the entropy method.
\newblock {\em The Annals of Probability\/}~{\em 31\/}(3), 1583--1614.

\bibitem[\protect\citeauthoryear{Charnes and Cooper}{Charnes and
  Cooper}{1962}]{charnes1962programming}
Charnes, A. and W.~W. Cooper (1962).
\newblock Programming with linear fractional functionals.
\newblock {\em Naval Research logistics quarterly\/}~{\em 9\/}(3-4), 181--186.

\bibitem[\protect\citeauthoryear{Chernozhukov, Newey, and Robins}{Chernozhukov
  et~al.}{2018}]{chernozhukov2018double}
Chernozhukov, V., W.~K. Newey, and J.~Robins (2018).
\newblock Double/de-biased machine learning using regularized riesz
  representers.
\newblock Technical report, cemmap working paper.

\bibitem[\protect\citeauthoryear{Chouldechova}{Chouldechova}{2017}]{chouldechova2017fair}
Chouldechova, A. (2017).
\newblock Fair prediction with disparate impact: A study of bias in recidivism
  prediction instruments.
\newblock {\em Big data\/}~{\em 5\/}(2), 153--163.

\bibitem[\protect\citeauthoryear{Corbett-Davies and Goel}{Corbett-Davies and
  Goel}{2018}]{corbett2018measure}
Corbett-Davies, S. and S.~Goel (2018).
\newblock The measure and mismeasure of fairness: A critical review of fair
  machine learning.
\newblock {\em arXiv preprint arXiv:1808.00023\/}.

\bibitem[\protect\citeauthoryear{Coston, Mishler, Kennedy, and
  Chouldechova}{Coston et~al.}{2020}]{coston2020counterfactual}
Coston, A., A.~Mishler, E.~H. Kennedy, and A.~Chouldechova (2020).
\newblock Counterfactual risk assessments, evaluation, and fairness.
\newblock In {\em Proceedings of the 2020 Conference on Fairness,
  Accountability, and Transparency}, pp.\  582--593.

\bibitem[\protect\citeauthoryear{Cowgill and Tucker}{Cowgill and
  Tucker}{2019}]{cowgill2019economics}
Cowgill, B. and C.~E. Tucker (2019).
\newblock Economics, fairness and algorithmic bias.
\newblock {\em preparation for: Journal of Economic Perspectives\/}.

\bibitem[\protect\citeauthoryear{Devroye, Gy{\"o}rfi, and Lugosi}{Devroye
  et~al.}{2013}]{devroye2013probabilistic}
Devroye, L., L.~Gy{\"o}rfi, and G.~Lugosi (2013).
\newblock {\em A probabilistic theory of pattern recognition}, Volume~31.
\newblock Springer Science \& Business Media.

\bibitem[\protect\citeauthoryear{Donini, Oneto, Ben-David, Shawe-Taylor, and
  Pontil}{Donini et~al.}{2018}]{donini2018empirical}
Donini, M., L.~Oneto, S.~Ben-David, J.~S. Shawe-Taylor, and M.~Pontil (2018).
\newblock Empirical risk minimization under fairness constraints.
\newblock In {\em Advances in Neural Information Processing Systems}, pp.\
  2791--2801.

\bibitem[\protect\citeauthoryear{Dwork, Hardt, Pitassi, Reingold, and
  Zemel}{Dwork et~al.}{2012}]{dwork2012fairness}
Dwork, C., M.~Hardt, T.~Pitassi, O.~Reingold, and R.~Zemel (2012).
\newblock Fairness through awareness.
\newblock In {\em Proceedings of the 3rd innovations in theoretical computer
  science conference}, pp.\  214--226.

\bibitem[\protect\citeauthoryear{Elliott and Lieli}{Elliott and
  Lieli}{2013}]{elliott2013predicting}
Elliott, G. and R.~P. Lieli (2013).
\newblock Predicting binary outcomes.
\newblock {\em Journal of Econometrics\/}~{\em 174\/}(1), 15--26.

\bibitem[\protect\citeauthoryear{Fan and Gijbels}{Fan and
  Gijbels}{1996}]{fan1996local}
Fan, J. and I.~Gijbels (1996).
\newblock {\em Local polynomial modelling and its applications: monographs on
  statistics and applied probability 66}, Volume~66.
\newblock CRC Press.

\bibitem[\protect\citeauthoryear{Farrell}{Farrell}{2015}]{farrell2015robust}
Farrell, M.~H. (2015).
\newblock Robust inference on average treatment effects with possibly more
  covariates than observations.
\newblock {\em Journal of Econometrics\/}~{\em 189\/}(1), 1--23.

\bibitem[\protect\citeauthoryear{Feldman and Kirman}{Feldman and
  Kirman}{1974}]{feldman1974fairness}
Feldman, A. and A.~Kirman (1974).
\newblock Fairness and envy.
\newblock {\em The American Economic Review\/}, 995--1005.

\bibitem[\protect\citeauthoryear{Finkelstein, Taubman, Wright, Bernstein,
  Gruber, Newhouse, Allen, Baicker, and Group}{Finkelstein
  et~al.}{2012}]{finkelstein2012oregon}
Finkelstein, A., S.~Taubman, B.~Wright, M.~Bernstein, J.~Gruber, J.~P.
  Newhouse, H.~Allen, K.~Baicker, and O.~H.~S. Group (2012).
\newblock The oregon health insurance experiment: evidence from the first year.
\newblock {\em The Quarterly journal of economics\/}~{\em 127\/}(3),
  1057--1106.

\bibitem[\protect\citeauthoryear{Florios and Skouras}{Florios and
  Skouras}{2008}]{florios2008exact}
Florios, K. and S.~Skouras (2008).
\newblock Exact computation of max weighted score estimators.
\newblock {\em Journal of Econometrics\/}~{\em 146\/}(1), 86--91.

\bibitem[\protect\citeauthoryear{Hardt, Price, and Srebro}{Hardt
  et~al.}{2016}]{hardt2016equality}
Hardt, M., E.~Price, and N.~Srebro (2016).
\newblock Equality of opportunity in supervised learning.
\newblock In {\em Advances in neural information processing systems}, pp.\
  3315--3323.

\bibitem[\protect\citeauthoryear{Hirano and Porter}{Hirano and
  Porter}{2009}]{hirano2009asymptotics}
Hirano, K. and J.~R. Porter (2009).
\newblock Asymptotics for statistical treatment rules.
\newblock {\em Econometrica\/}~{\em 77\/}(5), 1683--1701.

\bibitem[\protect\citeauthoryear{Karmarkar}{Karmarkar}{1984}]{karmarkar1984new}
Karmarkar, N. (1984).
\newblock A new polynomial-time algorithm for linear programming.
\newblock In {\em Proceedings of the sixteenth annual ACM symposium on Theory
  of computing}, pp.\  302--311.

\bibitem[\protect\citeauthoryear{Kasy and Abebe}{Kasy and
  Abebe}{2020}]{kasy2020}
Kasy, M. and R.~Abebe (2020).
\newblock Fairness, equality, and power in algorithmic decision making.
\newblock Technical report.

\bibitem[\protect\citeauthoryear{Kilbertus, Carulla, Parascandolo, Hardt,
  Janzing, and Sch{\"o}lkopf}{Kilbertus et~al.}{2017}]{kilbertus2017avoiding}
Kilbertus, N., M.~R. Carulla, G.~Parascandolo, M.~Hardt, D.~Janzing, and
  B.~Sch{\"o}lkopf (2017).
\newblock Avoiding discrimination through causal reasoning.
\newblock In {\em Advances in Neural Information Processing Systems}, pp.\
  656--666.

\bibitem[\protect\citeauthoryear{Kitagawa and Tetenov}{Kitagawa and
  Tetenov}{2018}]{KitagawaTetenov_EMCA2018}
Kitagawa, T. and A.~Tetenov (2018).
\newblock Who should be treated? {E}mpirical welfare maximization methods for
  treatment choice.
\newblock {\em Econometrica\/}~{\em 86\/}(2), 591--616.

\bibitem[\protect\citeauthoryear{Kitagawa and Tetenov}{Kitagawa and
  Tetenov}{2019}]{kitagawa2017equality}
Kitagawa, T. and A.~Tetenov (2019).
\newblock Equality-minded treatment choice.
\newblock {\em Journal of Business \& Economic Statistics\/}, 1--14.

\bibitem[\protect\citeauthoryear{Kleinberg, Ludwig, Mullainathan, and
  Rambachan}{Kleinberg et~al.}{2018}]{kleinberg2018algorithmic}
Kleinberg, J., J.~Ludwig, S.~Mullainathan, and A.~Rambachan (2018).
\newblock Algorithmic fairness.
\newblock In {\em Aea papers and proceedings}, Volume 108, pp.\  22--27.

\bibitem[\protect\citeauthoryear{Kosorok}{Kosorok}{2008}]{kosorok2008introduction}
Kosorok, M.~R. (2008).
\newblock {\em Introduction to empirical processes and semiparametric
  inference.}
\newblock Springer.

\bibitem[\protect\citeauthoryear{Kusner, Russell, Loftus, and Silva}{Kusner
  et~al.}{2019}]{kusner2019making}
Kusner, M., C.~Russell, J.~Loftus, and R.~Silva (2019).
\newblock Making decisions that reduce discriminatory impacts.
\newblock In {\em International Conference on Machine Learning}, pp.\
  3591--3600.

\bibitem[\protect\citeauthoryear{Liu, Radanovic, Dimitrakakis, Mandal, and
  Parkes}{Liu et~al.}{2017}]{liu2017calibrated}
Liu, Y., G.~Radanovic, C.~Dimitrakakis, D.~Mandal, and D.~C. Parkes (2017).
\newblock Calibrated fairness in bandits.
\newblock {\em arXiv preprint arXiv:1707.01875\/}.

\bibitem[\protect\citeauthoryear{Lyons and Zhang}{Lyons and
  Zhang}{2017}]{lyons2017impact}
Lyons, E. and L.~Zhang (2017).
\newblock The impact of entrepreneurship programs on minorities.
\newblock {\em American Economic Review\/}~{\em 107\/}(5), 303--07.

\bibitem[\protect\citeauthoryear{Lyons and Zhang}{Lyons and
  Zhang}{2018}]{lyons2018does}
Lyons, E. and L.~Zhang (2018).
\newblock Who does (not) benefit from entrepreneurship programs?
\newblock {\em Strategic Management Journal\/}~{\em 39\/}(1), 85--112.

\bibitem[\protect\citeauthoryear{Manski}{Manski}{2004}]{manski2004}
Manski (2004).
\newblock Statistical treatment rules for heterogeneous populations.
\newblock {\em Econometrica\/}~{\em 72\/}(4), 1221--1246.

\bibitem[\protect\citeauthoryear{Manski}{Manski}{1975}]{manski1975maximum}
Manski, C.~F. (1975).
\newblock Maximum score estimation of the stochastic utility model of choice.
\newblock {\em Journal of econometrics\/}~{\em 3\/}(3), 205--228.

\bibitem[\protect\citeauthoryear{Manski and Thompson}{Manski and
  Thompson}{1989}]{manski1989estimation}
Manski, C.~F. and T.~S. Thompson (1989).
\newblock Estimation of best predictors of binary response.
\newblock {\em Journal of Econometrics\/}~{\em 40\/}(1), 97--123.

\bibitem[\protect\citeauthoryear{Martinez, Bertran, and Sapiro}{Martinez
  et~al.}{2019}]{martinez2019fairness}
Martinez, N., M.~Bertran, and G.~Sapiro (2019).
\newblock Fairness with minimal harm: A pareto-optimal approach for healthcare.
\newblock {\em arXiv preprint arXiv:1911.06935\/}.

\bibitem[\protect\citeauthoryear{Mas-Colell, Whinston, Green,
  et~al.}{Mas-Colell et~al.}{1995}]{mas1995microeconomic}
Mas-Colell, A., M.~D. Whinston, J.~R. Green, et~al. (1995).
\newblock {\em Microeconomic theory}, Volume~1.
\newblock Oxford university press New York.

\bibitem[\protect\citeauthoryear{Mbakop and Tabord-Meehan}{Mbakop and
  Tabord-Meehan}{2021}]{mbakop2016model}
Mbakop, E. and M.~Tabord-Meehan (2021).
\newblock Model selection for treatment choice: Penalized welfare maximization.
\newblock {\em Econometrica\/}~{\em 89\/}(2), 825--848.

\bibitem[\protect\citeauthoryear{Murphy}{Murphy}{2003}]{murphy2003optimal}
Murphy, S.~A. (2003).
\newblock Optimal dynamic treatment regimes.
\newblock {\em Journal of the Royal Statistical Society: Series B (Statistical
  Methodology)\/}~{\em 65\/}(2), 331--355.

\bibitem[\protect\citeauthoryear{Nabi, Malinsky, and Shpitser}{Nabi
  et~al.}{2019}]{nabi2019learning}
Nabi, R., D.~Malinsky, and I.~Shpitser (2019).
\newblock Learning optimal fair policies.
\newblock {\em Proceedings of machine learning research\/}~{\em 97}, 4674.

\bibitem[\protect\citeauthoryear{Narita}{Narita}{2021}]{narita2021incorporating}
Narita, Y. (2021).
\newblock Incorporating ethics and welfare into randomized experiments.
\newblock {\em Proceedings of the National Academy of Sciences\/}~{\em
  118\/}(1).

\bibitem[\protect\citeauthoryear{Negishi}{Negishi}{1960}]{negishi1960welfare}
Negishi, T. (1960).
\newblock Welfare economics and existence of an equilibrium for a competitive
  economy.
\newblock {\em Metroeconomica\/}~{\em 12\/}(2-3), 92--97.

\bibitem[\protect\citeauthoryear{Newey}{Newey}{1990}]{newey1990semiparametric}
Newey, W.~K. (1990).
\newblock Semiparametric efficiency bounds.
\newblock {\em Journal of applied econometrics\/}~{\em 5\/}(2), 99--135.

\bibitem[\protect\citeauthoryear{Papadimitriou and Steiglitz}{Papadimitriou and
  Steiglitz}{1998}]{papadimitriou1998combinatorial}
Papadimitriou, C.~H. and K.~Steiglitz (1998).
\newblock {\em Combinatorial optimization: algorithms and complexity}.
\newblock Courier Corporation.

\bibitem[\protect\citeauthoryear{Rai}{Rai}{2018}]{rai2018statistical}
Rai, Y. (2018).
\newblock Statistical inference for treatment assignment policies.
\newblock {\em Unpublished Manuscript\/}.

\bibitem[\protect\citeauthoryear{Rambachan, Kleinberg, Ludwig, and
  Mullainathan}{Rambachan et~al.}{2020}]{rambachan2020economic}
Rambachan, A., J.~Kleinberg, J.~Ludwig, and S.~Mullainathan (2020).
\newblock An economic approach to regulating algorithms.

\bibitem[\protect\citeauthoryear{Robins and Rotnitzky}{Robins and
  Rotnitzky}{1995}]{robins1995semiparametric}
Robins, J.~M. and A.~Rotnitzky (1995).
\newblock Semiparametric efficiency in multivariate regression models with
  missing data.
\newblock {\em Journal of the American Statistical Association\/}~{\em
  90\/}(429), 122--129.

\bibitem[\protect\citeauthoryear{Rotblat}{Rotblat}{1999}]{rotblat1999hippocratic}
Rotblat, J. (1999).
\newblock A hippocratic oath for scientists.
\newblock {\em Science\/}~{\em 286\/}(5444), 1475--1475.

\bibitem[\protect\citeauthoryear{Rubin}{Rubin}{1990}]{rubin1990formal}
Rubin, D.~B. (1990).
\newblock Formal mode of statistical inference for causal effects.
\newblock {\em Journal of statistical planning and inference\/}~{\em 25\/}(3),
  279--292.

\bibitem[\protect\citeauthoryear{Stoye}{Stoye}{2012}]{stoye2012minimax}
Stoye, J. (2012).
\newblock Minimax regret treatment choice with covariates or with limited
  validity of experiments.
\newblock {\em Journal of Econometrics\/}~{\em 166\/}(1), 138--156.

\bibitem[\protect\citeauthoryear{Sun}{Sun}{2020}]{sun2020empirical}
Sun, L. (2020).
\newblock Empirical welfare maximization with constraints.
\newblock Technical report, Working paper.

\bibitem[\protect\citeauthoryear{Tetenov}{Tetenov}{2012}]{tetenov2012statistical}
Tetenov, A. (2012).
\newblock Statistical treatment choice based on asymmetric minimax regret
  criteria.
\newblock {\em Journal of Econometrics\/}~{\em 166\/}(1), 157--165.

\bibitem[\protect\citeauthoryear{Ustun, Liu, and Parkes}{Ustun
  et~al.}{2019}]{ustun2019fairness}
Ustun, B., Y.~Liu, and D.~Parkes (2019).
\newblock Fairness without harm: Decoupled classifiers with preference
  guarantees.
\newblock In {\em International Conference on Machine Learning}, pp.\
  6373--6382.

\bibitem[\protect\citeauthoryear{Vaidya}{Vaidya}{1990}]{vaidya1990algorithm}
Vaidya, P.~M. (1990).
\newblock An algorithm for linear programming which requires o (((m+ n) n 2+(m+
  n) 1.5 n) l) arithmetic operations.
\newblock {\em Mathematical Programming\/}~{\em 47\/}(1), 175--201.

\bibitem[\protect\citeauthoryear{Van Der~Vaart and Wellner}{Van Der~Vaart and
  Wellner}{1996}]{van1996weak}
Van Der~Vaart, A.~W. and J.~A. Wellner (1996).
\newblock Weak convergence.
\newblock In {\em Weak convergence and empirical processes}, pp.\  16--28.
  Springer.

\bibitem[\protect\citeauthoryear{Varian}{Varian}{1976}]{varian1976two}
Varian, H.~R. (1976).
\newblock Two problems in the theory of fairness.
\newblock {\em Journal of Public Economics\/}~{\em 5\/}(3-4), 249--260.

\bibitem[\protect\citeauthoryear{Viviano}{Viviano}{2019}]{viviano2019policy}
Viviano, D. (2019).
\newblock Policy targeting under network interference.
\newblock {\em arXiv preprint arXiv:1906.10258\/}.

\bibitem[\protect\citeauthoryear{Wainwright}{Wainwright}{2019}]{wainwright2019high}
Wainwright, M.~J. (2019).
\newblock {\em High-dimensional statistics: A non-asymptotic viewpoint},
  Volume~48.
\newblock Cambridge University Press.

\bibitem[\protect\citeauthoryear{Wolsey and Nemhauser}{Wolsey and
  Nemhauser}{1999}]{wolsey1999integer}
Wolsey, L.~A. and G.~L. Nemhauser (1999).
\newblock {\em Integer and combinatorial optimization}, Volume~55.
\newblock John Wiley \& Sons.

\bibitem[\protect\citeauthoryear{Xiao, Min, Yongfeng, Zhaoquan, Yiqun, and
  Shaoping}{Xiao et~al.}{2017}]{xiao2017fairness}
Xiao, L., Z.~Min, Z.~Yongfeng, G.~Zhaoquan, L.~Yiqun, and M.~Shaoping (2017).
\newblock Fairness-aware group recommendation with pareto-efficiency.
\newblock In {\em Proceedings of the Eleventh ACM Conference on Recommender
  Systems}, pp.\  107--115.

\bibitem[\protect\citeauthoryear{Zhou, Athey, and Wager}{Zhou
  et~al.}{2018}]{zhou2018offline}
Zhou, Z., S.~Athey, and S.~Wager (2018).
\newblock Offline multi-action policy learning: Generalization and
  optimization.
\newblock {\em arXiv preprint arXiv:1810.04778\/}.

\end{thebibliography}
\bibliographystyle{chicago}

\end{document}